\documentclass[10pt, twocolumn, twoside]{IEEEtran}
\usepackage[ruled,vlined]{algorithm2e}
\usepackage{cite}
\usepackage{graphicx}
\usepackage{psfrag}
\usepackage{subfigure}
\usepackage{url}
\usepackage{amsmath}
\usepackage{array}
\usepackage{amssymb}
\usepackage{amsfonts}
\usepackage{graphicx}
\usepackage{epstopdf}

\newtheorem{lemma}{Lemma}

\newtheorem{definition}{Definition}

\newtheorem{example}{Example}

\usepackage{textcomp}
\usepackage{multirow}
\usepackage{threeparttable}
\input{epsf.sty}

\title{Wireless Network-Coded Accumulate-Compute and Forward Two-Way Relaying}

\begin{document}

\author{
\authorblockN{Srishti Shukla \textsuperscript{$\dagger$}, Vijayvaradharaj T Muralidharan\textsuperscript{$\#$}  and B. Sundar Rajan\textsuperscript{$\dagger$}}\\
\authorblockA{Email: {$\lbrace$srishti, tmvijay, bsrajan$\rbrace$} @ece.iisc.ernet.in\\
\textsuperscript{$\dagger$}IISc Mathematics Initiative (IMI), Dept. of Mathematics and Dept. of Electrical Comm. Engg., IISc, Bangalore\\
\textsuperscript{$\#$} Dept. of Electrical Comm. Engg., IISc, Bangalore
}
}

\maketitle
\begin{abstract}
The design of modulation schemes for the physical layer network-coded two way wireless relaying scenario is considered. It was observed by Koike-Akino et al. \cite{KoPoTa} for the two way relaying scenario, that adaptively changing the network coding map used at the relay according to the channel conditions greatly reduces the impact of multiple access interference which occurs at the relay during the MA Phase and all these network coding maps should satisfy a requirement called \textit{exclusive law}. We extend this approach to an Accumulate-Compute and Forward protocol which employs two phases: Multiple Access (MA) phase consisting of two channel uses with independent messages in each channel use, and Broadcast (BC) phase having one channel use. Assuming that the two users transmit points from the same 4-PSK constellation, every such network coding map that satisfies the exclusive law can be represented by a Latin Square with side 16, and conversely, this relationship can be used to get the network coding maps satisfying the exclusive law. Two methods of obtaining this network coding map to be used at the relay are discussed. Using the structural properties of the Latin Squares for a given set of parameters, the problem of finding all the required maps is reduced to finding a small set of maps. Having obtained all the Latin Squares, the set of all possible channel realizations is quantized, depending on which one of the Latin Squares obtained optimizes the performance. The quantization thus obtained, is shown to be the same as the one obtained in \cite{MNR} for the 2-stage bidirectional relaying.   
\end{abstract}

\section{Background}

The concept of physical layer network coding has attracted a lot of attention in recent times. The idea of physical layer network coding for the two way relay channel was first introduced in \cite{ZhLiLa}, where the multiple access interference occurring at the relay was exploited so that the communication between the end nodes can be done using a two stage protocol. Information theoretic studies for the physical layer network coding scenario were reported in \cite{KiMiTa}, \cite{PoYo}. The design principles governing the choice of modulation schemes to be used at the nodes for uncoded transmission were studied in \cite{KoPoTa}. An extension for the case when the nodes use convolutional codes was done in \cite{KoPoTa_conv}. A multi-level coding scheme for the two-way relaying was proposed in \cite{HeN}. \\

We consider the two-way wireless relaying scenario shown in Fig. 1, where two-way data transfer takes place among the nodes A and B with the help of the relay R. It is assumed that the two nodes operate in half-duplex mode, i.e., they cannot transmit and receive at the same time in the same frequency band. The relaying protocol consists of two phases, \textit{multiple access} (MA) phase, consisting of two channel uses during which A and B transmit to R twice, two independent messages in the two channel uses, with points from 4-PSK constellation, and \textit{broadcast} (BC) phase, in which R transmits to A and B. The relay node R accumulates the information sent by the user nodes in the first and second channel use of the MA phase, and transmits in the BC phase a message that contains information about all the four messages received by it in the MA phase. Network Coding is employed at R in such a way that A(/B) can decode the two messages transmitted by B(/A), given that A(/B) knows its own messages. We call this strategy accumulate-compute and forward (ACF) protocol.\\

It was observed in \cite{KoPoTa} and \cite{MNR} for 4-PSK, that for uncoded transmission, the network coding map used at the relay needs to be changed adaptively according to the channel fade coefficient, in order to minimize the impact of multiple access interference. In other words, the set of all possible channel realizations is quantized into a finite number of regions, with a specific network coding map giving the best performance in a particular region. It is shown in \cite{NMR} for any choice of signal sets of equal cardinality used at the two users, that every such network coding map that satisfies the \textit{exclusive law} is representable as a Latin Square and conversely, this relationship can be used to get the network coding maps satisfying the exclusive law. \\

\begin{definition}
A Latin Square of order $M$ is an $M \times M$ array in which each cell contains a symbol from a set of t different symbols such that each symbol occurs at most once in each row and column \cite{Rod}.\\
\end{definition}

\begin{figure}[tp]
\center
\includegraphics[height=15mm]{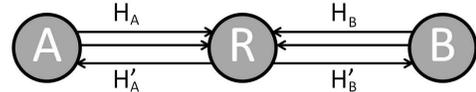}
\centering
{\caption{A two-way ACF relay channel}}
\end{figure}

  Similar to the ACF protocol, a store-and-forward protocol has been earlier studied in \cite{LXT}, for the two-way relaying channel. In \cite{LXT}, the authors derive an upper bound on the ergodic sum-capacity for the two-way relaying scenario when delay tends to infinity, and propose two alternative awaiting and broadcast (AAB) schemes which approach the new upper bound at high SNR. Using numerical results, they show that the proposed AAB schemes significantly outperforms the traditional physical layer network coding methods without delay in terms of ergodic maximum sum rates. However, modulation and physical layer network coding have not been addressed in \cite{LXT}.\\

\begin{figure}[tp]
\center
\includegraphics[height=40mm]{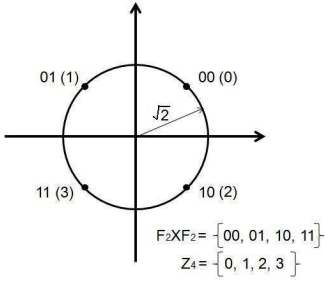}
\caption{4-PSK constellation}
\end{figure}

The remaining content is organized as follows: Section II discusses the basic concepts, definitions and a summary of the contributions of this paper. Section III demonstrates the network code obtained using Cartesian Product that is utilized at the relay for two-way ACF relaying which removes the fade states associated with the channels. In Section IV, we show how this network code can be obtained using Singularity Removal Constraints. Section V gives results based on structural properties of Latin Squares. In Section VI the complex plane is quantized depending on which one of the obtained Latin Squares maximizes the minimum cluster distance and Section VII gives the simulation results that demonstrate the improvement in the performance using the suggested scheme. Section VIII concludes the paper.

\section{Preliminaries}
Let $\mathcal{S}$ denote the symmetric 4-PSK constellation $\left\{\pm 1\pm i\right\}$ as shown in Fig. 2, used at A and B. Assume that A(/B) wants to send two 2-bit binary tuples to B(/A). Let $ \mu : \mathbb{F}^{2}_{2} \rightarrow \mathcal{S} $ denote the mapping from bits to complex symbols used at A and B where $\mathbb{F}_{2}=\left\{0,1\right\}$. Let $ x_{A_{1}}=\mu\left(s_{A_{1}}\right), x_{B_{1}}=\mu\left(s_{B_{1}}\right) \in \mathcal{S}$ denote the complex symbols transmitted by A and B at the first channel use respectively, and $ x_{A_{2}}=\mu\left(s_{A_{2}}\right), x_{B_{2}}=\mu\left(s_{B_{2}}\right) \in \mathcal{S}$ denote the complex symbols transmitted by A and B at the second channel use respectively, where $s_{A_{1}}, s_{B_{1}}, s_{A_{2}}, s_{B_{2}} \in \mathbb{F}^{2}_{2}.$\\

\noindent \textit{Multiple Access (MA) Phase:}\\
\indent  It is assumed that the channel state information is not available at the transmitting nodes A and B during the MA phase. The received signal at R at first channel use is given by 
\begin{equation}
\label{yr1}
Y_{R_{1}}=H_{A}x_{A_{1}}+H_{B}x_{B_{1}}+Z_{R_{1}}
\end{equation}
and the received signal at R at the second channel use,
\begin{equation}
\label{yr2}
Y_{R_{2}}=H_{A}x_{A_{2}}+H_{B}x_{B_{2}}+Z_{R_{2}}
\end{equation} 
where $H_{A}$ and $H_{B}$ are the fading coefficients associated with the A-R and B-R link respectively. Note that we are taking $H_{A}$ and $H_{B}$ to be the same for the two channel uses. The additive noise $Z_{R_{1}}$ and $Z_{R_{2}}$ are assumed to be $\mathcal{CN}\left(0,\sigma^2 \right)$, where $\mathcal{CN}\left(0,\sigma^2 \right)$ denotes the circularly symmetric complex Gaussian random variable with variance $\sigma^2$. We assume a block fading scenario, with $z=\gamma e^{j\theta}=H_{B}/H_{A}$, where $\gamma \in \mathbb{R}^+$ and $-\pi \leq \theta \leq \pi$, is referred to as the \textit{fade state} for the first and second transmission by A and B at the first and second channel use, and for simplicity can also be denoted by $\left(\gamma, \theta\right)$. Also, it is assumed that $z$ is distributed according to a continuous probability distribution. \\

\begin{figure*}
\footnotesize
\begin{align}
\label{dist}
&
d^{2}_{min}\left(\gamma e^{j\theta}\right)=\hspace{-0.5 cm}\min_{\substack {{((x_{A_{1}},x_{B_{1}}),(x_{A_{2}},x_{B_{2}})),((x'_{A_{1}},x'_{B_{1}}),(x'_{A_{2}},x'_{B_{2}})) \in \mathcal{S}^{4},} \\ {((x_{A_{1}},x_{B_{1}}),(x_{A_{2}},x_{B_{2}})) \neq ((x'_{A_{1}},x'_{B_{1}}),(x'_{A_{2}},x'_{B_{2}}))}}}\hspace{-0.1 cm} \left\{ \left| \left( x_{A_{1}}-x'_{A_{1}}\right)+\gamma e^{j \theta} \left(x_{B_{1}}-x'_{B_{1}}\right) \right|^{2} + \left| \left( x_{A_{2}}-x'_{A_{2}}\right)+\gamma e^{j \theta} \left(x_{B_{2}}-x'_{B_{2}}\right) \right|^{2} \right\}\\
\hline
\label{mle}
&
\left((\hat{x}_{A_{1}}, \hat{x}_{B_{1}}),(\hat{x}_{A_{2}}, \hat{x}_{B_{2}})\right)= \arg \min_{\substack {((x'_{A_{1}}, x'_{B_{1}}),(x'_{A_{2}}, x'_{B_{2}})) \in \mathcal{S}^{4}}} \left\{ \left|Y_{R_{1}} - H_{A} x'_{A_{1}}-H_{B} x'_{B_{1}}\right|^{2} + \left|Y_{R_{2}} - H_{A} x_{A'_{2}}-H_{B} x'_{B_{2}}\right|^{2}\right\}\\
\hline
\label{mel1}
&
\mathcal{M}^{\gamma, \theta}\left(\left(x_{A_{1}},x_{A_{2}}\right),\left(x_{B_{1}},x_{B_{2}}\right)\right) \neq \mathcal{M}^{\gamma, \theta}\left(\left(x'_{A_{1}},x'_{A_{2}}\right),\left(x_{B_{1}},x_{B_{2}}\right)\right), \ whenever \left(x_{A_{1}},x_{A_{2}}\right) \neq \left(x'_{A_{1}},x'_{A_{2}}\right) ~ \forall x_{B_{1}}, x_{B_{2}} \in \mathcal{S} \\
\hline
\label{mel2}
&
\mathcal{M}^{\gamma, \theta}\left(\left(x_{A_{1}},x_{A_{2}}\right),\left(x_{B_{1}},x_{B_{2}}\right)\right) \neq \mathcal{M}^{\gamma, \theta}\left(\left(x_{A_{1}},x_{A_{2}}\right),\left(x'_{B_{1}},x'_{B_{2}}\right)\right), \ whenever \left(x_{B_{1}},x_{B_{2}}\right) \neq \left(x'_{B_{1}},x'_{B_{2}}\right) ~ \forall x_{A_{1}}, x_{A_{2}} \in \mathcal{S} \\
\hline
\label{cl1}
&
\left(d_{min}^{\mathcal{L}_{i},\mathcal{L}_{j}}\left(\gamma e^{j \theta}\right)\right)^{2}=\hspace{-0.2 cm}\min_{\substack {{((x_{A_{1}},x_{B_{1}}),(x_{A_{2}},x_{B_{2}})) \in \mathcal{L}_{i}},\\ ((x'_{A_{1}},x'_{B_{1}}),(x'_{A_{2}},x'_{B_{2}})) \in \mathcal{L}_{j}}} \hspace{-0.2 cm} \left\{ \left| \left( x_{A_{1}}-x'_{A_{1}}\right)+\gamma e^{j \theta} \left(x_{B_{1}}-x'_{B_{1}}\right) \right|^{2} + \left| \left( x_{A_{2}}-x'_{A_{2}}\right)+\gamma e^{j \theta} \left(x_{B_{2}}-x'_{B_{2}}\right) \right|^{2} \right\}\\
\hline
\label{cl2}
&
d^{2}_{min}\left(\mathcal{C}^{\gamma, \theta}\right)=\hspace{-0.6 cm}\min_{\substack {{((x_{A_{1}},x_{B_{1}}),(x_{A_{2}},x_{B_{2}})),((x'_{A_{1}},x'_{B_{1}}),(x'_{A_{2}},x'_{B_{2}})) \in \mathcal{S}^{4},} \\ {\mathcal{M}^{\gamma, \theta}((x_{A_{1}},x_{B_{1}}),(x_{A_{2}},x_{B_{2}})) \neq \mathcal{M}^{\gamma, \theta}((x'_{A_{1}},x'_{B_{1}}),(x'_{A_{2}},x'_{B_{2}}))}}}\hspace{-0.6 cm} \left\{ \left| \left( x_{A_{1}}-x'_{A_{1}}\right)+\gamma e^{j \theta} \left(x_{B_{1}}-x'_{B_{1}}\right) \right|^{2} + \left| \left( x_{A_{2}}-x'_{A_{2}}\right)+\gamma e^{j \theta} \left(x_{B_{2}}-x'_{B_{2}}\right) \right|^{2} \right\}\\
\hline
\label{cl3}
&
d^{2}_{min}\left(\mathcal{C}^{h}, \gamma e^{j \theta}\right)=\hspace{-0.7 cm}\min_{\substack {{((x_{A_{1}},x_{B_{1}}),(x_{A_{2}},x_{B_{2}})),((x'_{A_{1}},x'_{B_{1}}),(x'_{A_{2}},x'_{B_{2}})) \in \mathcal{S}^{4},} \\ {\mathcal{M}^{h}((x_{A_{1}},x_{B_{1}}),(x_{A_{2}},x_{B_{2}})) \neq \mathcal{M}^{h}((x'_{A_{1}},x'_{B_{1}}),(x'_{A_{2}},x'_{B_{2}}))}}}\hspace{-0.4 cm} \left\{ \left| \left( x_{A_{1}}-x'_{A_{1}}\right)+ \gamma e^{j \theta} \left(x_{B_{1}}-x'_{B_{1}}\right) \right|^{2} + \left| \left( x_{A_{2}}-x'_{A_{2}}\right)+ \gamma e^{j \theta} \left(x_{B_{2}}-x'_{B_{2}}\right) \right|^{2} \right\}.\\
\hline
\nonumber
\end{align}
\end{figure*}

Let $ \mathcal{S}_{R} \left( \gamma, \theta \right)$ denote the effective constellation seen at the relay during the MA phase, i.e., 
$$ \mathcal{S}_{R} \left( \gamma, \theta \right) = \left\{(x_{i} + \gamma e^{j \theta}y_{i}, x_{j} + \gamma e^{j \theta}y_{j})| x_{i}, y_{i}, x_j, y_j \in \mathcal{S}\right\}. $$
The effective constellation remains the same over the two channel uses, since we assume $H_{A}$ and $H_{B}$ and hence the ratio $ H_{B}/H_{A}= \gamma e^{j \theta}$ to be the same during the two channel uses.\\

Let $d_{min}\left(\gamma e^{j \theta}\right)$ denote the minimum distance between the points in the constellation $ \mathcal{S}_{R} \left( \gamma, \theta \right) $ during MA phase, as given by (\ref{dist}) on the next page. From (\ref{dist}), it is clear that there exist values of $\gamma e^{j \theta}$, for which $d_{min}\left(\gamma e^{j \theta}\right)=0$. Let, $$\mathcal{H}=\left\{\gamma e^{j \theta} \in \mathbb{C} | d_{min}\left(\gamma e^{j \theta}\right)=0 \right\}.$$ The elements of $\mathcal{H}$ are called singular fade states. For singular fade states, $\left|\mathcal{S}_{R} \left( \gamma, \theta \right)\right|< 4^{4}.$\\

\begin{definition}
A fade state $\gamma e^{j \theta}$ is defined to be a \textit{singular fade state} for the \textit{ACF two-way relaying}, if the cardinality of the signal set $ \mathcal{S}_{R} \left( \gamma, \theta \right)$ is less than $4^{4}$. Let $\mathcal{H} $ denote the set of singular fade states for the two-way ACF relaying. \\
\end{definition}

Let $\left(\hat{x}_{A_{1}}, \hat{x}_{B_{1}}\right) \text{~and~} \left(\hat{x}_{A_{2}}, \hat{x}_{B_{2}}\right) \in \mathcal{S}^{2}$ denote the Maximum Likelihood (ML) estimate of $\left(x_{A_{1}}, x_{B_{1}}\right) \text{~and~} \left(x_{A_{2}}, x_{B_{2}}\right)$ at R based on the received complex numbers $Y_{R_{1}}$ and $Y_{R_{2}}$ at the two channel uses, as given in (\ref{mle}). \\

%
%
%

\noindent \textit{Broadcast (BC) Phase:}\\
\indent Depending on the value of $\gamma e^{j \theta}$, R chooses a map $\mathcal{M}^{\gamma, \theta} : \mathcal{S}^4 \rightarrow \mathcal{S}' $ where $\mathcal{S}^{'}$ is a complex signal set of size between $4^{2}$ and $4^{4}$ used by R during the \textit{BC} phase.\\

 The received signals at A and B during the BC phase are respectively given by,
\begin{equation}
Y_{A}=H'_{A}X_{R}+Z_{A} \text{~and~} Y_{B}=H'_{B}X_{R}+Z_{B}
\end{equation}

\noindent where $X_{R}=\mathcal{M}^{\gamma, \theta} \left(\left(\hat{x}_{A_{1}},\hat{x}_{B_{1}}\right), \left(\hat{x}_{A_{2}},\hat{x}_{B_{2}}\right)\right) \in \mathcal{S}'$ is the complex number transmitted by R. The fading coefficients corresponding to the R-A and R-B links are given by $H'_{A}$ and $H'_{B}$ respectively and the additive noises $Z_{A}$ and $Z_{B}$ are  $\mathcal{CN}\left(0,\sigma^{2}\right)$. \\

 The elements in $\mathcal{S}^4 $ which are mapped to the same signal point in $\mathcal{S}'$ by the map $\mathcal{M}^{\gamma, \theta}$ are said to form a cluster. Let $\left\{\mathcal{L}_{1}, \mathcal{L}_{2},.., \mathcal{L}_{l}\right\}$ denote the set of all such clusters. The formation of clusters is called clustering, denoted by $\mathcal{C}^{\gamma e^{j \theta}}$.\\

In order to ensure that A(/B) is able to decode B's(/A's) messages, the clustering $\mathcal{C}^{\gamma e^{j \theta}}$ should satisfy the exclusive law, as given in (\ref{mel1}), (\ref{mel2}) above.\\

\begin{definition}
The cluster distance between a pair of clusters $\mathcal{L}_i$ and $\mathcal{L}_j$ is the minimum among all the distances calculated between the points $\left(\left(x_{A_{1}},x_{B_{1}}\right), \left(x_{A_{2}},x_{B_{2}}\right)\right) \in \mathcal{L}_{i}$ and $\left(\left(x'_{A_{1}},x'_{B_{1}}\right), \left(x'_{A_{2}},x'_{B_{2}}\right)\right) \in \mathcal{L}_{j}$ in the effective constellation used by the relay node R, as given in \eqref{cl1} above.\\

\end{definition}

\begin{definition}
The \textit{minimum cluster distance} of the clustering $\mathcal{C}^{\gamma e^{j \theta}}$ is the minimum among all the cluster distances, as given in \eqref{cl2} above.\\
\end{definition}

The minimum cluster distance determines the performance during the MA phase of relaying. The performance during the BC phase is determined by the minimum distance of the signal set $\mathcal{S}^{'}$. For values of $ \gamma e^{j \theta}$ in the neighborhood of the singular fade states, the value of $d_{min}\left(\mathcal{C}^{\gamma e^{j \theta}}\right)$ is greatly reduced, a phenomenon referred to as \textit{distance shortening} [4]. To avoid distance shortening, for each singular fade state, a clustering needs to be chosen such that the minimum cluster distance is non zero. \\

A clustering $\mathcal{C}^h$ is said to remove singular fade state $h \in \mathcal{H}, $ if $d_{min}\left(\mathcal{C}^h\right)>0$. For a singular fade state $h \in \mathcal{H} $, let $\mathcal{C}^{h} $ denote the clustering which removes the singular fade state $h$ (if there are multiple clusterings which remove the same singular fade state $h$, choose any of the clusterings). Let $\mathcal{C^{H}}=\left\{\mathcal{C}^{h} :h\in \mathcal{H}\right\} $ denote the set of all such clusterings. \\

\begin{definition}
The minimum cluster distance of the clustering $\mathcal{C}^{h}, h\in \mathcal{H}$ at the fade state $\gamma e^{j \theta}$ which is not necessarily a singular fade state, denoted by $d_{min}\left(\mathcal{C}^{h},\gamma e^{j \theta}\right)$, is as given in (\ref{cl3}). Note that if $\gamma e^{j\theta} =h \in \mathcal{H},$ $d_{min}\left(\mathcal{C}^{h},h\right),$ reduces to $d_{min}\left(\mathcal{C}^{h}\right)$ given in \eqref{cl2}.
\end{definition}

In general, the channel fade state $\gamma e^{j \theta}$ need not be a singular fade state. In such a scenario, among all the clusterings which remove the singular fade states, the one which has the maximum value of the minimum cluster distance at $\gamma e^{j\theta}$ is chosen by the relay R. In other words, for $\gamma e^{j\theta} \notin \mathcal{H}, $ the clustering $\mathcal{C}^{\gamma,{\theta}} $ is chosen to be $\mathcal{C}^{h}$ by the relay R, which satisfies $ d_{min}\left(\mathcal{C}^{h},\gamma e^{j\theta}\right) \geq d_{min}\left(\mathcal{C}^{h'},\gamma e^{j\theta}\right), \forall h \neq h' \in \mathcal{H}.$ Since the clusterings which remove the singular fade states are known to all the three nodes and are finite in number, the clustering used for a particular realization of the fade state can be indicated by R to A and B using overhead bits.

In \cite{NMR}, such clusterings that remove singular fade states for two-way 2-stage relaying scenario were obtained with the help of Latin Squares.\\

The contributions of this paper are as follows:
\begin{itemize}
\item It is shown that if the users A and B transmit points from the same 4-PSK constellation, the clusterings proposed in \cite{NMR} for two user case can be utilized to get clusterings for this case for removing the singular fade states containing either 16 or 25 points by introducing the notion of Cartesian Product of Clusters. In other words, the $16 \times 16$ Latin Squares representing the ACF relaying can be obtained with the help of $4 \times 4$ Latin Squares representing the clusterings for two-way 2-stage relaying as given in \cite{NMR}.
\item Another clustering is proposed for the ACF protocol in the two-way relay channel called Direct Clustering. This clustering also removes the singular fade states and reduces the number of clusters for some cases. Using this clustering, the size of the resulting constellation used by the relay node R in the BC phase is reduced to 20 for a category of cases, as compared to the Cartesian Product approach which results in the constellation size being 25 for these cases.
\item The quantization of the complex plane that contains all the possible fade states, depending on which one of the obtained clusterings maximizes the minimum cluster distance, is proven to be the same as for the two-way 2-stage relaying scenario as done in \cite{MNR}.
\item Simulation results indicate that at high SNR, the schemes based on the ACF protocol performs better than the schemes proposed in \cite{KoPoTa},\cite{NMR} based on two-stage two way relaying. With 4-PSK signal set used at the end nodes, the ACF protocol achieves a maximum sum throughput of 8/3 bits/s/Hz, whereas it is 2 bits/s/Hz for the schemes based on 2-stage two way relaying.\\
\end{itemize}

\section{Exclusive Law and Latin Squares}
The nodes A and B transmit symbols from the same constellation, viz., 4-PSK. Our aim is to find the map that the relay node R should use in order to cluster the $4^4$ possibilities of $\left(\left(x_{A_{1}},x_{B_{1}}\right),\left(x_{A_{2}},x_{B_{2}}\right)\right)$ such that the exclusive law given by (\ref{mel1}), (\ref{mel2}) is satisfied. Consider the $16 \times 16$ array consisting of the 16 possibilities of $\left(x_{A_{1}},x_{A_{2}}\right) $ along the rows and the 16 possibilities of $\left(x_{B_{1}},x_{B_{2}}\right) $ along the columns. We fill this array with elements from $\mathcal{L}=\left\{ \mathcal{L}_{1}, \mathcal{L}_{2}, ..., \mathcal{L}_{t}\right\}$ where each symbol denotes a unique cluster. The constellation size used by the relay in the BC phase, or the number of clusters has to be at least 16, since each user needs 4 bit information corresponding to the two messages sent by the other user implying $t \geq 16$. In order to keep (\ref{mel1}) and (\ref{mel2}) satisfied, a symbol from $\mathcal{L}$ can occur at most once in each row and column. So the $16 \times 16$ array having $\left(x_{A_{1}},x_{A_{2}}\right) $ along the rows and  $\left(x_{B_{1}},x_{B_{2}}\right) $ along the columns must be a Latin Square of side 16 (Definition 1). The equivalence between the network code used by the relay in two-way relaying scenario and Latin Squares has been previously discussed in [7]. The clusters are obtained by putting together all those $\left(\left(x_{A_{1}},x_{B_{1}}\right),\left(x_{A_{2}},x_{B_{2}}\right)\right)$ for which the corresponding entry $\left(\left(x_{A_{1}},x_{A_{2}}\right),\left(x_{B_{1}},x_{B_{2}}\right)\right)$ in the array is the same. \\

From above, we can say that all the relay clusterings that satisfy the mutually exclusive law forms Latin Squares of order 16 with entries from $\mathcal{L}$ with $t \geq 16$, when the end nodes use PSK constellations of size 4. It therefore suffices to consider the network code used by the relay node in the BC phase to be a $16 \times 16$ array with rows(/columns) indexed by the 2-tuple consisting of the symbols sent by A(/B) during the first and second channel use. The cells of the array must be filled with elements of $\mathcal{L}$ in such a way, that the resulting array is a Latin Square of order 16 and $t \geq 16$. \\

\noindent \textit{\textbf{Removing Singular fade states and Constrained Latin Squares}}

The relay can manage with constellations of size 16 in BC phase, but it is observed that in some cases relay may not be able to remove the singular fade states and results in severe performance degradation in the MA phase. As stated in Section II, that a clustering $\mathcal{C}^{h}$ is said to remove singular fade state $h \in \mathcal{H}, $ if $d_{min}\left(\mathcal{C}^h\right)>0$. Removing singular fade states for a two-way ACF relay channel can also be defined as follows:\\

\begin{definition}
A clustering $\mathcal{C}^{h}$ is said to \textit{remove the singular fade state} $h \in \mathcal{H}$, if any two possibilities of the messages sent by the users $\left(\left(x_{A_{1}},x_{B_{1}}\right),\left(x_{A_{2}},x_{B_{2}}\right)\right), \left(\left(x'_{A_{1}},x'_{B_{1}}\right),\left(x'_{A_{2}},x'_{B_{2}}\right)\right) \in \mathcal{S}^{4}$ that satisfy\\
$$h= \frac{x'_{A_{1}}-x_{A_{1}}}{x_{B_{1}}-x'_{B_{1}}}= \frac{x'_{A_{2}}-x_{A_{2}}}{x_{B_{2}}-x'_{B_{2}}}, $$
are placed together in the same cluster by $\mathcal{C}^{h}$. \\
\end{definition}

\begin{definition}
A set $\left\{\left(\left(x_{A_{1}},x_{B_{1}}\right),\left(x_{A_{2}},x_{B_{2}}\right)\right)\right\} \subseteq \mathcal{S}^4$ consisting of all the possibilities of $((x_{A_{1}},x_{B_{1}}),(x_{A_{2}},x_{B_{2}}))$ that must be placed in the same cluster of the clustering used at relay node R in the BC phase in order to remove the singular fade state $h$ is referred to as a \textit{Singularity Removal Constraint} for the singular fade state $h$ in two-way ACF relaying scenario.\\
\end{definition}

As given in \cite{KoPoTa}, a complex number $\gamma e^{j \theta}$ is defined to be a \textit{singular fade state} for the \textit{two-way 2-stage relaying scenario}, if $ x_{A}+ \gamma e^{j \theta} x_{B} = x'_{A}+ \gamma e^{j \theta} x'_{B} \text{~for some~} (x_A,x_B),(x'_A,x'_B) \in \mathcal{S}^{2}$ and the set $\left\{(x_A,x_B) ~|~ (x_A,x_B) \in \mathcal{S}^{2}\right\}$ consisting of all the possibilities of $(x_A,x_B) \in \mathcal{S}^{2}$ that must be placed in the same cluster of the clustering that removes the fade state $\gamma e^{j \theta}$ for the two-way 2-stage relaying scenario is the corresponding set of singularity removal constraint for the singular fade state $\gamma e^{j \theta}$. As we show in the following lemma, the singular fade state for the ACF two-way relaying are the same as the singular fade state for the two-way 2-stage relaying scenario.\\

\begin{lemma} The singular fade states for the ACF two-way relaying scenario are the same as the 12 singular fade states for two-way 2-stage relaying scenario as computed in \cite{KoPoTa}. 
\end{lemma}
\begin{proof} Let $\gamma e^{j \theta}$ be a singular fade state for the ACF two-way relaying scenario. By definition, $ \exists ((x_{A_{1}},x_{B_{1}}),(x_{A_{2}},x_{B_{2}})),((x'_{A_{1}},x'_{B_{1}}),(x'_{A_{2}},x'_{B_{2}})) \in \mathcal{S}^{4} $ such that,
\begin{align}
\nonumber
&x_{A_{1}}+ \gamma e^{j \theta} x_{B_{1}} = x'_{A_{1}}+ \gamma e^{j \theta} x'_{B_{1}}, \text{~and} \\
\nonumber
&x_{A_{2}}+ \gamma e^{j \theta} x_{B_{2}} = x'_{A_{2}}+ \gamma e^{j \theta} x'_{B_{2}}. 
\end{align} 
where $x_{A_{1}},x_{B_{1}},x_{A_{2}},x_{B_{2}},x'_{A_{1}},x'_{B_{1}},x'_{A_{2}},x'_{B_{2}} \in \mathcal{S}$. \\

Then, by definition, $\gamma e^{j \theta}$ must be a singular fade state for two-user 2-stage relaying.\\

Conversely, let $\gamma e^{j \theta}$ be a singular fade state for two-user 2-stage relaying scenario. Then, $\exists (x_{A},x_{B}),~ (x'_{A},x'_{B}) \in \mathcal{S}^{2}$ such that,
$$ x_{A}+ \gamma e^{j \theta} x_{B} = x'_{A}+ \gamma e^{j \theta} x'_{B}.$$
Then, since for any $ (x,y) \in \mathcal{S},~ x+ \gamma e^{j \theta} y= x+ \gamma e^{j \theta} y, ~ \left\{ ((x_{A},x_{B}),(x,y)),~((x'_{A},x'_{B}),(x,y)) \right\}$ is a subset of a singularity removal constraint and $\gamma e^{j \theta}$ is a singular fade state for the ACF two-way relaying scenario.\\

Thus, the singular fade states for the ACF two-way relaying scenario and the singular fade states for the two-way 2-stage relaying scenario are the same.
\end{proof}

\vspace{0.5cm}
Let $\gamma e^{j \theta}$ be a fade state for the two-way ACF relaying scenario. Then $\gamma e^{j \theta}$ can be viewed as a fade state for the first and second channel use in the MA phase as shown in Lemma 1. In \cite{KoPoTa} and \cite{MNR}, it is shown that for the two-way 2-stage relaying, the $4^{2}$ possible pairs of symbols from 4-PSK constellation sent by the two users in the MA phase, can be clustered into a clustering dependent on a singular fade coefficient, of size 4 or 5 in a manner so as to remove this singular fade coefficient. In the case of two-way ACF relaying, at the end of MA phase, relay receives two complex numbers, given by (\ref{yr1}) and (\ref{yr2}). Instead of R transmitting a point from the $4^{4}$ point constellation resulting from all the possibilities of $\left(\left(x_{A_{1}}, x_{B_{1}}\right),\left(x_{A_{2}}, x_{B_{2}}\right)\right)$ the relay R can choose to group these  possibilities into clusters represented by a smaller constellation. One such clustering for the case when $\gamma e^{j \theta} $ can be obtained by utilizing the clustering provided in [7] for the two-way 2-stage problem in order to remove this fade state. Let $\mathcal{C}^{\left[h\right]}$ denote the clustering for the physical network coded two-way relaying scenario that removes the singular fade state $h \in \mathbb{C}$ for the two-way 2-stage relaying case as given in [7]. \\

\begin{definition} We define the \textit{Cartesian Product} of a clustering $\mathcal{C}^{\left[h\right]}=\left\{l_{1}, l_{2},...,l_{m}\right\}$ with itself denoted by $\mathcal{D}^{\left[h\right]}$, where for $i=1,2,...,m$;
$$l_{i}=\left\{\left(x_{i_{1}},y_{i_{1}}\right), \left(x_{i_{2}},y_{i_{2}}\right), ..., \left(x_{i_{s_{i}}},y_{i_{s_{i}}}\right)\right\}$$
with $x_{i_{p}}, y_{i_{p}} \in \mathbb{Z}_{4} ~ \forall~p=1,2,...,s_{i} $ as follows:
$$\mathcal{D}^{\left[h\right]}=\left\{\mathcal{C}^{\left\{l_{1},l_{1}\right\}},...,\mathcal{C}^{\left\{l_{1},l_{m}\right\}},...,\mathcal{C}^{\left\{l_{m},l_{1}\right\}},...,\mathcal{C}^{\left\{l_{m},l_{m}\right\}} \right\} $$
where,
{\footnotesize
\begin{align}
\nonumber
&\mathcal{C}^{\left\{l_{i},l_{j}\right\}}=\left\{\left((x_{i_{p}},y_{i_{p}}),(x_{j_{q}},y_{j_{q}})\right) | ~ p=1,2,..,s_{i} \text{~and~} q=1,2,..,s_{j} \right\}.
\end{align}
}
\end{definition}
\vspace{0.5cm}
\begin{lemma} Let $\gamma e^{j \theta} \in \mathcal{H}$. The clustering obtained by taking the Cartesian Product $\mathcal{D}^{\left[\gamma e^{j \theta}\right]}$ of $\mathcal{C}^{\left[\gamma e^{j \theta}\right]}$ with itself removes the singular fade state $\gamma e^{j \theta}$ for the two-way ACF relaying scenario.
\end{lemma}
\begin{proof} Let $\mathcal{C}^{\left[\gamma e^{j \theta}\right]}=\left\{l_1, l_2,...,l_m\right\}$,where for $i=1,2,...,m$,
$$ l_i=\left\{(x_{i_{1}},y_{i_{1}}),(x_{i_{2}},y_{i_{2}}),...,(x_{i_{s_{i}}},y_{i_{s_{i}}})\right\}$$
Then,
$$\mathcal{D}^{\left[\gamma e^{j \theta} \right]}=\left\{\mathcal{C}^{\left\{l_{1},l_{1}\right\}},...,\mathcal{C}^{\left\{l_{1},l_{m}\right\}},...,\mathcal{C}^{\left\{l_{m},l_{1}\right\}},...,\mathcal{C}^{\left\{l_{m},l_{m}\right\}} \right\} $$
where,
{\footnotesize
\begin{align}
\nonumber
&\mathcal{C}^{\left\{l_{i},l_{j}\right\}}=\left\{\left((x_{i_{p}},y_{i_{p}}),(x_{j_{q}},y_{j_{q}})\right) | ~ p=1,2,..,s_{i} \text{~and~} q=1,2,..,s_{j} \right\}.
\end{align}
}
By definition, a singularity removal constraint for the fade state $\gamma e^{j \theta}$ in the case of two-way ACF relaying scenario is a set $\left\{((x_{i},y_{i}),(x'_{i},y'_{i}))~|~i=1,2,...,t\right\}$ such that $\forall 1\leq i_{1},i_{2} \leq t$,
$$ \gamma e^{j \theta}=\frac{x_{i_{2}}-x_{i_{1}}}{y_{i_{1}}-y_{i_{2}}}=\frac{x'_{i_{2}}-x'_{i_{1}}}{y'_{i_{1}}-y'_{i_{2}}} $$

Now, $ \gamma e^{j \theta}=\frac{x_{i_{2}}-x_{i_{1}}}{y_{i_{1}}-y_{i_{2}}} \Rightarrow (x_{1},y_{1}) \text{~and~} (x_{2},y_{2})$ must belong to the same cluster, say $l_i$ in $\mathcal{C}^{\left[\gamma e^{j \theta}\right]}$ for it to remove the fade state $\gamma e^{j \theta}$. Similarly,
$ \gamma e^{j \theta}=\frac{x'_{i_{2}}-x'_{i_{1}}}{y'_{i_{1}}-y'_{i_{2}}} \Rightarrow (x'_{1},y'_{1}) \text{~and~} (x'_{2},y'_{2})$ must belong to the same cluster, say $l_j$ in $\mathcal{C}^{\left[\gamma e^{j \theta}\right]}$. This holds $\forall 1\leq i_{1},i_{2} \leq t$. Thus, the singularity removal constraint is,

{
\vspace{-0.2cm}
\begin{align}
\nonumber
&\left\{((x_{i},y_{i}),(x'_{i},y'_{1}))~|~i=1,2,...,t\right\} \subseteq \mathcal{C}^{\left\{l_{i},l_{j}\right\}},
\end{align}
}for some $1\leq i,j \leq n$.
Therefore, the clustering $\mathcal{D}^{\left[\gamma e^{j \theta}\right]}$ removes the singular fade state $\gamma e^{j \theta}$.
\end{proof}
\vspace{0.5cm}

It was studied in [4] by computer search, and then in [7] analytically, that there are 12 possible singular fade states in the complex plane for the case when two users transmit points from 4-PSK constellation in the MA phase. Out of these 12 fade states, 4 lie on the unit circle, 4 lie on a circle of radius $\sqrt{2}$ and 4 lie on a circle of radius 1/$\sqrt{2}$. The size of the constellation used at R in the BC phase for these cases is either 4 or 5. For the two-user ACF scenario we are dealing with, we have two channel uses in the MA phase, with A and B transmitting a message each in the first channel and second channel uses, with the two messages sent by each user in the first and second channel use being possibly different. Keeping in mind the three classes of fade states depending on the radius of the circle it lies on, we consider the following three cases:\\\\
\textit{Case 1:} $\gamma e^{j \theta} $ lies on the unit circle.\\\\
\textit{Case 2:} $\gamma e^{j \theta} $ lies on the circle of radius $ 1/\sqrt{2}$.\\\\
\textit{Case 3:} $\gamma e^{j \theta} $ lies on the circle of radius $ \sqrt{2}$.\\\\
\textbf{\textit{Case 1:} $\gamma e^{j \theta} $ lies on the unit circle.}\\

Since both user nodes A and B require $4$ bits of information from the other user, the size of the constellation that R uses will be at least $2^4=16$. The Cartesian Product of $\mathcal{C}^{\left[\gamma e^{j \theta}\right]}$ with itself consists of $16$ clusters since the clustering $\mathcal{C}^{\left[\gamma e^{j \theta}\right]}$ has $4$ clusters. We illustrate this case with the help of the following example, the remaining instances of fade states that lie on the unit circle can be obtained using this example, as we show in Lemma 6 in Section V.\\

\begin{example} Let the fade state $\gamma e^{j \theta}= j$. The clustering $\mathcal{C}^{\left[j\right]}$ for the case as given in \cite{NMR} is given by,
\begin{align}
\nonumber
\mathcal{C}^{\left[j\right]}=&\left\{l_{1},l_{2},l_{3},l_{4}\right\} 
\nonumber
\text{where,}\\
\nonumber
&l_{1}=\left\{\left(0,0\right),\left(1,2\right),\left(2,1\right),\left(3,3\right)\right\}\\
\nonumber
&l_{2}=\left\{\left(0,3\right),\left(1,1\right),\left(2,2\right),\left(3,0\right)\right\}\\
\nonumber
&l_{3}=\left\{\left(0,1\right),\left(1,3\right),\left(2,0\right),\left(3,2\right)\right\}\\
\nonumber
&l_{4}=\left\{\left(0,2\right),\left(1,0\right),\left(2,3\right),\left(3,1\right)\right\}.
\vspace{-0.8cm}
\nonumber
\end{align}
The Cartesian Product of the above clustering given by $\mathcal{D}^{\left[j\right]}=\left\{ \mathcal{C}^{\left\{l_{i},l_{j}\right\}}~ |~ i,j=1,2,3,4\right\}$ contains exactly $16$ clusters:
{\footnotesize
\begin{align}
\nonumber
\vspace{-0.5cm}
\mathcal{C}^{\left\{l_{1},l_{1}\right\}}=&\left\{((0,0),(0,0)), ((0,0),(1,2)), ((0,0),(2,1)), ((0,0),(3,3)), \right.\\
\nonumber 
& \left. ((1,2),(0,0)), ((1,2),(1,2)), ((1,2),(2,1)), ((1,2),(3,3)),\right.\\
\nonumber 
& \left. ((2,1),(0,0)), ((2,1),(1,2)), ((2,1),(2,1)), ((2,1),(3,3)),\right.\\
\nonumber 
& \left.  ((3,3),(0,0)), ((3,3),(1,2)), ((3,3),(2,1)), ((3,3),(3,3))\right\}\\
\nonumber
\mathcal{C}^{\left\{l_{1},l_{2}\right\}}=&\left\{((0,0),(0,3)), ((0,0),(1,1)), ((0,0),(2,2)), ((0,0),(3,0)), \right.\\
\nonumber 
& \left. ((1,2),(0,3)), ((1,2),(1,1)), ((1,2),(2,2)), ((1,2),(3,0)),\right.\\
\nonumber 
& \left. ((2,1),(0,3)), ((2,1),(1,1)), ((2,1),(2,2)), ((2,1),(3,0)),\right.\\
\nonumber 
& \left. ((3,3),(0,3)), ((3,3),(1,1)), ((3,3),(2,2)), ((3,3),(3,0))\right\}\\
\nonumber
\mathcal{C}^{\left\{l_{1},l_{3}\right\}}=&\left\{((0,0),(0,1)), ((0,0),(1,3)), ((0,0),(2,0)), ((0,0),(3,2)), \right.\\
\nonumber 
& \left. ((1,2),(0,1)), ((1,2),(1,3)), ((1,2),(2,0)), ((1,2),(3,2)),\right.\\
\nonumber 
& \left. ((2,1),(0,1)), ((2,1),(1,3)), ((2,1),(2,0)), ((2,1),(3,2)),\right.\\
\nonumber 
& \left. ((3,3),(0,1)), ((3,3),(1,3)), ((3,3),(2,0)), ((3,3),(3,2))\right\}\\
\nonumber
\mathcal{C}^{\left\{l_{1},l_{4}\right\}}=&\left\{((0,0),(0,2)), ((0,0),(1,0)), ((0,0),(2,3)), ((0,0),(3,1)), \right.\\
\nonumber 
& \left. ((1,2),(0,2)), ((1,2),(1,0)), ((1,2),(2,3)), ((1,2),(3,1)),\right.\\
\nonumber 
& \left. ((2,1),(0,2)), ((2,1),(1,0)), ((2,1),(2,3)), ((2,1),(3,1)),\right.\\
\nonumber 
& \left. ((3,3),(0,2)), ((3,3),(1,0)), ((3,3),(2,3)), ((3,3),(3,1))\right\}\\
\nonumber
\mathcal{C}^{\left\{l_{2},l_{1}\right\}}=&\left\{((0,3),(0,0)), ((0,3),(1,2)), ((0,3),(2,1)), ((0,3),(3,3)), \right.\\
\nonumber 
& \left. ((1,1),(0,0)), ((1,1),(1,2)), ((1,1),(2,1)), ((1,1),(3,3)),\right.\\
\nonumber 
& \left. ((2,2),(0,0)), ((2,2),(1,2)), ((2,2),(2,1)), ((2,2),(3,3)),\right.\\
\nonumber 
& \left. ((3,0),(0,0)), ((3,0),(1,2)), ((3,0),(2,1)), ((3,0),(3,3))\right\}\\
\nonumber
\mathcal{C}^{\left\{l_{2},l_{2}\right\}}=&\left\{((0,3),(0,3)), ((0,3),(1,1)), ((0,3),(2,2)), ((0,3),(3,0)), \right.\\
\nonumber 
& \left. ((1,1),(0,3)), ((1,1),(1,1)), ((1,1),(2,2)), ((1,1),(3,0)),\right.\\
\nonumber 
& \left. ((2,2),(0,3)), ((2,2),(1,1)), ((2,2),(2,2)), ((2,2),(3,0)),\right.\\
\nonumber 
& \left. ((3,0),(0,3)), ((3,0),(1,1)), ((3,0),(2,2)), ((3,0),(3,0))\right\}\\
\nonumber
\mathcal{C}^{\left\{l_{2},l_{3}\right\}}=&\left\{((0,3),(0,1)), ((0,3),(1,3)), ((0,3),(2,0)), ((0,3),(3,2)), \right.\\
\nonumber 
& \left. ((1,1),(0,1)), ((1,1),(1,3)), ((1,1),(2,0)), ((1,1),(3,2)),\right.\\
\nonumber 
& \left. ((2,2),(0,1)), ((2,2),(1,3)), ((2,2),(2,0)), ((2,2),(3,2)),\right.\\
\nonumber 
& \left. ((3,0),(0,1)), ((3,0),(1,3)), ((3,0),(2,0)), ((3,0),(3,2))\right\}\\
\nonumber
\mathcal{C}^{\left\{l_{2},l_{4}\right\}}=&\left\{((0,3),(0,2)), ((0,3),(1,0)), ((0,3),(2,3)), ((0,3),(3,1)), \right.\\
\nonumber 
& \left. ((1,1),(0,2)), ((1,1),(1,0)), ((1,1),(2,3)), ((1,1),(3,1)),\right.\\
\nonumber 
& \left. ((2,2),(0,2)), ((2,2),(1,0)), ((2,2),(2,3)), ((2,2),(3,1)),\right.\\
\nonumber 
& \left. ((3,0),(0,2)), ((3,0),(1,0)), ((3,0),(2,3)), ((3,0),(3,1))\right\}\\
\nonumber
\mathcal{C}^{\left\{l_{3},l_{1}\right\}}=&\left\{((0,1),(0,0)), ((0,1),(1,2)), ((0,1),(2,1)), ((0,1),(3,3)), \right.\\
\nonumber 
& \left. ((1,3),(0,0)), ((1,3),(1,2)), ((1,3),(2,1)), ((1,3),(3,3)),\right.\\
\nonumber 
& \left. ((2,0),(0,0)), ((2,0),(1,2)), ((2,0),(2,1)), ((2,0),(3,3)),\right.\\
\nonumber 
& \left. ((3,2),(0,0)), ((3,2),(1,2)), ((3,2),(2,1)), ((3,2),(3,3))\right\}\\
\nonumber
\mathcal{C}^{\left\{l_{3},l_{2}\right\}}=&\left\{((0,1),(0,3)), ((0,1),(1,1)), ((0,1),(2,2)), ((0,1),(3,0)), \right.\\
\nonumber 
& \left. ((1,3),(0,3)), ((1,3),(1,1)), ((1,3),(2,2)), ((1,3),(3,0)),\right.\\
\nonumber 
& \left. ((2,0),(0,3)), ((2,0),(1,1)), ((2,0),(2,2)), ((2,0),(3,0)),\right.\\
\nonumber 
& \left. ((3,2),(0,3)), ((3,2),(1,1)), ((3,2),(2,2)), ((3,2),(3,0))\right\}\\
\nonumber
\mathcal{C}^{\left\{l_{3},l_{3}\right\}}=&\left\{((0,1),(0,1)), ((0,1),(1,3)), ((0,1),(2,0)), ((0,1),(3,2)), \right.\\
\nonumber 
& \left. ((1,3),(0,1)), ((1,3),(1,3)), ((1,3),(2,0)), ((1,3),(3,2)),\right.\\
\nonumber 
& \left. ((2,0),(0,1)), ((2,0),(1,3)), ((2,0),(2,0)), ((2,0),(3,2)),\right.\\
\nonumber 
& \left. ((3,2),(0,1)), ((3,2),(1,3)), ((3,2),(2,0)), ((3,2),(3,2))\right\}\\
\nonumber
\mathcal{C}^{\left\{l_{3},l_{4}\right\}}=&\left\{((0,1),(0,2)), ((0,1),(1,0)), ((0,1),(2,3)), ((0,1),(3,1)), \right.\\
\nonumber 
& \left. ((1,3),(0,2)), ((1,3),(1,0)), ((1,3),(2,3)), ((1,3),(3,1)),\right.\\
\nonumber 
& \left. ((2,0),(0,2)), ((2,0),(1,0)), ((2,0),(2,3)), ((2,0),(3,1)),\right.\\
\nonumber 
& \left. ((3,2),(0,2)), ((3,2),(1,0)), ((3,2),(2,3)), ((3,2),(3,1))\right\}\\
\nonumber
\mathcal{C}^{\left\{l_{4},l_{1}\right\}}=&\left\{((0,2),(0,0)), ((0,2),(1,2)), ((0,2),(2,1)), ((0,2),(3,3)), \right.\\
\nonumber 
& \left. ((1,0),(0,0)), ((1,0),(1,2)), ((1,0),(2,1)), ((1,0),(3,3)),\right.\\
\nonumber 
& \left. ((2,3),(0,0)), ((2,3),(1,2)), ((2,3),(2,1)), ((2,3),(3,3)),\right.\\
\nonumber 
& \left. ((3,1),(0,0)), ((3,1),(1,2)), ((3,1),(2,1)), ((3,1),(3,3))\right\}
\end{align}
\begin{align}
\nonumber
\mathcal{C}^{\left\{l_{4},l_{2}\right\}}=&\left\{((0,2),(0,3)), ((0,2),(1,1)), ((0,2),(2,2)), ((0,2),(3,0)), \right.\\
\nonumber 
& \left. ((1,0),(0,3)), ((1,0),(1,1)), ((1,0),(2,2)), ((1,0),(3,0)),\right.\\
\nonumber 
& \left. ((2,3),(0,3)), ((2,3),(1,1)), ((2,3),(2,2)), ((2,3),(3,0)),\right.\\
\nonumber 
& \left. ((3,1),(0,3)), ((3,1),(1,1)), ((3,1),(2,2)), ((3,1),(3,0))\right\}\\
\nonumber
\mathcal{C}^{\left\{l_{4},l_{3}\right\}}=&\left\{((0,2),(0,1)), ((0,2),(1,3)), ((0,2),(2,0)), ((0,2),(3,2)), \right.\\
\nonumber 
& \left. ((1,0),(0,1)), ((1,0),(1,3)), ((1,0),(2,0)), ((1,0),(3,2)),\right.\\
\nonumber 
& \left. ((2,3),(0,1)), ((2,3),(1,3)), ((2,3),(2,0)), ((2,3),(3,2)),\right.\\
\nonumber 
& \left. ((3,1),(0,1)), ((3,1),(1,3)), ((3,1),(2,0)), ((3,1),(3,2))\right\}\\
\nonumber
\mathcal{C}^{\left\{l_{4},l_{4}\right\}}=&\left\{((0,2),(0,2)), ((0,2),(1,0)), ((0,2),(2,3)), ((0,2),(3,1)), \right.\\
\nonumber 
& \left. ((1,0),(0,2)), ((1,0),(1,0)), ((1,0),(2,3)), ((1,0),(3,1)),\right.\\
\nonumber 
& \left. ((2,3),(0,2)), ((2,3),(1,0)), ((2,3),(2,3)), ((2,3),(3,1)),\right.\\
\nonumber 
& \left. ((3,1),(0,2)), ((3,1),(1,0)), ((3,1),(2,3)), ((3,1),(3,1)) \right\}
\nonumber
\end{align}
}

\begin{figure*}
{\footnotesize
{
\renewcommand{\arraystretch}{1,3}
\begin{tabular}{!{\vrule width 1pt}c!{\vrule width 1pt}c|c|c|c!{\vrule width 1pt}c|c|c|c!{\vrule width 1pt}c|c|c|c!{\vrule width 1pt}c|c|c|c!{\vrule width 1pt}} \noalign{\hrule height 1pt}
    &$(0,0)$&$(0,1)$&$(0,2)$&$(0,3)$&$(1,0)$&$(1,1)$&$(1,2)$&$(1,3)$&$(2,0)$&$(2,1)$&$(2,2)$&$(2,3)$&$(3,0)$&$(3,1)$&$(3,2)$&$(3,3)$ \\\noalign{\hrule height 1pt}

 $(0,0)$ & $\mathcal{L}_{1}$ & $\mathcal{L}_{3}$ & $\mathcal{L}_{4}$ & $\mathcal{L}_{2}$ & $\mathcal{L}_{9}$ & $\mathcal{L}_{11}$ & $\mathcal{L}_{12}$ & $\mathcal{L}_{10}$ & $\mathcal{L}_{13}$ & $\mathcal{L}_{15}$ & $\mathcal{L}_{16}$ & $\mathcal{L}_{14}$ & $\mathcal{L}_{5}$ & $\mathcal{L}_{7}$ & $\mathcal{L}_{8}$ & $\mathcal{L}_{6}$ \\\hline 
 $(0,1)$ & $\mathcal{L}_{4}$ & $\mathcal{L}_{2}$ & $\mathcal{L}_{1}$ & $\mathcal{L}_{3}$ & $\mathcal{L}_{12}$ & $\mathcal{L}_{10}$ & $\mathcal{L}_{9}$ & $\mathcal{L}_{11}$ & $\mathcal{L}_{16}$ & $\mathcal{L}_{14}$ & $\mathcal{L}_{13}$ & $\mathcal{L}_{15}$ & $\mathcal{L}_{8}$ & $\mathcal{L}_{6}$ & $\mathcal{L}_{5}$ & $\mathcal{L}_{7}$ \\\hline
 $(0,2)$ & $\mathcal{L}_{3}$ & $\mathcal{L}_{1}$ & $\mathcal{L}_{2}$ & $\mathcal{L}_{4}$ & $\mathcal{L}_{11}$ & $\mathcal{L}_{9}$ & $\mathcal{L}_{10}$ & $\mathcal{L}_{12}$ & $\mathcal{L}_{15}$ & $\mathcal{L}_{13}$ & $\mathcal{L}_{14}$ & $\mathcal{L}_{16}$ & $\mathcal{L}_{7}$ & $\mathcal{L}_{5}$ & $\mathcal{L}_{6}$ & $\mathcal{L}_{8}$ \\\hline
 $(0,3)$ & $\mathcal{L}_{2}$ & $\mathcal{L}_{4}$ & $\mathcal{L}_{3}$ & $\mathcal{L}_{1}$ & $\mathcal{L}_{10}$ & $\mathcal{L}_{12}$ & $\mathcal{L}_{11}$ & $\mathcal{L}_{9}$ & $\mathcal{L}_{14}$ & $\mathcal{L}_{16}$ & $\mathcal{L}_{15}$ & $\mathcal{L}_{13}$ & $\mathcal{L}_{6}$ & $\mathcal{L}_{8}$ & $\mathcal{L}_{7}$ & $\mathcal{L}_{5}$ \\\noalign{\hrule height 1pt}
 $(1,0)$ & $\mathcal{L}_{13}$ & $\mathcal{L}_{15}$ & $\mathcal{L}_{16}$ & $\mathcal{L}_{14}$ & $\mathcal{L}_{5}$ & $\mathcal{L}_{7}$ & $\mathcal{L}_{8}$ & $\mathcal{L}_{6}$ & $\mathcal{L}_{1}$ & $\mathcal{L}_{3}$ & $\mathcal{L}_{4}$ & $\mathcal{L}_{2}$ & $\mathcal{L}_{9}$ & $\mathcal{L}_{11}$ & $\mathcal{L}_{12}$ & $\mathcal{L}_{10}$ \\\hline
 $(1,1)$ & $\mathcal{L}_{16}$ & $\mathcal{L}_{14}$ & $\mathcal{L}_{13}$ & $\mathcal{L}_{15}$ & $\mathcal{L}_{8}$ & $\mathcal{L}_{6}$ & $\mathcal{L}_{5}$ & $\mathcal{L}_{7}$ & $\mathcal{L}_{4}$ & $\mathcal{L}_{2}$ & $\mathcal{L}_{1}$ & $\mathcal{L}_{3}$ & $\mathcal{L}_{12}$ & $\mathcal{L}_{10}$ & $\mathcal{L}_{9}$ & $\mathcal{L}_{11}$ \\\hline
 $(1,2)$ & $\mathcal{L}_{15}$ & $\mathcal{L}_{13}$ & $\mathcal{L}_{14}$ & $\mathcal{L}_{16}$ & $\mathcal{L}_{7}$ & $\mathcal{L}_{5}$ & $\mathcal{L}_{6}$ & $\mathcal{L}_{8}$ & $\mathcal{L}_{3}$ & $\mathcal{L}_{1}$ & $\mathcal{L}_{2}$ & $\mathcal{L}_{4}$ & $\mathcal{L}_{11}$ & $\mathcal{L}_{9}$ & $\mathcal{L}_{10}$ & $\mathcal{L}_{12}$ \\\hline
 $(1,3)$ & $\mathcal{L}_{14}$ & $\mathcal{L}_{16}$ & $\mathcal{L}_{15}$ & $\mathcal{L}_{13}$ & $\mathcal{L}_{6}$ & $\mathcal{L}_{8}$ & $\mathcal{L}_{7}$ & $\mathcal{L}_{5}$ & $\mathcal{L}_{2}$ & $\mathcal{L}_{4}$ & $\mathcal{L}_{3}$ & $\mathcal{L}_{1}$ & $\mathcal{L}_{10}$ & $\mathcal{L}_{12}$ & $\mathcal{L}_{11}$ & $\mathcal{L}_{9}$ \\\noalign{\hrule height 1pt}
 $(2,0)$ & $\mathcal{L}_{9}$ & $\mathcal{L}_{11}$ & $\mathcal{L}_{12}$ & $\mathcal{L}_{10}$ & $\mathcal{L}_{1}$ & $\mathcal{L}_{3}$ & $\mathcal{L}_{4}$ & $\mathcal{L}_{2}$ & $\mathcal{L}_{5}$ & $\mathcal{L}_{7}$ & $\mathcal{L}_{8}$ & $\mathcal{L}_{6}$ & $\mathcal{L}_{13}$ & $\mathcal{L}_{15}$ & $\mathcal{L}_{16}$ & $\mathcal{L}_{14}$ \\\hline
 $(2,1)$ & $\mathcal{L}_{12}$ & $\mathcal{L}_{10}$ & $\mathcal{L}_{9}$ & $\mathcal{L}_{11}$ & $\mathcal{L}_{4}$ & $\mathcal{L}_{2}$ & $\mathcal{L}_{1}$ & $\mathcal{L}_{3}$ & $\mathcal{L}_{8}$ & $\mathcal{L}_{6}$ & $\mathcal{L}_{5}$ & $\mathcal{L}_{7}$ & $\mathcal{L}_{16}$ & $\mathcal{L}_{14}$ & $\mathcal{L}_{13}$ & $\mathcal{L}_{15}$ \\\hline
 $(2,2)$ & $\mathcal{L}_{11}$ & $\mathcal{L}_{9}$ & $\mathcal{L}_{10}$ & $\mathcal{L}_{12}$ & $\mathcal{L}_{3}$ & $\mathcal{L}_{1}$ & $\mathcal{L}_{2}$ & $\mathcal{L}_{4}$ & $\mathcal{L}_{7}$ & $\mathcal{L}_{5}$ & $\mathcal{L}_{6}$ & $\mathcal{L}_{8}$ & $\mathcal{L}_{15}$ & $\mathcal{L}_{13}$ & $\mathcal{L}_{14}$ & $\mathcal{L}_{16}$ \\\hline
 $(2,3)$ & $\mathcal{L}_{10}$ & $\mathcal{L}_{12}$ & $\mathcal{L}_{11}$ & $\mathcal{L}_{9}$ & $\mathcal{L}_{2}$ & $\mathcal{L}_{4}$ & $\mathcal{L}_{3}$ & $\mathcal{L}_{1}$ & $\mathcal{L}_{6}$ & $\mathcal{L}_{8}$ & $\mathcal{L}_{7}$ & $\mathcal{L}_{5}$ & $\mathcal{L}_{14}$ & $\mathcal{L}_{16}$ & $\mathcal{L}_{15}$ & $\mathcal{L}_{13}$ \\\noalign{\hrule height 1pt}
 $(3,0)$ & $\mathcal{L}_{5}$ & $\mathcal{L}_{7}$ & $\mathcal{L}_{8}$ & $\mathcal{L}_{6}$ & $\mathcal{L}_{13}$ & $\mathcal{L}_{15}$ & $\mathcal{L}_{16}$ & $\mathcal{L}_{14}$ & $\mathcal{L}_{9}$ & $\mathcal{L}_{11}$ & $\mathcal{L}_{12}$ & $\mathcal{L}_{10}$ & $\mathcal{L}_{1}$ & $\mathcal{L}_{3}$ & $\mathcal{L}_{4}$ & $\mathcal{L}_{2}$ \\\hline
 $(3,1)$ & $\mathcal{L}_{8}$ & $\mathcal{L}_{6}$ & $\mathcal{L}_{5}$ & $\mathcal{L}_{7}$ & $\mathcal{L}_{16}$ & $\mathcal{L}_{14}$ & $\mathcal{L}_{13}$ & $\mathcal{L}_{15}$ & $\mathcal{L}_{12}$ & $\mathcal{L}_{10}$ & $\mathcal{L}_{9}$ & $\mathcal{L}_{11}$ & $\mathcal{L}_{4}$ & $\mathcal{L}_{2}$ & $\mathcal{L}_{1}$ & $\mathcal{L}_{3}$ \\\hline
 $(3,2)$ & $\mathcal{L}_{7}$ & $\mathcal{L}_{5}$ & $\mathcal{L}_{6}$ & $\mathcal{L}_{8}$ & $\mathcal{L}_{15}$ & $\mathcal{L}_{13}$ & $\mathcal{L}_{14}$ & $\mathcal{L}_{16}$ & $\mathcal{L}_{11}$ & $\mathcal{L}_{9}$ & $\mathcal{L}_{10}$ & $\mathcal{L}_{12}$ & $\mathcal{L}_{3}$ & $\mathcal{L}_{1}$ & $\mathcal{L}_{2}$ & $\mathcal{L}_{4}$ \\\hline
 $(3,3)$ & $\mathcal{L}_{6}$ & $\mathcal{L}_{8}$ & $\mathcal{L}_{7}$ & $\mathcal{L}_{5}$ & $\mathcal{L}_{14}$ & $\mathcal{L}_{16}$ & $\mathcal{L}_{15}$ & $\mathcal{L}_{13}$ & $\mathcal{L}_{10}$ & $\mathcal{L}_{12}$ & $\mathcal{L}_{11}$ & $\mathcal{L}_{9}$ & $\mathcal{L}_{2}$ & $\mathcal{L}_{4}$ & $\mathcal{L}_{3}$ & $\mathcal{L}_{1}$ \\\noalign{\hrule height 1pt}

\end{tabular}
}
}
\caption{Latin Square $L$ representing the clustering at the relay for the case $\gamma e^{j \theta}=j$, obtained using Cartesian Cluster Product, with the 4-PSK symbols that A(B) sent in the first and second channel use along the rows(columns)}
\label{ls1}
\end{figure*}

The entries of the above clusters are of the form $ ((x_{A_{1}},x_{B_{1}}),(x_{A_{2}},x_{B_{2}}))$, i.e., in the order A's transmission during the first channel use, B's transmission during the first channel use, A's transmission during the second channel use, B's transmission during the second channel use. We now represent these clusters by a Latin Square of side $16$, with $ (x_{A_{1}},x_{A_{2}})$ along the rows, and $(x_{B_{1}},x_{B_{2}})$ along the columns. The $((x_{A_{1}},x_{A_{2}}),(x_{B_{1}},x_{B_{2}}))$ entry of the Latin Square as dictated by the clusters above are as follows:

{\footnotesize
\begin{align}
\nonumber
\vspace{-0.5cm}
\mathcal{L}_{1}:=&\left\{((0,0),(0,0)), ((0,1),(0,2)), ((0,2),(0,1)), ((0,3),(0,3)), \right.\\
\nonumber 
& \left. ((1,0),(2,0)), ((1,1),(2,2)), ((1,2),(2,1)), ((1,3),(2,3)),\right.\\
\nonumber 
& \left. ((2,0),(1,0)), ((2,1),(1,2)), ((2,2),(1,1)), ((2,3),(1,3)),\right.\\
\nonumber 
& \left.  ((3,0),(3,0)), ((3,1),(3,2)), ((3,2),(3,1)), ((3,3),(3,3))\right\}
\end{align}
\begin{align}
\nonumber
\mathcal{L}_{2}:=&\left\{((0,0),(0,3)), ((0,1),(0,1)), ((0,2),(0,2)), ((0,3),(0,0)), \right.\\
\nonumber 
& \left. ((1,0),(2,3)), ((1,1),(2,1)), ((1,2),(2,2)), ((1,3),(2,0)),\right.\\
\nonumber 
& \left. ((2,0),(1,3)), ((2,1),(1,1)), ((2,2),(1,2)), ((2,3),(1,0)),\right.\\
\nonumber 
& \left. ((3,0),(3,3)), ((3,1),(3,1)), ((3,2),(3,2)), ((3,3),(3,0))\right\}
\end{align}
\begin{align}
\nonumber
\mathcal{L}_{3}:=&\left\{((0,0),(0,1)), ((0,1),(0,3)), ((0,2),(0,0)), ((0,3),(0,2)), \right.\\
\nonumber 
& \left. ((1,0),(2,1)), ((1,1),(2,3)), ((1,2),(2,0)), ((1,3),(2,2)),\right.\\
\nonumber 
& \left. ((2,0),(1,1)), ((2,1),(1,3)), ((2,2),(1,0)), ((2,3),(1,2)),\right.\\
\nonumber 
& \left. ((3,0),(3,1)), ((3,1),(3,3)), ((3,2),(3,0)), ((3,3),(3,2))\right\}
\end{align}
\begin{align}
\nonumber
\mathcal{L}_{4}:=&\left\{((0,0),(0,2)), ((0,1),(0,0)), ((0,2),(0,3)), ((0,3),(0,1)), \right.\\
\nonumber 
& \left. ((1,0),(2,2)), ((1,1),(2,0)), ((1,2),(2,3)), ((1,3),(2,1)),\right.\\
\nonumber 
& \left. ((2,0),(1,2)), ((2,1),(1,0)), ((2,2),(1,3)), ((2,3),(1,1)),\right.\\
\nonumber 
& \left. ((3,0),(3,2)), ((3,1),(3,0)), ((3,2),(3,3)), ((3,3),(3,1))\right\}
\end{align}
\begin{align}
\nonumber
\mathcal{L}_{5}:=&\left\{((0,0),(3,0)), ((0,1),(3,2)), ((0,2),(3,1)), ((0,3),(3,3)), \right.\\
\nonumber 
& \left. ((1,0),(1,0)), ((1,1),(1,2)), ((1,2),(1,1)), ((1,3),(1,3)),\right.\\
\nonumber 
& \left. ((2,0),(2,0)), ((2,1),(2,2)), ((2,2),(2,1)), ((2,3),(2,3)),\right.\\
\nonumber 
& \left. ((3,0),(0,0)), ((3,1),(0,2)), ((3,2),(0,1)), ((3,3),(0,3))\right\}
\end{align}
\begin{align}
\nonumber
\mathcal{L}_{6}:=&\left\{((0,0),(3,3)), ((0,1),(3,1)), ((0,2),(3,2)), ((0,3),(3,0)), \right.\\
\nonumber 
& \left. ((1,0),(1,3)), ((1,1),(1,1)), ((1,2),(1,2)), ((1,3),(1,0)),\right.\\
\nonumber 
& \left. ((2,0),(2,3)), ((2,1),(2,1)), ((2,2),(2,2)), ((2,3),(2,0)),\right.\\
\nonumber 
& \left. ((3,0),(0,3)), ((3,1),(0,1)), ((3,2),(0,2)), ((3,3),(0,0))\right\}
\end{align}
\begin{align}
\nonumber
\mathcal{L}_{7}:=&\left\{((0,0),(3,1)), ((0,1),(3,3)), ((0,2),(3,0)), ((0,3),(3,2)), \right.\\
\nonumber 
& \left. ((1,0),(1,1)), ((1,1),(1,3)), ((1,2),(1,0)), ((1,3),(1,2)),\right.\\
\nonumber 
& \left. ((2,0),(2,1)), ((2,1),(2,3)), ((2,2),(2,0)), ((2,3),(2,2)),\right.\\
\nonumber 
& \left. ((3,0),(0,1)), ((3,1),(0,3)), ((3,2),(0,0)), ((3,3),(0,2))\right\}
\end{align}
\begin{align}
\nonumber
\mathcal{L}_{8}:=&\left\{((0,0),(3,2)), ((0,1),(3,0)), ((0,2),(3,3)), ((0,3),(3,1)), \right.\\
\nonumber 
& \left. ((1,0),(1,2)), ((1,1),(1,0)), ((1,2),(1,3)), ((1,3),(1,1)),\right.\\
\nonumber 
& \left. ((2,0),(2,2)), ((2,1),(2,0)), ((2,2),(2,3)), ((2,3),(2,1)),\right.\\
\nonumber 
& \left. ((3,0),(0,2)), ((3,1),(0,0)), ((3,2),(0,3)), ((3,3),(0,1))\right\}
\end{align}
\begin{align}
\nonumber
\mathcal{L}_{9}:=&\left\{((0,0),(1,0)), ((0,1),(1,2)), ((0,2),(1,1)), ((0,3),(1,3)), \right.\\
\nonumber 
& \left. ((1,0),(3,0)), ((1,1),(3,2)), ((1,2),(3,1)), ((1,3),(3,3)),\right.\\
\nonumber 
& \left. ((2,0),(0,0)), ((2,1),(0,2)), ((2,2),(0,1)), ((2,3),(0,3)),\right.\\
\nonumber 
& \left. ((3,0),(2,0)), ((3,1),(2,2)), ((3,2),(2,1)), ((3,3),(2,3))\right\}
\end{align}
\begin{align}
\nonumber
\mathcal{L}_{10}:=&\left\{((0,0),(1,3)), ((0,1),(1,1)), ((0,2),(1,2)), ((0,3),(1,0)), \right.\\
\nonumber 
& \left. ((1,0),(3,3)), ((1,1),(3,1)), ((1,2),(3,2)), ((1,3),(3,0)),\right.\\
\nonumber 
& \left. ((2,0),(0,3)), ((2,1),(0,1)), ((2,2),(0,2)), ((2,3),(0,0)),\right.\\
\nonumber 
& \left. ((3,0),(2,3)), ((3,1),(2,1)), ((3,2),(2,2)), ((3,3),(2,0))\right\}
\end{align}
\begin{align}
\nonumber
\mathcal{L}_{11}:=&\left\{((0,0),(1,1)), ((0,1),(1,3)), ((0,2),(1,0)), ((0,3),(1,2)), \right.\\
\nonumber 
& \left. ((1,0),(3,1)), ((1,1),(3,3)), ((1,2),(3,0)), ((1,3),(3,2)),\right.\\
\nonumber 
& \left. ((2,0),(0,1)), ((2,1),(0,3)), ((2,2),(0,0)), ((2,3),(0,2)),\right.\\
\nonumber 
& \left. ((3,0),(2,1)), ((3,1),(2,3)), ((3,2),(2,0)), ((3,3),(2,2))\right\}
\end{align}
\begin{align}
\nonumber
\mathcal{L}_{12}:=&\left\{((0,0),(1,2)), ((0,1),(1,0)), ((0,2),(1,3)), ((0,3),(1,1)), \right.\\
\nonumber 
& \left. ((1,0),(3,2)), ((1,1),(3,0)), ((1,2),(3,3)), ((1,3),(3,1)),\right.\\
\nonumber 
& \left. ((2,0),(0,2)), ((2,1),(0,0)), ((2,2),(0,3)), ((2,3),(0,1)),\right.\\
\nonumber 
& \left. ((3,0),(2,2)), ((3,1),(2,0)), ((3,2),(2,3)), ((3,3),(2,1))\right\}
\end{align}
\begin{align}
\nonumber
\mathcal{L}_{13}:=&\left\{((0,0),(2,0)), ((0,1),(2,2)), ((0,2),(2,1)), ((0,3),(2,3)), \right.\\
\nonumber 
& \left. ((1,0),(0,0)), ((1,1),(0,2)), ((1,2),(0,1)), ((1,3),(0,3)),\right.\\
\nonumber 
& \left. ((2,0),(3,0)), ((2,1),(3,2)), ((2,2),(3,1)), ((2,3),(3,3)),\right.\\
\nonumber 
& \left. ((3,0),(1,0)), ((3,1),(1,2)), ((3,2),(1,1)), ((3,3),(1,3))\right\}
\end{align}
\begin{align}
\nonumber
\mathcal{L}_{14}:=&\left\{((0,0),(2,3)), ((0,1),(2,1)), ((0,2),(2,2)), ((0,3),(2,0)), \right.\\
\nonumber 
& \left. ((1,0),(0,3)), ((1,1),(0,1)), ((1,2),(0,2)), ((1,3),(0,0)),\right.\\
\nonumber 
& \left. ((2,0),(3,3)), ((2,1),(3,1)), ((2,2),(3,2)), ((2,3),(3,0)),\right.\\
\nonumber 
& \left. ((3,0),(1,3)), ((3,1),(1,1)), ((3,2),(1,2)), ((3,3),(1,0))\right\}
\end{align}
\begin{align}
\nonumber
\mathcal{L}_{15}:=&\left\{((0,0),(2,1)), ((0,1),(2,3)), ((0,2),(2,0)), ((0,3),(2,2)), \right.\\
\nonumber 
& \left. ((1,0),(0,1)), ((1,1),(0,3)), ((1,2),(0,0)), ((1,3),(0,2)),\right.\\
\nonumber 
& \left. ((2,0),(3,1)), ((2,1),(3,3)), ((2,2),(3,0)), ((2,3),(3,2)),\right.\\
\nonumber 
& \left. ((3,0),(1,1)), ((3,1),(1,3)), ((3,2),(1,0)), ((3,3),(1,2))\right\}
\end{align}
\begin{align}
\nonumber
\mathcal{L}_{16}:=&\left\{((0,0),(2,2)), ((0,1),(2,0)), ((0,2),(2,3)), ((0,3),(2,1)), \right.\\
\nonumber 
& \left. ((1,0),(0,2)), ((1,1),(0,0)), ((1,2),(0,3)), ((1,3),(0,1)),\right.\\
\nonumber 
& \left. ((2,0),(3,2)), ((2,1),(3,0)), ((2,2),(3,3)), ((2,3),(3,1)),\right.\\
\nonumber 
& \left. ((3,0),(1,2)), ((3,1),(1,0)), ((3,2),(1,3)), ((3,3),(1,1)) \right\}
\nonumber
\end{align}
}


\begin{figure}[h]
\centering
\subfigure[$4 \times 4$ blocks in $L$]{
{
\begin{tabular}{|c|c|c|c|c|}
\hline $\:$    & 0       & 1       &  2     & 3\\
\hline 0    & $l_{1}$ & $l_{3}$ & $l_{4}$ & $l_{2}$\\ 
\hline 1    & $l_{4}$ & $l_{2}$ & $l_{1}$ & $l_{3}$\\ 
\hline 2    & $l_{3}$ & $l_{1}$ & $l_{2}$ & $l_{4}$\\ 
\hline 3    & $l_{2}$ & $l_{4}$ & $l_{3}$ & $l_{1}$\\ 
\hline 
\end{tabular}}
\label{L1}
}
\subfigure[The array $L_B$]{
$\left[ {\begin{array}{cccc}
\alpha_1 &  \alpha_3 & \alpha_4 & \alpha_2 \\
\alpha_4 & \alpha_2 & \alpha_1  & \alpha_3  \\
\alpha_3 & \alpha_1 & \alpha_2  & \alpha_4  \\
\alpha_2 & \alpha_4 & \alpha_3  & \alpha_1  \\
\end{array} } \right]$ 
\label{fig:L_B}

}
\caption[]{Latin Square representing the clustering $\mathcal{C}^{\left[j\right]}$ with the symbol sent by A(B) along the rows(columns).}
\end{figure}

\begin{figure*}
{\footnotesize
{
\renewcommand{\arraystretch}{1,3}
\begin{tabular}{!{\vrule width 1pt}c!{\vrule width 1pt}c|c|c|c!{\vrule width 1pt}c|c|c|c!{\vrule width 1pt}c|c|c|c!{\vrule width 1pt}c|c|c|c!{\vrule width 1pt}} \noalign{\hrule height 1pt}
    &$(0,0)$&$(0,1)$&$(0,2)$&$(0,3)$&$(1,0)$&$(1,1)$&$(1,2)$&$(1,3)$&$(2,0)$&$(2,1)$&$(2,2)$&$(2,3)$&$(3,0)$&$(3,1)$&$(3,2)$&$(3,3)$ \\\noalign{\hrule height 1pt}

 $(0,0)$ & $\mathcal{L}_{25}$ & $\mathcal{L}_{21}$ & $\mathcal{L}_{22}$ & $\mathcal{L}_{23}$ & $\mathcal{L}_{5}$ & $\mathcal{L}_{1}$ & $\mathcal{L}_{2}$ & $\mathcal{L}_{3}$ & $\mathcal{L}_{10}$ & $\mathcal{L}_{6}$ & $\mathcal{L}_{7}$ & $\mathcal{L}_{8}$ & $\mathcal{L}_{15}$ & $\mathcal{L}_{11}$ & $\mathcal{L}_{12}$ & $\mathcal{L}_{13}$ \\\hline 
 $(0,1)$ & $\mathcal{L}_{24}$ & $\mathcal{L}_{25}$ & $\mathcal{L}_{21}$ & $\mathcal{L}_{22}$ & $\mathcal{L}_{4}$ & $\mathcal{L}_{5}$ & $\mathcal{L}_{1}$ & $\mathcal{L}_{2}$ & $\mathcal{L}_{9}$ & $\mathcal{L}_{10}$ & $\mathcal{L}_{6}$ & $\mathcal{L}_{7}$ & $\mathcal{L}_{14}$ & $\mathcal{L}_{15}$ & $\mathcal{L}_{11}$ & $\mathcal{L}_{12}$ \\\hline
 $(0,2)$ & $\mathcal{L}_{23}$ & $\mathcal{L}_{24}$ & $\mathcal{L}_{25}$ & $\mathcal{L}_{21}$ & $\mathcal{L}_{3}$ & $\mathcal{L}_{4}$ & $\mathcal{L}_{5}$ & $\mathcal{L}_{1}$ & $\mathcal{L}_{8}$ & $\mathcal{L}_{9}$ & $\mathcal{L}_{10}$ & $\mathcal{L}_{6}$ & $\mathcal{L}_{13}$ & $\mathcal{L}_{14}$ & $\mathcal{L}_{15}$ & $\mathcal{L}_{11}$ \\\hline
 $(0,3)$ & $\mathcal{L}_{22}$ & $\mathcal{L}_{23}$ & $\mathcal{L}_{24}$ & $\mathcal{L}_{25}$ & $\mathcal{L}_{2}$ & $\mathcal{L}_{3}$ & $\mathcal{L}_{4}$ & $\mathcal{L}_{5}$ & $\mathcal{L}_{7}$ & $\mathcal{L}_{8}$ & $\mathcal{L}_{9}$ & $\mathcal{L}_{10}$ & $\mathcal{L}_{12}$ & $\mathcal{L}_{13}$ & $\mathcal{L}_{14}$ & $\mathcal{L}_{15}$ \\\noalign{\hrule height 1pt}
 $(1,0)$ & $\mathcal{L}_{20}$ & $\mathcal{L}_{16}$ & $\mathcal{L}_{17}$ & $\mathcal{L}_{18}$ & $\mathcal{L}_{25}$ & $\mathcal{L}_{21}$ & $\mathcal{L}_{22}$ & $\mathcal{L}_{23}$ & $\mathcal{L}_{5}$ & $\mathcal{L}_{1}$ & $\mathcal{L}_{2}$ & $\mathcal{L}_{3}$ & $\mathcal{L}_{10}$ & $\mathcal{L}_{6}$ & $\mathcal{L}_{7}$ & $\mathcal{L}_{8}$ \\\hline
 $(1,1)$ & $\mathcal{L}_{19}$ & $\mathcal{L}_{20}$ & $\mathcal{L}_{16}$ & $\mathcal{L}_{17}$ & $\mathcal{L}_{24}$ & $\mathcal{L}_{25}$ & $\mathcal{L}_{21}$ & $\mathcal{L}_{22}$ & $\mathcal{L}_{4}$ & $\mathcal{L}_{5}$ & $\mathcal{L}_{1}$ & $\mathcal{L}_{2}$ & $\mathcal{L}_{9}$ & $\mathcal{L}_{10}$ & $\mathcal{L}_{6}$ & $\mathcal{L}_{7}$ \\\hline
 $(1,2)$ & $\mathcal{L}_{18}$ & $\mathcal{L}_{19}$ & $\mathcal{L}_{20}$ & $\mathcal{L}_{16}$ & $\mathcal{L}_{23}$ & $\mathcal{L}_{24}$ & $\mathcal{L}_{25}$ & $\mathcal{L}_{21}$ & $\mathcal{L}_{3}$ & $\mathcal{L}_{4}$ & $\mathcal{L}_{5}$ & $\mathcal{L}_{1}$ & $\mathcal{L}_{8}$ & $\mathcal{L}_{9}$ & $\mathcal{L}_{10}$ & $\mathcal{L}_{6}$ \\\hline
 $(1,3)$ & $\mathcal{L}_{17}$ & $\mathcal{L}_{18}$ & $\mathcal{L}_{19}$ & $\mathcal{L}_{20}$ & $\mathcal{L}_{22}$ & $\mathcal{L}_{23}$ & $\mathcal{L}_{24}$ & $\mathcal{L}_{25}$ & $\mathcal{L}_{2}$ & $\mathcal{L}_{3}$ & $\mathcal{L}_{4}$ & $\mathcal{L}_{5}$ & $\mathcal{L}_{7}$ & $\mathcal{L}_{8}$ & $\mathcal{L}_{9}$ & $\mathcal{L}_{10}$ \\\noalign{\hrule height 1pt}
 $(2,0)$ & $\mathcal{L}_{15}$ & $\mathcal{L}_{11}$ & $\mathcal{L}_{12}$ & $\mathcal{L}_{13}$ & $\mathcal{L}_{20}$ & $\mathcal{L}_{16}$ & $\mathcal{L}_{17}$ & $\mathcal{L}_{18}$ & $\mathcal{L}_{25}$ & $\mathcal{L}_{21}$ & $\mathcal{L}_{22}$ & $\mathcal{L}_{23}$ & $\mathcal{L}_{5}$ & $\mathcal{L}_{1}$ & $\mathcal{L}_{2}$ & $\mathcal{L}_{3}$ \\\hline
 $(2,1)$ & $\mathcal{L}_{14}$ & $\mathcal{L}_{15}$ & $\mathcal{L}_{11}$ & $\mathcal{L}_{12}$ & $\mathcal{L}_{19}$ & $\mathcal{L}_{20}$ & $\mathcal{L}_{16}$ & $\mathcal{L}_{17}$ & $\mathcal{L}_{24}$ & $\mathcal{L}_{25}$ & $\mathcal{L}_{21}$ & $\mathcal{L}_{22}$ & $\mathcal{L}_{4}$ & $\mathcal{L}_{5}$ & $\mathcal{L}_{1}$ & $\mathcal{L}_{2}$ \\\hline
 $(2,2)$ & $\mathcal{L}_{13}$ & $\mathcal{L}_{14}$ & $\mathcal{L}_{15}$ & $\mathcal{L}_{11}$ & $\mathcal{L}_{18}$ & $\mathcal{L}_{19}$ & $\mathcal{L}_{20}$ & $\mathcal{L}_{16}$ & $\mathcal{L}_{23}$ & $\mathcal{L}_{24}$ & $\mathcal{L}_{25}$ & $\mathcal{L}_{21}$ & $\mathcal{L}_{3}$ & $\mathcal{L}_{4}$ & $\mathcal{L}_{5}$ & $\mathcal{L}_{1}$ \\\hline
 $(2,3)$ & $\mathcal{L}_{12}$ & $\mathcal{L}_{13}$ & $\mathcal{L}_{14}$ & $\mathcal{L}_{15}$ & $\mathcal{L}_{17}$ & $\mathcal{L}_{18}$ & $\mathcal{L}_{19}$ & $\mathcal{L}_{20}$ & $\mathcal{L}_{22}$ & $\mathcal{L}_{23}$ & $\mathcal{L}_{24}$ & $\mathcal{L}_{25}$ & $\mathcal{L}_{2}$ & $\mathcal{L}_{3}$ & $\mathcal{L}_{4}$ & $\mathcal{L}_{5}$ \\\noalign{\hrule height 1pt}
 $(3,0)$ & $\mathcal{L}_{10}$ & $\mathcal{L}_{6}$ & $\mathcal{L}_{7}$ & $\mathcal{L}_{8}$ & $\mathcal{L}_{15}$ & $\mathcal{L}_{11}$ & $\mathcal{L}_{12}$ & $\mathcal{L}_{13}$ & $\mathcal{L}_{20}$ & $\mathcal{L}_{16}$ & $\mathcal{L}_{17}$ & $\mathcal{L}_{18}$ & $\mathcal{L}_{25}$ & $\mathcal{L}_{21}$ & $\mathcal{L}_{22}$ & $\mathcal{L}_{23}$ \\\hline
 $(3,1)$ & $\mathcal{L}_{9}$ & $\mathcal{L}_{10}$ & $\mathcal{L}_{6}$ & $\mathcal{L}_{7}$ & $\mathcal{L}_{14}$ & $\mathcal{L}_{15}$ & $\mathcal{L}_{11}$ & $\mathcal{L}_{12}$ & $\mathcal{L}_{19}$ & $\mathcal{L}_{20}$ & $\mathcal{L}_{16}$ & $\mathcal{L}_{17}$ & $\mathcal{L}_{24}$ & $\mathcal{L}_{25}$ & $\mathcal{L}_{21}$ & $\mathcal{L}_{22}$ \\\hline
 $(3,2)$ & $\mathcal{L}_{8}$ & $\mathcal{L}_{9}$ & $\mathcal{L}_{10}$ & $\mathcal{L}_{6}$ & $\mathcal{L}_{13}$ & $\mathcal{L}_{14}$ & $\mathcal{L}_{15}$ & $\mathcal{L}_{11}$ & $\mathcal{L}_{18}$ & $\mathcal{L}_{19}$ & $\mathcal{L}_{20}$ & $\mathcal{L}_{16}$ & $\mathcal{L}_{23}$ & $\mathcal{L}_{24}$ & $\mathcal{L}_{25}$ & $\mathcal{L}_{21}$ \\\hline
 $(3,3)$ & $\mathcal{L}_{7}$ & $\mathcal{L}_{8}$ & $\mathcal{L}_{9}$ & $\mathcal{L}_{10}$ & $\mathcal{L}_{12}$ & $\mathcal{L}_{13}$ & $\mathcal{L}_{14}$ & $\mathcal{L}_{15}$ & $\mathcal{L}_{17}$ & $\mathcal{L}_{18}$ & $\mathcal{L}_{19}$ & $\mathcal{L}_{20}$ & $\mathcal{L}_{22}$ & $\mathcal{L}_{23}$ & $\mathcal{L}_{24}$ & $\mathcal{L}_{25}$ \\\noalign{\hrule height 1pt}

\end{tabular}
}
}
\caption{Latin Square representing the clustering at the relay for the case $\gamma e^{j \theta}=0.5+0.5j$, obtained using Cartesian Cluster Product, with the 4-PSK symbols that A(B) sent in the first and second channel use along the rows(columns)}
\label{ls2}
\end{figure*}

The resulting Latin Square representing the clusters denoted by say $L$, is shown in Fig. 3. This $16 \times 16$ array $L$ can be divided into 16 blocks of $4 \times 4$ arrays. Let $L_B=\left[L_{i,j}\right] $ where each $L_{i,j}$ is a $4\times 4$ array for $i,j=1,2,3,4$ as shown in Fig. 3. Each $L_{i,j}$ is in one-to-one correspondence with the Latin Square obtained in \cite{NMR} for removing the singular fade state $\gamma e^{j \theta}=j$ for the two-way 2-stage relaying scenario representing the clustering $\mathcal{C}^{\left[j\right]}$ given in Fig \ref{L1}. Also let, 
\begin{align*}
\alpha_1 &:= L_{1,1}=L_{2,3}=L_{3,2}=L_{4,4},\\ 
\alpha_2 &:= L_{1,4}=L_{2,2}=L_{3,3}=L_{4,1}, \\
\alpha_3 &:= L_{1,2}=L_{2,4}=L_{3,1}=L_{4,3} \quad \mathrm{and} \\
\alpha_4 &:= L_{1,3}=L_{2,1}=L_{3,4}=L_{4,2}.
\end{align*}
This makes $L_B$ of the form shown in Fig. \ref{fig:L_B}.

This makes the block matrix $L_B$ also consistent with the Latin Square given in Fig. 4. The reason behind $L_B$ and $L_{i,j}$ being consistent with this Latin Square is as follows: each $L_{i,j}$ corresponds to some fixed values of the symbols A and B send during the first channel use, with the symbols sent by A and B during second channel use varying along the rows and columns respectively. The Latin Square in Fig. 3 has been obtained by taking the Cartesian Product of the clustering for removing the fade state $\gamma e^{j \theta}=j$ with itself. The Cartesian Product utilizes the clustering that is represented by the Latin Square given in Fig. 4, given by $\mathcal{C}^{\left[ j \right]}$, for the case for both the first and second channel use in the MA phase, which makes each $L_{i,j} ~ i,j=1,2,3,4$ and $L_B$ in one-to-one correspondence with this Latin Square that represents the clustering $\mathcal{C}^{\left[j\right]}$. \\

\end{example}

\textbf{\textit{Case 2:} $\gamma e^{j \theta} $ lies on the circle of radius $ 1/\sqrt{2}$.}\\

In this case, the Cartesian Product of $\mathcal{C}^{\left[\gamma e^{j \theta}\right]}$ with itself consists of $25$ clusters since the clustering $\mathcal{C}^{\left[\gamma e^{j \theta}\right]}$ has $5$ clusters. We now give an example of this case. The remaining instances of this case can be obtained from this example as will be shown later in Section V, Lemma 6.\\

\begin{example} Consider the case when $\gamma e^{j \theta}=0.5+0.5j.$ The clustering $\mathcal{C}^{\left[0.5+0.5j\right]}$ given in \cite{NMR} that removes this fade state for the two-way 2-stage relaying scenario is given by:
\begin{align}
\nonumber
\vspace{-0.8cm}
\mathcal{C}^{\left[0.5+0.5j\right]}=&\left\{l_1,l_2,l_3,l_4,l_5\right\},\\
\text{where,}\\
\nonumber
&l_{1}=\left\{\left(0,1\right),\left(1,2\right),\left(2,3\right)\right\}\\
\nonumber
&l_{2}=\left\{\left(0,2\right),\left(1,3\right),\left(3,0\right)\right\}\\
\nonumber
&l_{3}=\left\{\left(0,3\right),\left(2,0\right),\left(3,1\right)\right\}\\
\nonumber
&l_{4}=\left\{\left(1,0\right),\left(2,1\right),\left(3,2\right)\right\}\\
\nonumber
&l_{5}=\left\{\left(0,0\right),\left(1,1\right),\left(2,2\right),\left(3,3\right)\right\}.
\nonumber
\end{align}

The Cartesian Product of the above clustering given by $\mathcal{D}^{\left[0.5+0.5j\right]}=\left\{ \mathcal{C}^{\left\{l_{i},l_{j}\right\}}~ |~ i,j=1,2,3,4,5\right\}$ contains exactly $25$ clusters. The clusters and the corresponding constraints for the Latin Square representing the clustering have been listed in the Appendix A. The Cartesian Product of the clustering $\mathcal{C}^{\left[0.5+0.5j\right]}$ with itself, denoted by $\mathcal{D}^{\left[0.5+0.5j\right]}$ can be represented by the Latin Square given in Fig. 5.\\

\begin{figure}[ht]
\centering
\renewcommand{\arraystretch}{1,3}
\begin{tabular}{!{\vrule width 1pt}c!{\vrule width 1pt}c|c|c|c!{\vrule width 1pt}}\noalign{\hrule height 1pt}
    &$0$&$1$&$2$&$3$\\\noalign{\hrule height 1pt}
 $0$ & $l_{5}$ & $l_{1}$ & $l_{2}$ & $l_{3}$\\\hline 
 $1$ & $l_{4}$ & $l_{5}$ & $l_{1}$ & $l_{2}$\\\hline
 $2$ & $l_{3}$ & $l_{4}$ & $l_{5}$ & $l_{1}$\\\hline
 $3$ & $l_{2}$ & $l_{3}$ & $l_{4}$ & $l_{5}$\\\noalign{\hrule height 1pt}
\end{tabular}
\caption{Latin Square $l$ representing the clustering $\mathcal{C}^{\left[0.5+0.5j\right]}$ with the symbol sent by A(B) along the rows(columns)}
\end{figure}

\begin{figure*}
{\footnotesize
{
\renewcommand{\arraystretch}{1,3}
\begin{tabular}{!{\vrule width 1pt}c!{\vrule width 1pt}c|c|c|c!{\vrule width 1pt}c|c|c|c!{\vrule width 1pt}c|c|c|c!{\vrule width 1pt}c|c|c|c!{\vrule width 1pt}} \noalign{\hrule height 1pt}
    &$(0,0)$&$(0,1)$&$(0,2)$&$(0,3)$&$(1,0)$&$(1,1)$&$(1,2)$&$(1,3)$&$(2,0)$&$(2,1)$&$(2,2)$&$(2,3)$&$(3,0)$&$(3,1)$&$(3,2)$&$(3,3)$ \\\noalign{\hrule height 1pt}

 $(0,0)$ & $\mathcal{L}_{25}$ & $\mathcal{L}_{21}$ & $\mathcal{L}_{24}$ & $\mathcal{L}_{22}$ & $\mathcal{L}_{5}$ & $\mathcal{L}_{1}$ & $\mathcal{L}_{4}$ & $\mathcal{L}_{2}$ & $\mathcal{L}_{20}$ & $\mathcal{L}_{16}$ & $\mathcal{L}_{19}$ & $\mathcal{L}_{17}$ & $\mathcal{L}_{10}$ & $\mathcal{L}_{6}$ & $\mathcal{L}_{9}$ & $\mathcal{L}_{7}$ \\\hline 
 $(0,1)$ & $\mathcal{L}_{22}$ & $\mathcal{L}_{25}$ & $\mathcal{L}_{23}$ & $\mathcal{L}_{24}$ & $\mathcal{L}_{2}$ & $\mathcal{L}_{5}$ & $\mathcal{L}_{3}$ & $\mathcal{L}_{4}$ & $\mathcal{L}_{17}$ & $\mathcal{L}_{20}$ & $\mathcal{L}_{18}$ & $\mathcal{L}_{19}$ & $\mathcal{L}_{7}$ & $\mathcal{L}_{10}$ & $\mathcal{L}_{8}$ & $\mathcal{L}_{9}$ \\\hline
 $(0,2)$ & $\mathcal{L}_{23}$ & $\mathcal{L}_{24}$ & $\mathcal{L}_{25}$ & $\mathcal{L}_{21}$ & $\mathcal{L}_{3}$ & $\mathcal{L}_{4}$ & $\mathcal{L}_{5}$ & $\mathcal{L}_{1}$ & $\mathcal{L}_{18}$ & $\mathcal{L}_{19}$ & $\mathcal{L}_{20}$ & $\mathcal{L}_{16}$ & $\mathcal{L}_{8}$ & $\mathcal{L}_{9}$ & $\mathcal{L}_{10}$ & $\mathcal{L}_{6}$ \\\hline
 $(0,3)$ & $\mathcal{L}_{21}$ & $\mathcal{L}_{23}$ & $\mathcal{L}_{22}$ & $\mathcal{L}_{25}$ & $\mathcal{L}_{1}$ & $\mathcal{L}_{3}$ & $\mathcal{L}_{2}$ & $\mathcal{L}_{5}$ & $\mathcal{L}_{16}$ & $\mathcal{L}_{18}$ & $\mathcal{L}_{17}$ & $\mathcal{L}_{20}$ & $\mathcal{L}_{6}$ & $\mathcal{L}_{8}$ & $\mathcal{L}_{7}$ & $\mathcal{L}_{10}$ \\\noalign{\hrule height 1pt}
 $(1,0)$ & $\mathcal{L}_{10}$ & $\mathcal{L}_{6}$ & $\mathcal{L}_{9}$ & $\mathcal{L}_{7}$ & $\mathcal{L}_{25}$ & $\mathcal{L}_{21}$ & $\mathcal{L}_{24}$ & $\mathcal{L}_{22}$ & $\mathcal{L}_{15}$ & $\mathcal{L}_{11}$ & $\mathcal{L}_{14}$ & $\mathcal{L}_{12}$ & $\mathcal{L}_{20}$ & $\mathcal{L}_{16}$ & $\mathcal{L}_{19}$ & $\mathcal{L}_{17}$ \\\hline
 $(1,1)$ & $\mathcal{L}_{7}$ & $\mathcal{L}_{10}$ & $\mathcal{L}_{8}$ & $\mathcal{L}_{9}$ & $\mathcal{L}_{22}$ & $\mathcal{L}_{25}$ & $\mathcal{L}_{23}$ & $\mathcal{L}_{24}$ & $\mathcal{L}_{12}$ & $\mathcal{L}_{15}$ & $\mathcal{L}_{13}$ & $\mathcal{L}_{14}$ & $\mathcal{L}_{17}$ & $\mathcal{L}_{20}$ & $\mathcal{L}_{18}$ & $\mathcal{L}_{19}$ \\\hline
 $(1,2)$ & $\mathcal{L}_{8}$ & $\mathcal{L}_{9}$ & $\mathcal{L}_{10}$ & $\mathcal{L}_{6}$ & $\mathcal{L}_{23}$ & $\mathcal{L}_{24}$ & $\mathcal{L}_{25}$ & $\mathcal{L}_{21}$ & $\mathcal{L}_{13}$ & $\mathcal{L}_{14}$ & $\mathcal{L}_{15}$ & $\mathcal{L}_{11}$ & $\mathcal{L}_{18}$ & $\mathcal{L}_{19}$ & $\mathcal{L}_{20}$ & $\mathcal{L}_{16}$ \\\hline
 $(1,3)$ & $\mathcal{L}_{6}$ & $\mathcal{L}_{8}$ & $\mathcal{L}_{7}$ & $\mathcal{L}_{10}$ & $\mathcal{L}_{21}$ & $\mathcal{L}_{23}$ & $\mathcal{L}_{22}$ & $\mathcal{L}_{25}$ & $\mathcal{L}_{11}$ & $\mathcal{L}_{13}$ & $\mathcal{L}_{12}$ & $\mathcal{L}_{15}$ & $\mathcal{L}_{16}$ & $\mathcal{L}_{18}$ & $\mathcal{L}_{17}$ & $\mathcal{L}_{20}$ \\\noalign{\hrule height 1pt}
 $(2,0)$ & $\mathcal{L}_{15}$ & $\mathcal{L}_{11}$ & $\mathcal{L}_{14}$ & $\mathcal{L}_{12}$ & $\mathcal{L}_{20}$ & $\mathcal{L}_{16}$ & $\mathcal{L}_{19}$ & $\mathcal{L}_{17}$ & $\mathcal{L}_{25}$ & $\mathcal{L}_{21}$ & $\mathcal{L}_{24}$ & $\mathcal{L}_{22}$ & $\mathcal{L}_{5}$ & $\mathcal{L}_{1}$ & $\mathcal{L}_{4}$ & $\mathcal{L}_{2}$ \\\hline
 $(2,1)$ & $\mathcal{L}_{12}$ & $\mathcal{L}_{15}$ & $\mathcal{L}_{13}$ & $\mathcal{L}_{14}$ & $\mathcal{L}_{17}$ & $\mathcal{L}_{20}$ & $\mathcal{L}_{18}$ & $\mathcal{L}_{19}$ & $\mathcal{L}_{22}$ & $\mathcal{L}_{25}$ & $\mathcal{L}_{23}$ & $\mathcal{L}_{24}$ & $\mathcal{L}_{2}$ & $\mathcal{L}_{5}$ & $\mathcal{L}_{3}$ & $\mathcal{L}_{4}$ \\\hline
 $(2,2)$ & $\mathcal{L}_{13}$ & $\mathcal{L}_{14}$ & $\mathcal{L}_{15}$ & $\mathcal{L}_{11}$ & $\mathcal{L}_{18}$ & $\mathcal{L}_{19}$ & $\mathcal{L}_{20}$ & $\mathcal{L}_{16}$ & $\mathcal{L}_{23}$ & $\mathcal{L}_{24}$ & $\mathcal{L}_{25}$ & $\mathcal{L}_{21}$ & $\mathcal{L}_{3}$ & $\mathcal{L}_{4}$ & $\mathcal{L}_{5}$ & $\mathcal{L}_{1}$ \\\hline
 $(2,3)$ & $\mathcal{L}_{11}$ & $\mathcal{L}_{13}$ & $\mathcal{L}_{12}$ & $\mathcal{L}_{15}$ & $\mathcal{L}_{16}$ & $\mathcal{L}_{18}$ & $\mathcal{L}_{17}$ & $\mathcal{L}_{20}$ & $\mathcal{L}_{21}$ & $\mathcal{L}_{23}$ & $\mathcal{L}_{22}$ & $\mathcal{L}_{25}$ & $\mathcal{L}_{1}$ & $\mathcal{L}_{3}$ & $\mathcal{L}_{2}$ & $\mathcal{L}_{5}$ \\\noalign{\hrule height 1pt}
 $(3,0)$ & $\mathcal{L}_{5}$ & $\mathcal{L}_{1}$ & $\mathcal{L}_{4}$ & $\mathcal{L}_{2}$ & $\mathcal{L}_{15}$ & $\mathcal{L}_{11}$ & $\mathcal{L}_{14}$ & $\mathcal{L}_{12}$ & $\mathcal{L}_{10}$ & $\mathcal{L}_{6}$ & $\mathcal{L}_{9}$ & $\mathcal{L}_{7}$ & $\mathcal{L}_{25}$ & $\mathcal{L}_{21}$ & $\mathcal{L}_{24}$ & $\mathcal{L}_{22}$ \\\hline
 $(3,1)$ & $\mathcal{L}_{2}$ & $\mathcal{L}_{5}$ & $\mathcal{L}_{3}$ & $\mathcal{L}_{4}$ & $\mathcal{L}_{12}$ & $\mathcal{L}_{15}$ & $\mathcal{L}_{13}$ & $\mathcal{L}_{14}$ & $\mathcal{L}_{7}$ & $\mathcal{L}_{10}$ & $\mathcal{L}_{8}$ & $\mathcal{L}_{9}$ & $\mathcal{L}_{22}$ & $\mathcal{L}_{25}$ & $\mathcal{L}_{23}$ & $\mathcal{L}_{24}$ \\\hline
 $(3,2)$ & $\mathcal{L}_{3}$ & $\mathcal{L}_{4}$ & $\mathcal{L}_{5}$ & $\mathcal{L}_{1}$ & $\mathcal{L}_{13}$ & $\mathcal{L}_{14}$ & $\mathcal{L}_{15}$ & $\mathcal{L}_{11}$ & $\mathcal{L}_{8}$ & $\mathcal{L}_{9}$ & $\mathcal{L}_{10}$ & $\mathcal{L}_{6}$ & $\mathcal{L}_{23}$ & $\mathcal{L}_{24}$ & $\mathcal{L}_{25}$ & $\mathcal{L}_{21}$ \\\hline
 $(3,3)$ & $\mathcal{L}_{1}$ & $\mathcal{L}_{3}$ & $\mathcal{L}_{2}$ & $\mathcal{L}_{5}$ & $\mathcal{L}_{11}$ & $\mathcal{L}_{13}$ & $\mathcal{L}_{12}$ & $\mathcal{L}_{15}$ & $\mathcal{L}_{6}$ & $\mathcal{L}_{8}$ & $\mathcal{L}_{7}$ & $\mathcal{L}_{10}$ & $\mathcal{L}_{21}$ & $\mathcal{L}_{23}$ & $\mathcal{L}_{22}$ & $\mathcal{L}_{25}$ \\\noalign{\hrule height 1pt}

\end{tabular}
}
}
\caption{Latin Square representing the clustering at the relay for the case $\gamma e^{j \theta}=1+j$, obtained using Cartesian Cluster Product, with the 4-PSK symbols that A(B) sent in the first and second channel use along the rows(columns)}
\label{ls2}
\end{figure*}

As explained in the previous example, let the $16 \times 16$ Latin Square obtained as shown in Fig. 5 be denoted by say $L'$ and $L'_{B}=\left[L'_{i,j}\right]$ with $i,j=1,2,3,4$. Then both $L'_{B}$ and $L'_{i,j}$ must be consistent with the Latin Square of side $4$ given in Fig. 6, denoted by $l$, which represents the clustering $\mathcal{C}^{\left[0.5+0.5j\right]}$. As can be seen in Fig. 5, $l$ is repeated in each block $L'_{i,j}$ with a possibly different set of five symbols amongst $\left\{\mathcal{L}_{1},\mathcal{L}_{2},..., \mathcal{L}_{25}\right\}$ denoting the five symbols $\left\{l_{1}, l_{2},...,l_{5}\right\}$ in each $L'_{i,j}$. More precisely, the blocks $ L'_{1,1}=L'_{2,2}=L'_{3,3}=L'_{4,4}$ are the same as $l$, with the symbols $\mathcal{L}_{21}, \mathcal{L}_{22}, ..., \mathcal{L}_{25}$ replacing the symbols $l_{1},l_{2},...,l_{5}$ respectively. Similarly, the blocks $L'_{1,2}=L'_{2,3}=L'_{3,4}$ are the same as $l$ with symbols $\mathcal{L}_{1}, \mathcal{L}_{2}, ..., \mathcal{L}_{5}$ replacing the symbols $l_{1},l_{2},...,l_{5}$ respectively, the blocks $L'_{1,3}=L'_{2,4}=L'_{4,1}$ are the same as $l$ with symbols $\mathcal{L}_{6}, \mathcal{L}_{7}, ..., \mathcal{L}_{10}$ replacing the symbols $l_{1},l_{2},...,l_{5}$ respectively, the blocks $L'_{1,4}=L'_{3,1}=L'_{4,2}$ are the same as $l$ with symbols $\mathcal{L}_{11}, \mathcal{L}_{12}, ..., \mathcal{L}_{15}$ replacing the symbols $l_{1},l_{2},...,l_{5}$ respectively, and the blocks $L'_{2,1}=L'_{3,2}=L'_{4,3}$ are the same as $l$ with symbols $\mathcal{L}_{16}, \mathcal{L}_{17}, ..., \mathcal{L}_{20}$ replacing the symbols $l_{1},l_{2},...,l_{5}$ respectively. Thus, the array $L'$ can be obtained using $l$ by simply using a different set of five symbols to denote $l_1,l_2,...,l_5$ for every set of blocks corresponding to a symbol amongst $l_1,l_2,...,l_5$ in $l$. We will illustrate this in the next example by obtaining the $16 \times 16$ Latin Square using the $4\times 4$ Latin Square given in \cite{NMR} for the case.\\
\end{example}

\textbf{\textit{Case 3:} $\gamma e^{j \theta} $ lies on the circle of radius $ \sqrt{2}$.}\\

In this case, the Cartesian Product of $\mathcal{C}^{\left[\gamma e^{j \theta}\right]}$ with itself consists of $25$ clusters since the clustering $\mathcal{C}^{\left[\gamma e^{j \theta}\right]}$ has $5$ clusters. An instance of this case is as follows.\\

\begin{example} Consider the case when $\gamma e^{j \theta}=1+j.$ The clustering $\mathcal{C}^{\left[1+j\right]}$ given in \cite{NMR} that removes this fade state for the two-way 2-stage relaying scenario is given by:
\begin{align}
\nonumber
\vspace{-0.8cm}
\mathcal{C}^{\left[1+j\right]}=&\left\{l_{1},l_{2},l_{3},l_{4},l_{5}\right\} \\
\nonumber
\text{where,}\\
\nonumber
&l_{1}=\left\{\left(0,1\right),\left(2,3\right),\left(3,0\right)\right\}\\
\nonumber
&l_{2}=\left\{\left(0,3\right),\left(1,0\right),\left(3,2\right)\right\}\\
\nonumber
&l_{3}=\left\{\left(1,2\right),\left(2,0\right),\left(3,1\right)\right\}\\
\nonumber
&l_{4}=\left\{\left(0,2\right),\left(1,3\right),\left(2,1\right)\right\}\\
\nonumber
&l_{5}=\left\{\left(0,0\right),\left(1,1\right),\left(2,2\right),\left(3,3\right)\right\}.
\nonumber
\end{align}

This clustering can be represented by a Latin Square of side $4$ denoted by $l'$ as shown in Fig. 8. \\

\begin{figure}[ht]
\centering
\renewcommand{\arraystretch}{1,3}
\begin{tabular}{!{\vrule width 1pt}c!{\vrule width 1pt}c|c|c|c!{\vrule width 1pt}}\noalign{\hrule height 1pt}
    &$0$&$1$&$2$&$3$\\\noalign{\hrule height 1pt}
 $0$ & $l_{5}$ & $l_{1}$ & $l_{4}$ & $l_{2}$\\\hline 
 $1$ & $l_{2}$ & $l_{5}$ & $l_{3}$ & $l_{4}$\\\hline
 $2$ & $l_{3}$ & $l_{4}$ & $l_{5}$ & $l_{1}$\\\hline
 $3$ & $l_{1}$ & $l_{3}$ & $l_{2}$ & $l_{5}$\\\noalign{\hrule height 1pt}
\end{tabular}
\caption{Latin Square representing the clustering $\mathcal{C}^{\left[1+j\right]}$ with the symbol sent by A(B) along the rows(columns)}
\end{figure}

The Cartesian Product of the above clustering given by $\mathcal{D}^{\left[1+j\right]}=\left\{ \mathcal{C}^{\left\{l_{i},l_{j}\right\}}~ |~ i,j=1,2,3,4,5\right\}$ contains exactly $25$ clusters as given in the Appendix B. We represent these clusters by a Latin Square of side $16$, with $ (x_{A_{1}},x_{A_{2}})$ along the rows, and $(x_{B_{1}},x_{B_{2}})$ along the columns. The $((x_{A_{1}},x_{A_{2}}),(x_{B_{1}},x_{B_{2}}))$ entry of the Latin Square as dictated by the clusters above are also listed in the Appendix B. 

Let the $16 \times 16$ Latin Square that represents the clustering obtained as a Cartesian Product of $\mathcal{C}^{\left[1+j\right]}$ with itself be denoted by say $L''$, and $L''_{B}=\left[L''_{i,j}\right]$ with $i,j=1,2,3,4$. Then each of $L''_{B}$ and $L''_{i,j}$ must be consistent with the Latin Square of side $4$ given in Fig. 8 which represents the clustering $\mathcal{C}^{\left[1+j\right]}$. We denote $l_1,l_2, ...,l_5$ in $L_{1,2}=L_{3,4}=L_{4,1}$ by $\mathcal{L}_{1}, \mathcal{L}_{2},...,\mathcal{L}_{5}$, in $L_{1,4}=L_{2,1}=L_{4,3}$ by $\mathcal{L}_{5}, \mathcal{L}_{6},...,\mathcal{L}_{10}$, in blocks $L_{2,3}=L_{3,1}=L_{3,2}$ by $\mathcal{L}_{10}, \mathcal{L}_{11},...,\mathcal{L}_{15}$, in blocks $L_{1,3}=L_{2,4}=L_{3,2}$ by $\mathcal{L}_{16}, \mathcal{L}_{17},...,\mathcal{L}_{20}$ and in blocks $L_{1,1}=L_{2,2}=L_{3,3}=L_{4,4}$ by $\mathcal{L}_{21}, \mathcal{L}_{22},...,\mathcal{L}_{25}$. Placing these blocks in accordance with $l'$, the Cartesian Product of the clustering $\mathcal{C}^{\left[1+j\right]}$ with itself, denoted by $\mathcal{D}^{\left[1+j\right]}$ can be represented by the Latin Square given in Fig. 7.\\

The Latin Square which removes the singular fade state $\frac{1}{\gamma}e^{-j\theta}$ can be obtained by taking the transpose of the Latin Square which removes the singular fade state ${\gamma}e^{j\theta}$. For example, the Latin Square which removes the singular fade state $\sqrt{2}e^{j\frac{\pi}{4}}$ can also be obtained by taking the transpose of the Latin Square which removes the singular fade state $\frac{1}{\sqrt{2}}e^{-j\frac{\pi}{4}}.$ The reason for this is as follows: the case when the singular fade state is $\frac{1}{\gamma}e^{-j\theta}$ can be equivalently viewed as the case when the singular fade states is  ${\gamma}e^{j\theta}$ with the users A and B interchanged. Interchanging the users is equivalent to taking transpose of the Latin Square.\\
\end{example}

\section{Clusterings from Latin Square of Lower Size}

In this section, we deal with any $2^ \lambda $-PSK constellation. In \cite{MNR}, it was shown that the singular fade states for the two-way 2-stage relaying scenario lie on circles centered at the origin. From Lemma 1, since the singular fade states for the ACF two-way relaying scenario are the same as that of the two-way 2-stage relaying scenario, it follows that the singular fade states for the ACF two-way relaying scenario lie on circles centered at the origin as well. In this section, it is shown that for each circle, it is enough if we obtain one Latin Square which removes a singular fade state on that circle. The Latin Squares which remove the other singular fade states on that circle can be obtained by some elementary operations on the Latin Square, which are described in the sequel.\\

 For a Latin Square $L$ of order $2^{2\lambda}$, let $L_{i,j}, 0 \leq i,j \leq 2^{\lambda}-1,$ denote the Latin Sub-square of order $2^{\lambda}$ obtained by taking only the rows $2^{\lambda}i$ to $2^{\lambda}(i+1)-1$ and only the columns $2^{\lambda}j$ to $2^{\lambda}(j+1)-1$ of $L$. For example, in Fig. \ref{Latin_subsquare_ex}, the Latin Sub-squares $L_{i,j}, i,j \in \lbrace 0,1 \rbrace,$ of order 2 corresponding to the Latin Square $L$ of order 4 are shown. Let $L_B$ denote the Square of order $2^\lambda$, associated with the Latin Square $L$ of order $2^{2\lambda}$, with $L_{i,j},0 \leq i,j \leq 2^{\lambda}-1,$ as its entries. For example the square $L_B$ of order 2 associated to a Latin Square $L$ of order 4 is as shown in Fig. \ref{L_B_example}.\\

 \begin{figure}[htbp]
\centering
\includegraphics[totalheight=1.25in,width=2.5in]{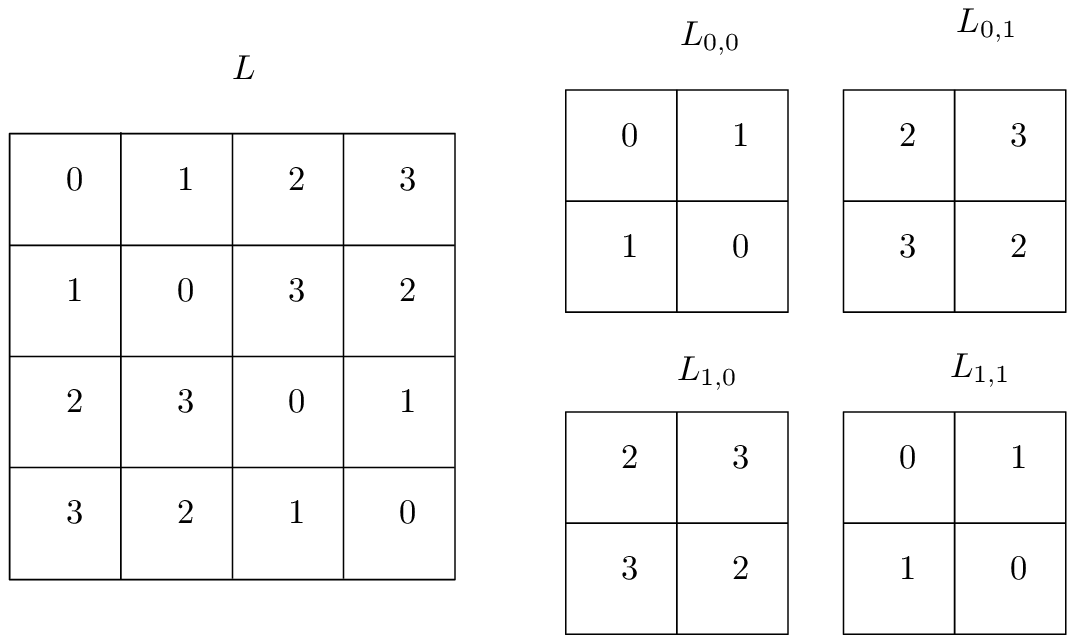}
\caption{Obtaining Latin Squares $L_{i,j}$'s from the Latin Square $L$.}     
\label{Latin_subsquare_ex}        
\end{figure}

  \begin{figure}[htbp]
\centering
\includegraphics[totalheight=1in,width=1in]{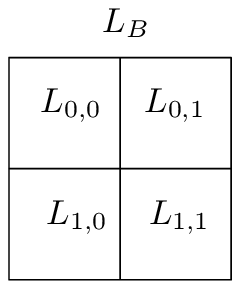}
\caption{The Square $L_B$ order 2 corresponding to a Latin Square $L$ of order 4.}     
\label{L_B_example}        
\end{figure}

 \begin{lemma}
 Consider the two-way ACF relaying with $2^\lambda$-PSK signal set used at nodes A and B. The Latin Square $L''$ of order $2^{2\lambda}$ which removes the singular fade state $(\gamma, \theta^{\prime})$, can be obtained from the Latin Square $L$ of order $2^{2\lambda}$ which removes the singular fade state $(\gamma, \theta)$, where $\theta'-\theta = k\frac{2\pi}{2^{\lambda}}$, as follows: Cyclic shift the columns of each one the $2^{2\lambda}$ Latin Squares $L_{i,j}, 0 \leq i,j \leq 2^{\lambda}-1,$ $k$ times to the left to get the Latin Square $L'$. Cyclic shift the columns of the Square $L'_{B}$ associated with $L'$, $k$ times to the left, to get the Square $L''_{B}$ associated with the Latin Square $L''.$\\

\begin{proof}
For the singular fade state $\gamma e^{j\theta}$, let {\footnotesize $\left\lbrace\left((x_{A_1},x_{A_2}),(x_{B_1},x_{B_2})\right),\left((x'_{A_1},x'_{A_2}),(x'_{B_1},x'_{B_2})\right)\right\rbrace$} be a singularity removal constraint , i.e.,\\
 \begin{equation}
 \label{eqn_sing}
 \gamma e^{j\theta}=\frac{x'_{A_1}-x_{A_1}}{x_{B_1}-x'_{B_1}}=\frac{x'_{A_2}-x_{A_2}}{x_{B_2}-x'_{B_2}}.
 \end{equation}

From \eqref{eqn_sing}, it follows that 

{\footnotesize 
\begin{align}
\nonumber
&\left\lbrace\left((x_{A_1},x_{A_2}),(x_{B_1} e^{-\frac{jk2\pi}{2^\lambda}},x_{B_2} e^{-\frac{jk2\pi}{2^\lambda}})\right),\right.\\
\nonumber
&\left.\hspace{1.6 cm}\left((x'_{A_1},x'_{A_2})(x'_{B_1}e^{-\frac{jk2\pi}{2^\lambda}},x'_{B_2} e^{-\frac{jk2\pi}{2^\lambda}})\right)\right\rbrace
\end{align}
} is a singularity removal constraint for the singular fade state $(\gamma, \theta^{\prime})$, where $\theta'-\theta = k\frac{2\pi}{2^\lambda}.$ In other words, the rotation in the $\gamma e^{j\theta}$ plane by an angle $\theta$ can be viewed equivalently as a rotation of the constellations used by B during the MA phases by an angle $\theta^\prime - \theta$. Note that the columns of the Latin Square $L$ which removes the singular fade states $\gamma e^{j\theta}$ are indexed by the symbols $\left(x_{B_1},x_{B_2}\right)$ transmitted by B during the two MA phases. Rotating the signal set used by B during the second MA phase by an angle $\frac{2k\pi}{2^\lambda}$ is equivalent to shifting the columns of the Latin Sub-squares $L_{i,j}$ $k$ times to the left. Similarly, rotating the signal set used by B during the first MA phase by an angle $\frac{2k\pi}{2^{\lambda}}$ is equivalent to cyclic shifting the columns of the square $L_B$, $k$ times to the left. This completes the proof. 

%
%
%
\end{proof}
 \end{lemma}
 
\begin{figure}
\centering

\subfigure[The Latin Square $L$ that removes the singular fade state $(\gamma=1,\theta=0)$]{
\includegraphics[totalheight=2.8in,width=2.8in]{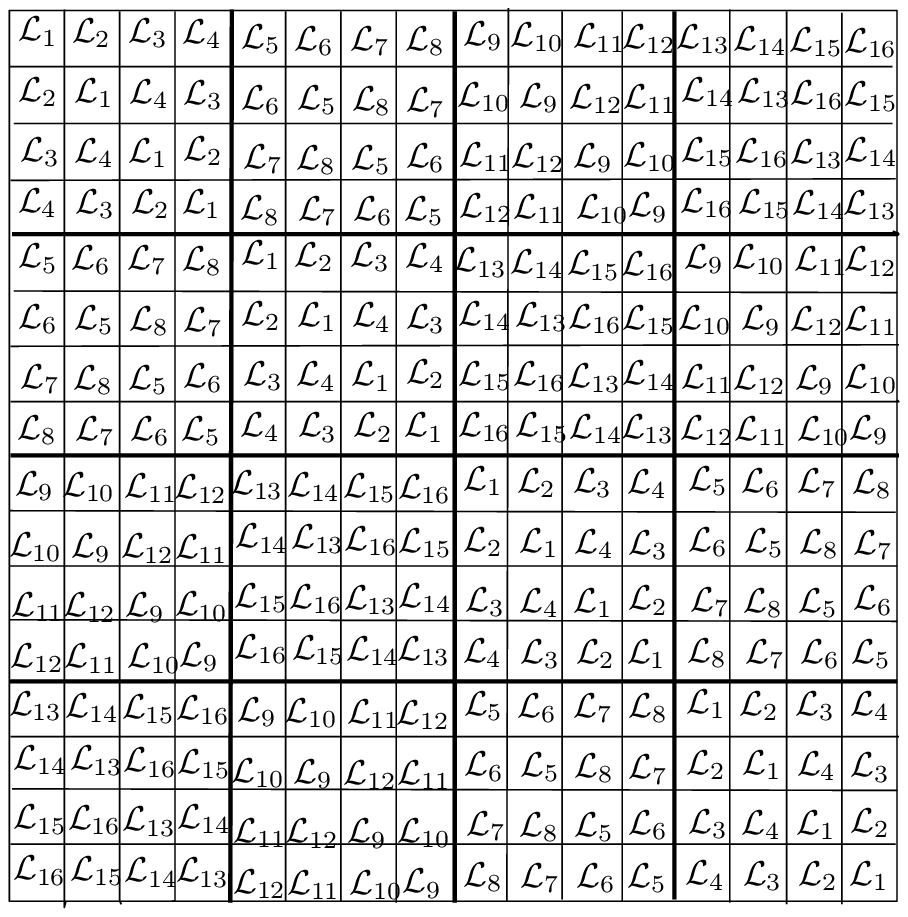}
\label{fig:ex1}	
}


\subfigure[The Latin Square $L'$ obtained form the Latin Square $L$ using the procedure described in Lemma 4]{
\includegraphics[totalheight=2.8in,width=2.8in]{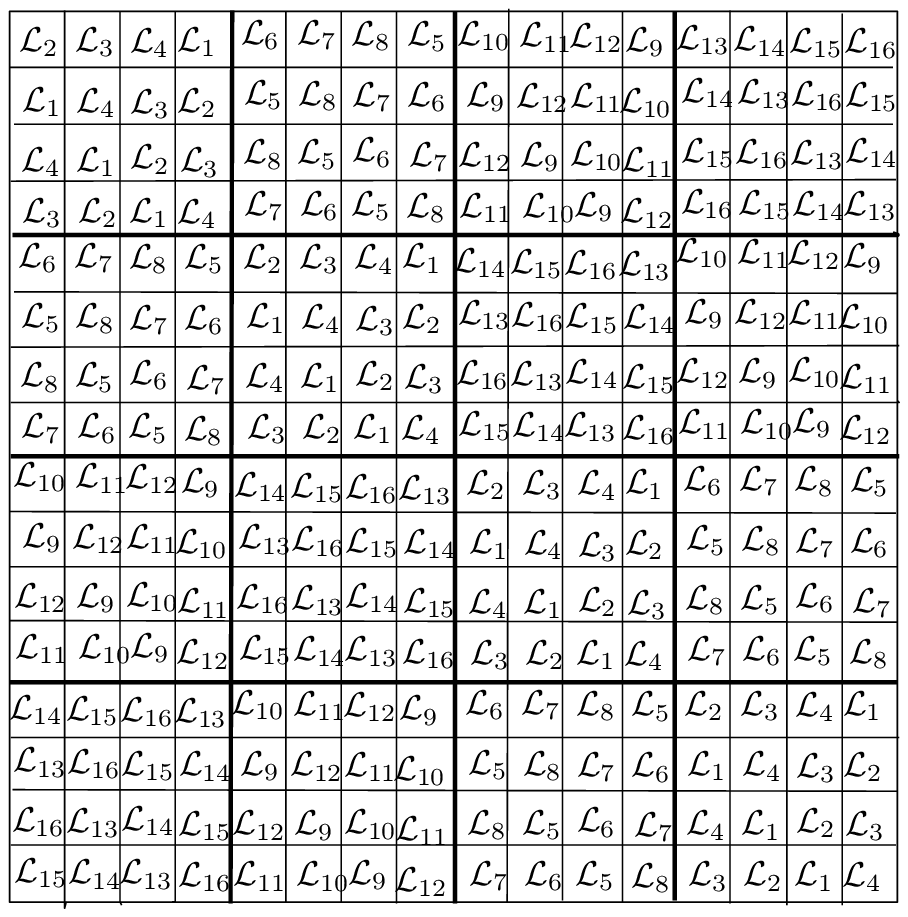}
\label{fig:ex2}	
}
\subfigure[The Latin Square $L''$ that removes the singular fade state $(\gamma=1,\theta=\frac{\pi}{2})$]{
\includegraphics[totalheight=2.8in,width=2.8in]{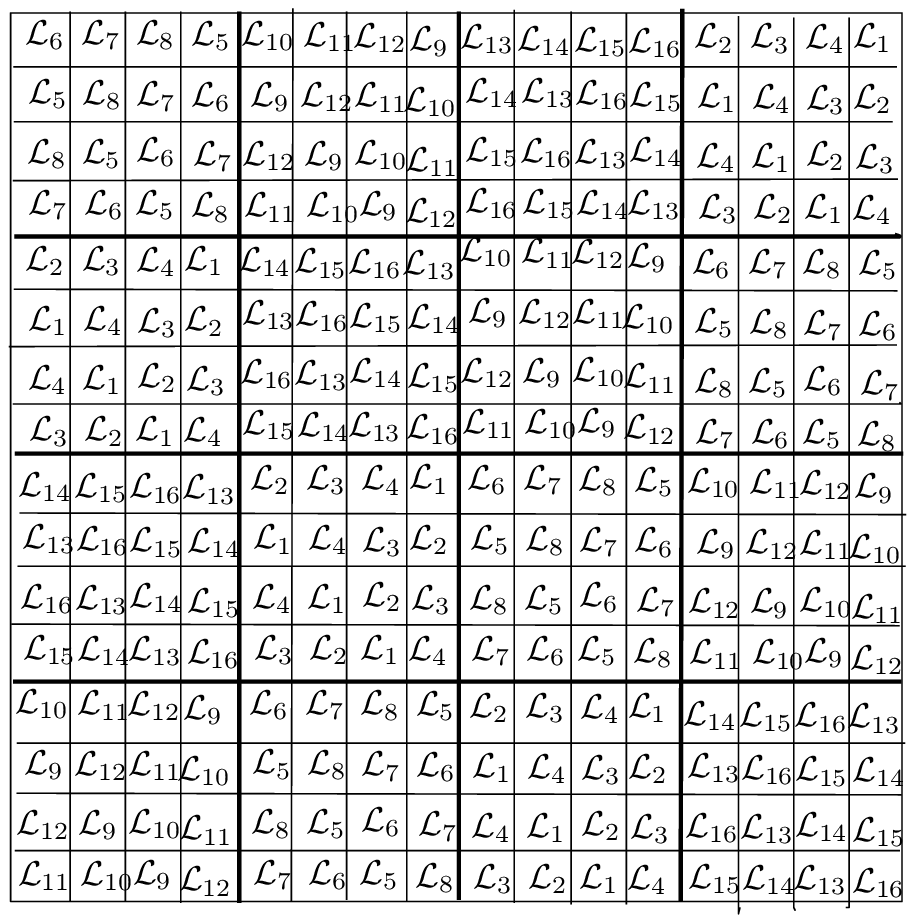}
\label{fig:ex3}	
}
\caption{Construction of the Latin Square $L''$ which removes $(\gamma=1,\theta=\pi/2)$ from the Latin Square $L$ which removes $(\gamma=1,\theta=0)$}
\label{fig:ex}
\end{figure}

 For example consider the Latin Square $L$ in Fig. \ref{fig:ex1} which removes the singular fade state $(\gamma=1,\theta=0).$ The Latin Square $L''$ which removes the singular fade state $(\gamma=1,\theta=\frac{\pi}{2})$ can be obtained from $L$ as follows: The columns of the Latin Sub-squares $L_{i,j}, 0 \leq i,j \leq 3$ are cyclically shifted once to the left, to obtain the Latin Square $L'$ shown in Fig. \ref{fig:ex2}. The columns of the Square $L'_{B}$ associated with the Latin Square $L'$ are cyclically shifted once to the left, to obtain the Square $L''_{B}$. The Latin Square $L''$ associated with the Square $L''_B$ which removes the singular fade state $(\gamma=1,\theta=\frac{\pi}{2})$ is shown in Fig. \ref{fig:ex3}.\\   
 
 For the ACF two-way relaying scenario with 4-PSK signal set used at the end nodes, the twelve singular fade states lie on three circles with radii $1,$  $\frac{1}{\sqrt{2}},$ and $\sqrt{2}.$ The Latin Squares which remove all the twelve singular fade states can be obtained from the three Latin Squares which remove the singular fade states $j,$ $0.5+0.5j$ and $1+j$ given in Fig. 3, Fig. 5 and Fig.7  respectively.

\section{Direct Clustering}

\begin{algorithm}[H]
\label{Alg}
\SetLine
\linesnumbered
\KwIn{The constrained $16 \times 16$ array}
\KwOut{A Latin Square representing the clustering map at the relay}

Start with the constrained $16 \times 16$ array

Initialize all empty cells of the array to 0

\For{$1\leq i\leq 16 $}{

\For{$1\leq j\leq 16 $}{

\If{cell $\left(i,j\right)$ of the array is NULL}{

Initialize c=1

\eIf{$\mathcal{L}_{c}$ does not occur in the $i^{th}$ row or the $j^{th}$ column of the array}{
replace 0 at cell $\left(i,j\right)$ of the array with $\mathcal{L}_{c}$\;
}{
c=c+1\; 
}
}
}
}

\caption{Obtaining the $16 \times 16$ Latin Square from the $16 \times 16 $ array constrained using Singularity Removal Constraints}

\end{algorithm}

Recall that there are three classes of singular fade states depending on the radius of the circle it lies on {\footnotesize\textit{(Case 1:)}} $\gamma e^{j \theta} $ lies on the unit circle, {\footnotesize\textit{(Case 2:)}} $\gamma e^{j \theta} $ lies on the circle of radius $ 1/\sqrt{2}$ and {\footnotesize\textit{(Case 3:)}} $\gamma e^{j \theta} $ lies on the circle of radius $ \sqrt{2}$. The number of clusters in the clustering utilized by relay node R during BC phase obtained using Cartesian Product in the three cases is $16$, $25$ and $25$ respectively. It is observed that, if instead of taking the Cartesian Product of the clusterings given in \cite{NMR}, the Cartesian Product of the \textit{Singularity Removal Constraints} corresponding to each fade state are used to fill a $16 \times 16$ array, and the resulting incomplete array so obtained is completed using Algorithm \ref{Alg}, so as to form a Latin Square of side 16, then the number of clusters of the resulting clustering corresponding the to this Latin Square can be reduced from 25 to a lesser number in both \textit{Case 2} and \textit{Case 3}. We call this the Direct Clustering. We now explain this clustering with the help of examples in the second and the third case only, since for the first case, the minimum number of clusters required, i.e., $16$, can be achieved using Cartesian Product Clustering as shown in Section III.\\\\

\noindent \textbf{\textit{Case 2:} $\gamma e^{j \theta} $ lies on the circle of radius $ 1/\sqrt{2}$.}\\

In this case, there are a total of $80$ singularity removal constraints as shown in the following lemma.\\

\begin{lemma} When $\gamma e^{j \theta} $ lies on the circle of radius $ 1/\sqrt{2}$, there are a total of $80$ singularity removal constraints.
\end{lemma}
\begin{proof} Let the singularity removal constraints for the two-way ACF relaying be the set $\left\{ \mathcal{C}_{1},\mathcal{C}_{2}, ..., \mathcal{C}_{s}\right\}$. Let, for $t=1,2,...,s$,

{\footnotesize
\vspace{-0.4cm}
\begin{align}
\nonumber
&\mathcal{C}_{t}=\left\{((x_{1_{k}},y_{1_{k}}),(x_{2_{k}},y_{2_{k}})) ~ | ~ x_{1_{k}},y_{1_{k}},x_{2_{k}},y_{2_{k}} \in \mathcal{S},~k=1,2,...n_t \right\}.
\end{align}
\vspace{-0.4cm}
}

Then, for $1 \leq k_1, k_2 \leq n_t$,
\begin{align}
\nonumber
&x_{1_{k_{1}}}+ \gamma e^{j \theta} y_{1_{k_{1}}}=x_{1_{k_{2}}}+ \gamma e^{j \theta} y_{1_{k_{2}}} \text{~and}\\
\nonumber
&x_{2_{k_{1}}}+ \gamma e^{j \theta} y_{2_{k_{1}}}=x_{2_{k_{2}}}+ \gamma e^{j \theta} y_{2_{k_{2}}}.
\end{align}

Since in the case of two-way ACF relaying, the user nodes A and B transmit twice to the relay node R, these constraints for the ACF relaying can be obtained by taking the Cartesian Product of all sets of the form 
\begin{align}
\nonumber
&\left\{(x_{A_{l}},x_{B_{l}})~|~ x_{A_{l_{1}}}+\gamma e^{j \theta} x_{B_{l_{1}}}=x_{A_{l_{2}}}+\gamma e^{j \theta} x_{B_{l_{2}}} ~\forall l_1,l_2 \right. \\
\nonumber
& \left. ~~~~~~~~~~~~~~~~~~~~~~~~~~~~\text{~where~} x_{A_{l_{1}}},x_{B_{l_{1}}},x_{A_{l_{2}}},x_{B_{l_{2}}} \in \mathcal{S} \right\}.
\end{align}

These sets can be of two types:
\begin{enumerate}
\item The singularity removal constraints corresponding to the fade state $\gamma e^{j \theta} $ for the two-way 2-stage relaying as given in \cite{NMR}. Let us denote the 2-stage singularity removal constraints as:
\begin{align}
\nonumber
&l_{i}=\left\{(x_{i_{1}},y_{i_{1}}),(x'_{i_{1}},y'_{i_{1}})\right\} \text{~for~} i=1,2,3,4.
\end{align}

\item The sets of the form $\left\{(x_A,x_B),(x_A,x_B)\right\}$ for $(x_A,x_B)\in \mathcal{S}^2$ where $(x_A,x_B)\notin l_{i} ~ \forall ~ i=1,2,3,4$, since $ x_A + \gamma e^{j \theta} x_B= x_A + \gamma e^{j \theta} x_B$. The $(x_A,x_B)$ for which $(x_A,x_B) \in l_{i} ~ \text{for some} ~ i=1,2,3,4,5$ are not considered in this category, as it already occurs in some set of the above category.
\end{enumerate}

The Cartesian Products of these sets amongst themselves must be the singularity removal constraints for the ACF relaying. Now, the constraint sets so obtained are also of two types:
\begin{enumerate}
\item For $i,j\in \lbrace 1,2,3,4 \rbrace$,
{\footnotesize
\begin{align}
\nonumber
l_{i} \times l_{j}= & \left\{((x_{i_{1}},y_{i_{1}}),(x'_{j_{1}},y'_{j_{1}})),((x'_{i_{1}},y'_{i_{1}}),(x_{j_{1}},y_{j_{1}})), \right. \\ 
\nonumber
& \left.((x_{i_{1}},y_{i_{1}}),(x_{j_{1}},y_{j_{1}})),((x'_{i_{1}},y'_{i_{1}}),(x'_{j_{1}},y'_{j_{1}}))\right\},\\
\nonumber
\end{align}
}
These singularity removal constraints account for 16 of the total number of constraints. \\

\item For $i \in \lbrace 1,2,3,4 \rbrace$ and the four $m_k :=((x_{A_{k}},x_{B_{k}}),(x_{A_{k}},x_{B_{k}}))$ for $k \in \lbrace 1,2,...,8 \rbrace$ that satisfy $(x_{A_{k}},x_{B_{k}}) \notin l_{j} ~ \forall j \in \lbrace 1,2,3,4 \rbrace$;

{\footnotesize
\vspace{-0.4cm}
\begin{align}
\nonumber
l_{i} \times m_k= &\left\{((x_{i_{1}},y_{i_{1}}),(x_{A_{k}},x_{B_{k}})),((x'_{i_{1}},y'_{i_{1}}),(x_{A_{k}},x_{B_{k}})) \right\}\\
\nonumber
m_k \times l_{i}= & \left\{((x_{A_{k}},x_{B_{k}}),(x_{i_{1}},y_{i_{1}})),((x_{A_{k}},x_{B_{k}}),(x'_{i_{1}},y'_{i_{1}})) \right\}\\
\nonumber
\end{align}
}
These singularity removal constraints account for remaining 64 constraints. \\
\end{enumerate}

Thus, the set of singularity removal constraints for two-way ACF relaying becomes,

{\footnotesize
\begin{align}
\nonumber
\left\{l_{i} \times l_{j} ~|~ i,j=1,2,3,4\right\} &\cup \left\{l_i \times m_k ~|~ i=1,2,3,4, ~k=1,2,...,8\right\} \\
\nonumber
&\cup \left\{m_k \times l_i ~|~ i=1,2,3,4, ~k=1,2,...,8\right\},
\end{align}
}where the subset $\left\{l_{i} \times l_{j} ~|~ i,j=1,2,3,4\right\}$ contains 16 constraints, and the subsets $\left\{l_i \times m_k ~|~ i=1,2,3,4, ~k=1,2,...,8\right\}$ and $\left\{m_k \times l_i ~|~ i=1,2,3,4, ~k=1,2,...,8\right\}$ contain 32 constraints each, which amount to a total of 80 singularity removal constraints.\\
\end{proof}
\vspace{0.5cm}

The $16\times 16$ Latin Square representing these constraints can be completed using 20 symbols, as we show in the following example.\\

\begin{figure*}
{\scriptsize
\begin{tabular}{||c||l|l|l||}\hline
        & ~~~~Singularity Removal Constraints for $\gamma e^{j \theta}=-0.5+0.5j$ &~~~~~~~~ Latin Square Constraints for $\gamma e^{j \theta}=-0.5+0.5j$ & Cluster \\\hline \hline

(1) &  $\left\{((0,0),(1,3)), ((1,3),(0,0)), ((0,0),(0,0)), ((1,3),(1,3)) \right\} $ &  $\left\{((0,1),(0,3)), ((1,0),(3,0)), ((0,0),(0,0)), ((1,1),(3,3)) \right\} $ & $\mathcal{L}_{1}$  \\\hline
(2) &  $\left\{((1,1),(3,2)), ((3,2),(1,1)),((1,1),(1,1)),((3,2),(3,2))\right\} $   &  $\left\{((1,3),(1,2)), ((3,1),(2,1)),((1,1),(1,1)),((3,3),(2,2))\right\}$ & $\mathcal{L}_{2}$  \\\hline
(3) &  $\left\{((0,1),(2,2)), ((2,2),(0,1)),((0,1),(0,1)),((2,2),(2,2))\right\} $   &  $\left\{((0,2),(1,2)), ((2,0),(2,1)),((0,0),(1,1)),((2,2),(2,2))\right\} $ &  $\mathcal{L}_{3}$ \\\hline
(4) &  $\left\{((2,0),(3,3)), ((3,3),(2,0)),((2,0),(2,0)),((3,3),(3,3))\right\} $   &  $\left\{((2,3),(0,3)), ((3,2),(3,0)),((2,2),(0,0)),((3,3),(3,3))\right\} $ &  $\mathcal{L}_{4}$  \\\hline
(5) &  $\left\{((0,0),(0,1)), ((1,3),(2,2)),((1,3),(0,1)),((0,0),(2,2))\right\} $   &  $\left\{((0,0),(0,1)), ((1,2),(3,2)),((1,0),(3,1)),((0,2),(0,2))\right\} $ &   $\mathcal{L}_{3}$ \\\hline
(6) &  $\left\{((0,0),(1,1)), ((1,3),(3,2)), ((1,3),(1,1)), ((0,0),(3,2)) \right\} $ & $\left\{((0,0),(0,1)), ((1,3),(1,0)), ((1,3),(0,1)), ((0,0),(1,0)) \right\} $  &  $\mathcal{L}_{2}$  \\\hline
(7) &  $\left\{((0,0),(2,0)), ((1,3),(3,3)), ((1,3),(2,0)), ((0,0),(3,3)) \right\} $ & $\left\{((0,2),(0,0)), ((1,3),(3,3)), ((1,2),(3,0)), ((0,3),(0,3)) \right\}$ &  $\mathcal{L}_{5}$ \\\hline
(8) &  $\left\{((0,1),(0,0)), ((2,2),(1,3)), ((2,2),(0,0)), ((0,1),(1,3)) \right\} $  & $\left\{((0,0),(1,0)), ((2,1),(2,3)), ((2,0),(2,0)), ((0,1),(1,3)) \right\}$  &  $\mathcal{L}_{4}$ \\\hline
(9) &  $\left\{((0,1),(1,1)), ((2,2),(3,2)), ((2,2),(3,2)), ((0,1),(1,1)) \right\} $  & $\left\{((0,1),(1,1)), ((2,3),(2,2)), ((2,3),(2,2)), ((0,1),(1,1)) \right\}$  &  $\mathcal{L}_{6}$ \\\hline
(10)&  $\left\{((0,1),(2,0)), ((2,2),(3,3)), ((2,2),(3,3)), ((0,1),(2,0)) \right\} $  & $\left\{((0,2),(1,0)), ((2,3),(2,3)), ((2,3),(2,3)), ((0,2),(1,0)) \right\}$  &  $\mathcal{L}_{1}$ \\\hline
(11) &  $\left\{((1,1),(0,0)), ((3,2),(1,3)), ((3,2),(0,0)), ((1,1),(1,3)) \right\} $ & $\left\{((1,0),(1,0)), ((3,1),(2,3)), ((3,0),(2,0)), ((1,1),(1,3)) \right\}$ & $\mathcal{L}_{5}$  \\\hline
(12) &  $\left\{((1,1),(0,1)), ((3,2),(2,2)), ((3,2),(0,1)), ((1,1),(2,2)) \right\} $  &  $\left\{((1,0),(1,1)), ((3,2),(2,2)), ((3,0),(2,1)), ((1,2),(1,2)) \right\}$ &  $\mathcal{L}_{7}$ \\\hline
(13) &  $\left\{((1,1),(2,0)), ((3,2),(3,3)), ((3,2),(2,0)), ((1,1),(3,3)) \right\} $  &  $\left\{((1,2),(1,0)), ((3,3),(2,3)), ((3,2),(2,0)), ((1,3),(1,3)) \right\}$ & $\mathcal{L}_{6}$  \\\hline
(14) &  $\left\{((2,0),(0,0)), ((3,3),(1,3)), ((3,3),(0,0)), ((2,0),(1,3)) \right\} $  &  $\left\{((2,0),(0,0)), ((3,1),(3,3)), ((3,0),(3,0)), ((2,1),(0,3)) \right\}$ &  $\mathcal{L}_{8}$ \\\hline
(15) &  $\left\{((2,0),(0,1)), ((3,3),(2,2)), ((3,3),(0,1)), ((2,0),(2,2)) \right\} $  &  $\left\{((2,0),(0,1)), ((3,2),(3,2)), ((3,0),(3,1)), ((2,2),(0,2)) \right\}$ & $\mathcal{L}_{9}$  \\\hline
(16) &  $\left\{((2,0),(1,1)), ((3,3),(3,2)), ((3,3),(1,1)), ((2,0),(3,2)) \right\} $  & $\left\{((2,1),(0,1)), ((3,3),(3,2)), ((3,1),(3,1)), ((2,3),(0,2)) \right\}$ & $\mathcal{L}_{7}$ \\\hline
(17) &  $\left\{((0,0),(0,2)), ((1,3),(0,2)) \right\} $  & $\left\{((0,0),(0,2)), ((1,0),(3,2))\right\}$ & $\mathcal{L}_{6}$ \\\hline
(18) &  $\left\{((0,0),(0,3)), ((1,3),(0,3)) \right\} $  & $\left\{((0,0),(0,3)), ((1,0),(3,3)\right\}$ & $\mathcal{L}_{9}$ \\\hline
(19) &  $\left\{((0,0),(1,0)), ((1,3),(1,0)) \right\} $  & $\left\{((0,1),(0,0)), ((1,1),(3,0))\right\}$ & $\mathcal{L}_{7}$ \\\hline
(20) &  $\left\{((0,0),(1,2)), ((1,3),(1,2)) \right\} $  & $\left\{((0,1),(0,2)), ((1,1),(3,2))\right\}$ & $\mathcal{L}_{8}$ \\\hline
(21) &  $\left\{((0,0),(2,1)), ((1,3),(2,1)) \right\} $  & $\left\{((0,2),(0,1)), ((1,2),(3,1))\right\}$ & $\mathcal{L}_{4}$ \\\hline
(22) &  $\left\{((0,0),(2,3)), ((1,3),(2,3)) \right\} $  & $\left\{((0,2),(0,3)), ((1,2),(3,3))\right\}$ & $\mathcal{L}_{10}$  \\\hline
(23) &  $\left\{((0,0),(3,0)), ((1,3),(3,0)) \right\} $  & $\left\{((0,3),(0,0)), ((1,3),(3,0))\right\}$ & $\mathcal{L}_{9}$ \\\hline
(24) &  $\left\{((0,0),(3,1)), ((1,3),(3,1)) \right\} $  & $\left\{((0,3),(0,1)), ((1,3),(3,1))\right\}$ & $\mathcal{L}_{8}$ \\\hline
(25) &  $\left\{((0,2),(0,0)), ((0,2),(1,3)) \right\} $  & $\left\{((0,0),(2,0)), ((0,1),(2,3))\right\}$ & $\mathcal{L}_{10}$ \\\hline
(26) &  $\left\{((0,3),(0,0)), ((0,3),(1,3)) \right\} $  & $\left\{((0,0),(3,0)), ((0,1),(3,3))\right\}$ & $\mathcal{L}_{11}$ \\\hline
(27) &  $\left\{((1,0),(0,0)), ((1,0),(1,3)) \right\} $  & $\left\{((1,0),(0,0)), ((1,1),(0,3))\right\}$ & $\mathcal{L}_{11}$\\\hline
(28) &  $\left\{((1,2),(0,0)), ((1,2),(1,3)) \right\} $  & $\left\{((1,0),(2,0)), ((1,1),(2,3))\right\}$ & $\mathcal{L}_{12}$\\\hline
(29) &  $\left\{((2,1),(0,0)), ((2,1),(1,3)) \right\} $  & $\left\{((2,0),(1,0)), ((2,1),(1,3))\right\}$ & $\mathcal{L}_{3}$ \\\hline
(30) &  $\left\{((2,3),(0,0)), ((2,3),(1,3)) \right\} $  & $\left\{((2,0),(3,0)), ((2,1),(3,3))\right\}$ & $\mathcal{L}_{12}$\\\hline
(31) &  $\left\{((3,0),(0,0)), ((3,0),(1,3)) \right\} $  & $\left\{((3,0),(0,0)), ((3,1),(0,3))\right\}$ & $\mathcal{L}_{2}$ \\\hline
(32) &  $\left\{((3,1),(0,0)), ((3,1),(1,3)) \right\} $  & $\left\{((3,0),(1,0)), ((3,1),(1,3))\right\}$ & $\mathcal{L}_{10}$\\\hline
(33) &  $\left\{((0,1),(0,2)), ((2,2),(0,2)) \right\} $  & $\left\{((0,0),(1,2)), ((2,0),(2,2))\right\}$ & $\mathcal{L}_{5}$\\\hline
(34) &  $\left\{((0,1),(0,3)), ((2,2),(0,3)) \right\} $ & $\left\{((0,0),(1,3)), ((2,0),(2,3)) \right\}$ & $\mathcal{L}_{7}$\\\hline
(35) &  $\left\{((0,1),(1,0)), ((2,2),(1,0)) \right\} $ & $\left\{((0,1),(1,0)), ((2,1),(2,0)) \right\}$ & $\mathcal{L}_{9}$\\\hline
(36) &  $\left\{((0,1),(1,2)), ((2,2),(1,2)) \right\} $ & $\left\{((0,1),(1,2)), ((2,1),(2,2)) \right\}$ & $\mathcal{L}_{13}$\\\hline
(37) &  $\left\{((0,1),(2,1)), ((2,2),(2,1)) \right\} $ & $\left\{((0,2),(1,1)), ((2,2),(2,1)) \right\}$ & $\mathcal{L}_{8}$\\\hline
(38) &  $\left\{((0,1),(2,3)), ((2,2),(2,3)) \right\} $ & $\left\{((0,2),(1,3)), ((2,2),(2,3)) \right\}$ & $\mathcal{L}_{11}$\\\hline
(39) &  $\left\{((0,1),(3,0)), ((2,2),(3,0)) \right\} $ & $\left\{((0,3),(1,0)), ((2,3),(2,0)) \right\}$ & $\mathcal{L}_{11}$\\\hline
(40) &  $\left\{((0,1),(3,1)), ((2,2),(3,1)) \right\} $ & $\left\{((0,3),(1,1)), ((2,3),(2,1)) \right\}$ & $\mathcal{L}_{10}$\\\hline
(41) &  $\left\{((0,2),(0,1)), ((0,2),(2,2)) \right\} $ & $\left\{((0,0),(2,1)), ((0,2),(2,2)) \right\}$ & $\mathcal{L}_{12}$\\\hline
(42) &  $\left\{((0,3),(0,1)), ((0,3),(2,2)) \right\} $ & $\left\{((0,0),(3,1)), ((0,2),(3,2)) \right\}$ & $\mathcal{L}_{13}$\\\hline
(43) &  $\left\{((1,0),(0,1)), ((1,0),(2,2)) \right\} $ & $\left\{((1,0),(0,1)), ((1,2),(0,2)) \right\}$ & $\mathcal{L}_{13}$\\\hline
(44) &  $\left\{((1,2),(0,1)), ((1,2),(2,2)) \right\} $ & $\left\{((1,0),(2,1)), ((1,2),(2,2)) \right\}$ & $\mathcal{L}_{14}$\\\hline
(45) &  $\left\{((2,1),(0,1)), ((2,1),(2,2)) \right\} $ & $\left\{((2,0),(1,1)), ((2,2),(1,2)) \right\}$ & $\mathcal{L}_{14}$\\\hline
(46) &  $\left\{((2,3),(0,1)), ((2,3),(2,2)) \right\} $ & $\left\{((2,0),(3,1)), ((2,2),(3,2)) \right\}$ & $\mathcal{L}_{10}$\\\hline
(47) &  $\left\{((3,0),(0,1)), ((3,0),(2,2)) \right\} $ & $\left\{((3,0),(0,1)), ((3,2),(0,2)) \right\}$ & $\mathcal{L}_{1}$\\\hline
(48) &  $\left\{((3,1),(0,1)), ((3,1),(2,2)) \right\} $ & $\left\{((3,0),(1,1)), ((3,2),(1,2)) \right\}$ & $\mathcal{L}_{11}$\\\hline
(49) &  $\left\{((1,1),(0,2)), ((3,2),(0,2)) \right\} $ & $\left\{((1,0),(1,2)), ((3,0),(2,2)) \right\}$ & $\mathcal{L}_{4}$\\\hline
(50) &  $\left\{((1,1),(0,3)), ((3,2),(0,3)) \right\} $ & $\left\{((1,0),(1,3)), ((3,0),(2,3)) \right\}$ & $\mathcal{L}_{15}$\\\hline
(51) &  $\left\{((1,1),(1,0)), ((3,2),(1,0)) \right\} $ & $\left\{((1,1),(1,0)), ((3,1),(2,0)) \right\}$ & $\mathcal{L}_{13}$\\\hline
(52) &  $\left\{((1,1),(1,2)), ((3,2),(1,2)) \right\} $ & $\left\{((1,1),(1,2)), ((3,1),(2,2)) \right\}$ & $\mathcal{L}_{9}$\\\hline
(53) &  $\left\{((1,1),(2,1)), ((3,2),(2,1)) \right\} $ & $\left\{((1,2),(1,1)), ((3,2),(2,1)) \right\}$ & $\mathcal{L}_{15}$\\\hline
(54) &  $\left\{((1,1),(2,3)), ((3,2),(2,3)) \right\} $ & $\left\{((1,2),(1,3)), ((3,2),(2,3)) \right\}$ & $\mathcal{L}_{2}$\\\hline
(55) &  $\left\{((1,1),(3,0)), ((3,2),(3,0)) \right\} $ & $\left\{((1,3),(1,0)), ((3,3),(2,0)) \right\}$ & $\mathcal{L}_{14}$\\\hline
(56) &  $\left\{((1,1),(3,1)), ((3,2),(3,1)) \right\} $ & $\left\{((1,3),(1,1)), ((3,3),(2,1)) \right\}$ & $\mathcal{L}_{1}$\\\hline
(57) &  $\left\{((0,2),(1,1)), ((0,2),(3,2)) \right\} $ & $\left\{((0,1),(2,1)), ((0,3(,(2,2)) \right\}$ & $\mathcal{L}_{16}$\\\hline
(58) &  $\left\{((0,3),(1,1)), ((0,3),(3,2)) \right\} $ & $\left\{((0,1),(3,1)), ((0,3),(3,2)) \right\}$ & $\mathcal{L}_{12}$\\\hline
(59) &  $\left\{((1,0),(1,1)), ((1,0),(3,2)) \right\} $ & $\left\{((1,1),(0,1)), ((1,3),(0,2)) \right\}$ & $\mathcal{L}_{10}$\\\hline
(60) &  $\left\{((1,2),(1,1)), ((1,2),(3,2)) \right\} $ & $\left\{((1,1),(2,1)), ((1,3),(2,2)) \right\}$ & $\mathcal{L}_{17}$\\\hline
(61) &  $\left\{((2,1),(1,1)), ((2,1),(3,2)) \right\} $ & $\left\{((2,1),(1,1)), ((2,3),(1,2)) \right\}$ & $\mathcal{L}_{16}$\\\hline
(62) &  $\left\{((2,3),(1,1)), ((2,3),(3,2)) \right\} $ & $\left\{((2,1),(3,1)), ((2,3),(3,2)) \right\}$ & $\mathcal{L}_{5}$\\\hline
(63) &  $\left\{((3,0),(1,1)), ((3,0),(3,2)) \right\} $ & $\left\{((3,1),(0,1)), ((3,3),(0,2)) \right\}$ & $\mathcal{L}_{11}$\\\hline
(64) &  $\left\{((3,1),(1,1)), ((3,1),(3,2)) \right\} $ & $\left\{((3,1),(1,1)), ((3,3),(1,2)) \right\}$ & $\mathcal{L}_{12}$\\\hline
(65) &  $\left\{((2,0),(0,2)), ((3,3),(0,2)) \right\} $ & $\left\{((2,0),(0,2)), ((3,0),(3,2)) \right\}$ & $\mathcal{L}_{16}$\\\hline
(66) &  $\left\{((2,0),(0,3)), ((3,3),(0,3)) \right\} $ & $\left\{((2,0),(0,3)), ((3,0),(3,3)) \right\}$ & $\mathcal{L}_{6}$\\\hline
(67) &  $\left\{((2,0),(1,0)), ((3,3),(1,0)) \right\} $ & $\left\{((2,1),(0,0)), ((3,1),(3,0)) \right\}$ & $\mathcal{L}_{14}$\\\hline
(68) &  $\left\{((2,0),(1,2)), ((3,3),(1,2)) \right\} $ & $\left\{((2,1),(0,2)), ((3,1),(3,2)) \right\}$ & $\mathcal{L}_{15}$\\\hline
(69) &  $\left\{((2,0),(2,1)), ((3,3),(2,1)) \right\} $ & $\left\{((2,2),(0,1)), ((3,2),(3,1)) \right\}$ & $\mathcal{L}_{16}$\\\hline
(70) &  $\left\{((2,0),(2,3)), ((3,3),(2,3)) \right\} $ & $\left\{((2,2),(3,3)), ((3,2),(3,3)) \right\}$ & $\mathcal{L}_{3}$\\\hline
(71) &  $\left\{((2,0),(3,0)), ((3,3),(3,0)) \right\} $ & $\left\{((2,3),(0,0)), ((3,3),(3,0)) \right\}$ & $\mathcal{L}_{13}$\\\hline
(72) &  $\left\{((2,0),(3,1)), ((3,3),(3,1)) \right\} $ & $\left\{((2,3),(0,1)), ((3,3),(3,1)) \right\}$ & $\mathcal{L}_{15}$\\\hline
(73) &  $\left\{((0,2),(2,0)), ((0,2),(3,3)) \right\} $ & $\left\{((0,2),(2,0)), ((0,3),(2,3)) \right\}$ & $\mathcal{L}_{17}$\\\hline
(74) &  $\left\{((0,3),(2,0)), ((0,3),(3,3)) \right\} $ & $\left\{((0,2),(3,0)), ((0,3),(3,3)) \right\}$ & $\mathcal{L}_{15}$\\\hline
(75) &  $\left\{((1,0),(2,0)), ((1,0),(3,3)) \right\} $ & $\left\{((1,2),(0,0)), ((1,3),(0,3)) \right\}$ & $\mathcal{L}_{12}$\\\hline
(76) &  $\left\{((1,2),(2,0)), ((1,2),(3,3)) \right\} $ & $\left\{((1,2),(2,0)), ((1,3),(2,3)) \right\}$ & $\mathcal{L}_{16}$\\\hline
(77) &  $\left\{((2,1),(2,0)), ((2,1),(3,3)) \right\} $ & $\left\{((2,2),(1,0)), ((2,3),(1,3)) \right\}$ & $\mathcal{L}_{12}$\\\hline
(78) &  $\left\{((2,3),(2,0)), ((2,3),(3,3)) \right\} $ & $\left\{((2,2),(3,0)), ((2,3),(3,3)) \right\}$ & $\mathcal{L}_{17}$\\\hline
(79) &  $\left\{((3,0),(2,0)), ((3,0),(3,3)) \right\} $ & $\left\{((3,2),(0,0)), ((3,3),(0,3)) \right\}$ & $\mathcal{L}_{17}$\\\hline
(80) &  $\left\{((3,1),(2,0)), ((3,1),(3,3)) \right\} $ & $\left\{((3,2),(1,0)), ((3,3),(1,3)) \right\}$ & $\mathcal{L}_{8}$\\\hline

\end{tabular}
}
\label{case2}
\vspace{-.1 cm}
\caption{Singularity Removal Constraints Constraints for $\gamma e^{j \theta}=-0.5+0.5j$}
\end{figure*}

\begin{figure*}
{\footnotesize
{
\renewcommand{\arraystretch}{1,3}
\begin{tabular}{||c||c|c|c|c|c|c|c|c|c|c|c|c|c|c|c|c||}\hline
    &$(0,0)$&$(0,1)$&$(0,2)$&$(0,3)$&$(1,0)$&$(1,1)$&$(1,2)$&$(1,3)$&$(2,0)$&$(2,1)$&$(2,2)$&$(2,3)$&$(3,0)$&$(3,1)$&$(3,2)$&$(3,3)$ \\\hline \hline
 $(0,0)$ &  $\mathcal{L}_{1}$ & $\mathcal{L}_{3}$ & $\mathcal{L}_{6}$ & $\mathcal{L}_{9}$ & $\mathcal{L}_{4}$ & $\mathcal{L}_{2}$ & $\mathcal{L}_{5}$ & $\mathcal{L}_{7}$ & $\mathcal{L}_{10}$ & $\mathcal{L}_{12}$ & \pmb{$\mathcal{L}_{8}$} & \pmb{$\mathcal{L}_{14}$} & $\mathcal{L}_{11}$ & $\mathcal{L}_{13}$ & \pmb{$\mathcal{L}_{17}$} & \pmb{$\mathcal{L}_{18}$} \\\hline 
 $(0,1)$ & $\mathcal{L}_{7}$ &  $\mathcal{L}_{2}$ & $\mathcal{L}_{8}$ & $\mathcal{L}_{1}$ & $\mathcal{L}_{9}$ & $\mathcal{L}_{6}$ & $\mathcal{L}_{13}$ & $\mathcal{L}_{4}$ & \pmb{$\mathcal{L}_{3}$} & $\mathcal{L}_{16}$ & \pmb{$\mathcal{L}_{15}$} & $\mathcal{L}_{10}$ & \pmb{$\mathcal{L}_{18}$} & $\mathcal{L}_{12}$ & \pmb{$\mathcal{L}_{14}$} & $\mathcal{L}_{11}$ \\\hline
 $(0,2)$ & $\mathcal{L}_{5}$ & $\mathcal{L}_{4}$ &  $\mathcal{L}_{3}$ & $\mathcal{L}_{10}$ & $\mathcal{L}_{1}$ & $\mathcal{L}_{8}$ & $\mathcal{L}_{2}$ & $\mathcal{L}_{11}$ & $\mathcal{L}_{17}$ & \pmb{$\mathcal{L}_{9}$} & $\mathcal{L}_{12}$ & \pmb{$\mathcal{L}_{18}$} & $\mathcal{L}_{15}$ & \pmb{$\mathcal{L}_{6}$} & $\mathcal{L}_{13}$ & \pmb{$\mathcal{L}_{7}$} \\\hline
 $(0,3)$ & $\mathcal{L}_{9}$ & $\mathcal{L}_{8}$ & $\mathcal{L}_{2}$ &  $\mathcal{L}_{5}$ & $\mathcal{L}_{11}$ & $\mathcal{L}_{10}$ & $\mathcal{L}_{6}$ & $\mathcal{L}_{1}$ & \pmb{$\mathcal{L}_{7}$} & \pmb{$\mathcal{L}_{4}$} & $\mathcal{L}_{16}$ & $\mathcal{L}_{17}$ & \pmb{$\mathcal{L}_{3}$} & \pmb{$\mathcal{L}_{14}$} & $\mathcal{L}_{12}$ & $\mathcal{L}_{15}$ \\\hline
 $(1,0)$ & $\mathcal{L}_{11}$ & $\mathcal{L}_{13}$ & \pmb{$\mathcal{L}_{17}$} & \pmb{$\mathcal{L}_{16}$} &  $\mathcal{L}_{5}$ & $\mathcal{L}_{7}$ & $\mathcal{L}_{4}$ & $\mathcal{L}_{15}$ & $\mathcal{L}_{12}$ & $\mathcal{L}_{14}$ & \pmb{$\mathcal{L}_{10}$} & \pmb{$\mathcal{L}_{8}$} & $\mathcal{L}_{1}$ & $\mathcal{L}_{3}$ & $\mathcal{L}_{6}$ & $\mathcal{L}_{9}$ \\\hline
 $(1,1)$ & \pmb{$\mathcal{L}_{6}$} & $\mathcal{L}_{10}$ & \pmb{$\mathcal{L}_{4}$} & $\mathcal{L}_{11}$ & $\mathcal{L}_{13}$ &  $\mathcal{L}_{3}$ & $\mathcal{L}_{9}$ & $\mathcal{L}_{5}$ & \pmb{$\mathcal{L}_{15}$} & $\mathcal{L}_{17}$ & \pmb{$\mathcal{L}_{18}$} & $\mathcal{L}_{12}$ & $\mathcal{L}_{7}$ & $\mathcal{L}_{2}$ & $\mathcal{L}_{8}$ & $\mathcal{L}_{1}$ \\\hline
 $(1,2)$ & $\mathcal{L}_{12}$ & \pmb{$\mathcal{L}_{17}$} & $\mathcal{L}_{13}$ & \pmb{$\mathcal{L}_{18}$} & $\mathcal{L}_{6}$ & $\mathcal{L}_{15}$ &  $\mathcal{L}_{7}$ & $\mathcal{L}_{2}$ & $\mathcal{L}_{16}$ & \pmb{$\mathcal{L}_{11}$} & $\mathcal{L}_{14}$ & \pmb{$\mathcal{L}_{9}$} & $\mathcal{L}_{5}$ & $\mathcal{L}_{4}$ & $\mathcal{L}_{3}$ & $\mathcal{L}_{10}$ \\\hline
 $(1,3)$ & \pmb{$\mathcal{L}_{15}$} & \pmb{$\mathcal{L}_{18}$} & $\mathcal{L}_{10}$ & $\mathcal{L}_{12}$ & $\mathcal{L}_{14}$ & $\mathcal{L}_{1}$ & $\mathcal{L}_{3}$ &  $\mathcal{L}_{6}$ & \pmb{$\mathcal{L}_{19}$} & \pmb{$\mathcal{L}_{13}$} & $\mathcal{L}_{17}$ & $\mathcal{L}_{16}$ & $\mathcal{L}_{9}$ & $\mathcal{L}_{8}$ & $\mathcal{L}_{2}$ & $\mathcal{L}_{5}$ \\\hline
 $(2,0)$ & $\mathcal{L}_{8}$ & $\mathcal{L}_{9}$ & $\mathcal{L}_{16}$ & $\mathcal{L}_{6}$ & $\mathcal{L}_{3}$ & $\mathcal{L}_{14}$ & \pmb{$\mathcal{L}_{1}$} & \pmb{$\mathcal{L}_{13}$} &  $\mathcal{L}_{4}$ & $\mathcal{L}_{2}$ & $\mathcal{L}_{5}$ & $\mathcal{L}_{7}$ & $\mathcal{L}_{12}$ & $\mathcal{L}_{10}$ & \pmb{$\mathcal{L}_{11}$} & \pmb{$\mathcal{L}_{19}$} \\\hline
 $(2,1)$ & $\mathcal{L}_{14}$ & $\mathcal{L}_{7}$ & $\mathcal{L}_{15}$ & $\mathcal{L}_{8}$ & \pmb{$\mathcal{L}_{2}$} & $\mathcal{L}_{16}$ & \pmb{$\mathcal{L}_{10}$} & $\mathcal{L}_{3}$ & $\mathcal{L}_{9}$ &  $\mathcal{L}_{6}$ & $\mathcal{L}_{13}$ & $\mathcal{L}_{4}$ & \pmb{$\mathcal{L}_{19}$} & $\mathcal{L}_{5}$ & \pmb{$\mathcal{L}_{1}$} & $\mathcal{L}_{12}$ \\\hline
 $(2,2)$ & $\mathcal{L}_{4}$ & $\mathcal{L}_{16}$ & $\mathcal{L}_{9}$ & $\mathcal{L}_{3}$ & $\mathcal{L}_{12}$ & \pmb{$\mathcal{L}_{5}$} & $\mathcal{L}_{14}$ & \pmb{$\mathcal{L}_{18}$} & $\mathcal{L}_{1}$ & $\mathcal{L}_{8}$ &  $\mathcal{L}_{2}$ & $\mathcal{L}_{11}$ & $\mathcal{L}_{17}$ & \pmb{$\mathcal{L}_{19}$}& $\mathcal{L}_{10}$ & \pmb{$\mathcal{L}_{13}$} \\\hline
 $(2,3)$ & $\mathcal{L}_{13}$ & $\mathcal{L}_{15}$ & $\mathcal{L}_{7}$ & $\mathcal{L}_{4}$ & \pmb{$\mathcal{L}_{18}$} & \pmb{$\mathcal{L}_{9}$} & $\mathcal{L}_{16}$ & $\mathcal{L}_{12}$ & $\mathcal{L}_{11}$ & $\mathcal{L}_{10}$ & $\mathcal{L}_{6}$ &  $\mathcal{L}_{1}$ & \pmb{$\mathcal{L}_{2}$} & \pmb{$\mathcal{L}_{20}$} & $\mathcal{L}_{5}$ & $\mathcal{L}_{17}$ \\\hline
 $(3,0)$ & $\mathcal{L}_{2}$ & $\mathcal{L}_{1}$ & \pmb{$\mathcal{L}_{12}$} & \pmb{$\mathcal{L}_{13}$} & $\mathcal{L}_{10}$ & $\mathcal{L}_{11}$ & \pmb{$\mathcal{L}_{17}$} & \pmb{$\mathcal{L}_{14}$} & $\mathcal{L}_{5}$ & $\mathcal{L}_{7}$ & $\mathcal{L}_{4}$ & $\mathcal{L}_{15}$ &  $\mathcal{L}_{8}$ & $\mathcal{L}_{9}$ & $\mathcal{L}_{16}$ & $\mathcal{L}_{6}$ \\\hline
 $(3,1)$ & \pmb{$\mathcal{L}_{16}$} & $\mathcal{L}_{11}$ & \pmb{$\mathcal{L}_{18}$} & $\mathcal{L}_{2}$ & \pmb{$\mathcal{L}_{17}$} & $\mathcal{L}_{12}$ & \pmb{$\mathcal{L}_{19}$} & $\mathcal{L}_{10}$ & $\mathcal{L}_{13}$ & $\mathcal{L}_{3}$ & $\mathcal{L}_{9}$ & $\mathcal{L}_{5}$ & $\mathcal{L}_{14}$ &  $\mathcal{L}_{7}$ & $\mathcal{L}_{15}$ & $\mathcal{L}_{8}$ \\\hline
 $(3,2)$ & $\mathcal{L}_{17}$ & \pmb{$\mathcal{L}_{5}$} & $\mathcal{L}_{1}$ & \pmb{$\mathcal{L}_{14}$} & $\mathcal{L}_{8}$ & \pmb{$\mathcal{L}_{13}$} & $\mathcal{L}_{11}$ & \pmb{$\mathcal{L}_{19}$} & $\mathcal{L}_{6}$ & $\mathcal{L}_{5}$ & $\mathcal{L}_{7}$ & $\mathcal{L}_{2}$ & $\mathcal{L}_{4}$ & $\mathcal{L}_{16}$ &  $\mathcal{L}_{9}$ & $\mathcal{L}_{3}$ \\\hline
 $(3,3)$ & \pmb{$\mathcal{L}_{10}$} & \pmb{$\mathcal{L}_{19}$} & $\mathcal{L}_{11}$ & $\mathcal{L}_{17}$ & \pmb{$\mathcal{L}_{16}$} & \pmb{$\mathcal{L}_{18}$} & $\mathcal{L}_{12}$ & $\mathcal{L}_{8}$ & $\mathcal{L}_{14}$ & $\mathcal{L}_{1}$ & $\mathcal{L}_{3}$ & $\mathcal{L}_{6}$ & $\mathcal{L}_{13}$ & $\mathcal{L}_{15}$ & $\mathcal{L}_{7}$ &  $\mathcal{L}_{4}$ \\\hline
 
\end{tabular}
}
}
\label{table2}
\caption{Latin Square representing the clustering at the relay for the case $\gamma e^{j \theta}$ obtained using Direct Clustering, with the 4-PSK symbols that A(B) sent in the first and second channel use along the rows(columns)}
\end{figure*}

\begin{example} Consider the case for which $\gamma e^{j \theta}=-0.5+0.5j$. The singularity removal constraints for the case $\gamma e^{j \theta}=-0.5+0.5j$ in two-way 2-stage relaying as given in \cite{NMR} are:\\
{\small $\left\{(0,0),(1,3)\right\}, \left\{(1,1),(3,2)\right\}, \left\{(0,1),(2,2)\right\} \text{~and~} \left\{(2,0),(3,3)\right\}.$}\\

As a result, the singularity removal constraints for the two-way ACF relaying are as shown in Fig. 12 and Fig. 13. The clusters as shown in the third column of the table, are chosen such that each cluster satisfies the mutually exclusive laws given by (\ref{mel1}) and (\ref{mel2}). \\

The constraints can be represented using 17 symbols. In order to complete the latin square, we use Algorithm 1. A total of 20 symbols suffice in completing the array. Fig. 14 represents the Latin Square representing the clustering at the relay for the case $\gamma e^{j \theta}=-0.5+0.5j$, with the 4-PSK symbols that A(B) sent in the first and second channel use along the rows(columns).\\

\end{example}

\noindent \textbf{\textit{Case 3:} $\gamma e^{j \theta} $ lies on the circle of radius $\sqrt{2}$.}\\

In this case, there are a total of $80$ singularity removal constraints as stated in the following lemma:\\

\begin{lemma} When $\gamma e^{j \theta} $ lies on the circle of radius $ \sqrt{2}$, there are a total of $80$ singularity removal constraints.
\end{lemma}
\textit{We omit the proof of this lemma, as it is the same as that of Lemma 3.}

The resulting constrained $16 \times 16$ array can be uniquely completed with 20 symbols as shown in the following example:\\

\begin{example} Consider the case for which $\gamma e^{j \theta}=-1+j$. The singularity removal constraints for the case $\gamma e^{j \theta}=-1+j$ in two-way 2-stage relaying as given in \cite{NMR} are:\\
{\small $\left\{(0,0),(3,2)\right\}, \left\{(0,1),(3,3)\right\}, \left\{(1,1),(2,0)\right\} \text{~and~} \left\{(1,3),(2,2)\right\}.$}\\

As a result, the constraints representing the singularity removal constraints for the two-way ACF relaying can be represented by the bold letters in the $16 \times 16$ Latin Square. The constraints can be represented using 18 symbols. In order to complete the latin square using Algorithm 1, a total of 20 symbols suffice. Fig. 15 represents the clustering. The constraints representing the singularity removal constraints for this example are given in Appendix C.\\

\begin{figure*}
{\footnotesize
{
\renewcommand{\arraystretch}{1,3}
\begin{tabular}{||c||c|c|c|c|c|c|c|c|c|c|c|c|c|c|c|c||}\hline
    &$(0,0)$&$(0,1)$&$(0,2)$&$(0,3)$&$(1,0)$&$(1,1)$&$(1,2)$&$(1,3)$&$(2,0)$&$(2,1)$&$(2,2)$&$(2,3)$&$(3,0)$&$(3,1)$&$(3,2)$&$(3,3)$ \\\hline \hline
 $(0,0)$ &  $\mathcal{L}_{1}$ & $\mathcal{L}_{2}$ & $\mathcal{L}_{7}$ & $\mathcal{L}_{5}$ & $\mathcal{L}_{4}$ & $\mathcal{L}_{3}$ & $\mathcal{L}_{10}$ & $\mathcal{L}_{9}$ & $\mathcal{L}_{8}$ & $\mathcal{L}_{12}$ &  \pmb{$\mathcal{L}_{15} $} &  \pmb{$\mathcal{L}_{16} $} & $\mathcal{L}_{11}$ & $\mathcal{L}_{14}$ &  \pmb{$\mathcal{L}_{13} $} &  \pmb{$\mathcal{L}_{17} $} \\\hline 
 $(0,1)$ & $\mathcal{L}_{8}$ &  $\mathcal{L}_{3}$ & $\mathcal{L}_{9}$ & $\mathcal{L}_{4}$ & $\mathcal{L}_{10}$ & $\mathcal{L}_{1}$ & $\mathcal{L}_{7}$ & $\mathcal{L}_{2}$ &  \pmb{$\mathcal{L}_{6} $} & $\mathcal{L}_{13}$ &  \pmb{$\mathcal{L}_{11} $} & $\mathcal{L}_{14}$ &  \pmb{$\mathcal{L}_{15} $} & $\mathcal{L}_{12}$ &  \pmb{$\mathcal{L}_{18} $} & $\mathcal{L}_{16}$ \\\hline
 $(0,2)$ & $\mathcal{L}_{3}$ & $\mathcal{L}_{6}$ &  $\mathcal{L}_{4}$ & $\mathcal{L}_{9}$ & $\mathcal{L}_{1}$ & $\mathcal{L}_{10}$ & $\mathcal{L}_{2}$ & $\mathcal{L}_{5}$ & $\mathcal{L}_{13}$ &  \pmb{$\mathcal{L}_{18} $} & $\mathcal{L}_{14}$ &  \pmb{$\mathcal{L}_{7} $} & $\mathcal{L}_{12}$ &  \pmb{$\mathcal{L}_{19} $} & $\mathcal{L}_{16}$ &  \pmb{$\mathcal{L}_{20} $} \\\hline
 $(0,3)$ & $\mathcal{L}_{9}$ & $\mathcal{L}_{10}$ & $\mathcal{L}_{1}$ &  $\mathcal{L}_{2}$ & $\mathcal{L}_{13}$ & $\mathcal{L}_{6}$ & $\mathcal{L}_{4}$ & $\mathcal{L}_{3}$ &  \pmb{$\mathcal{L}_{15} $} &  \pmb{$\mathcal{L}_{5} $} & $\mathcal{L}_{8}$ & $\mathcal{L}_{12}$ &  \pmb{$\mathcal{L}_{17} $} &  \pmb{$\mathcal{L}_{18} $} & $\mathcal{L}_{11}$ & $\mathcal{L}_{14}$ \\\hline
 $(1,0)$ & $\mathcal{L}_{2}$ & $\mathcal{L}_{1}$ &  \pmb{$\mathcal{L}_{3} $} &  \pmb{$\mathcal{L}_{18} $} &  $\mathcal{L}_{6}$ & $\mathcal{L}_{7}$ & $\mathcal{L}_{13}$ & $\mathcal{L}_{12}$ & $\mathcal{L}_{10}$ & $\mathcal{L}_{11}$ &  \pmb{$\mathcal{L}_{16} $} &  \pmb{$\mathcal{L}_{19} $} & $\mathcal{L}_{5}$ & $\mathcal{L}_{8}$ & $\mathcal{L}_{17}$ & $\mathcal{L}_{15}$ \\\hline
 $(1,1)$ &  \pmb{$\mathcal{L}_{10} $} & $\mathcal{L}_{4}$ &  \pmb{$\mathcal{L}_{12} $} & $\mathcal{L}_{17}$ & $\mathcal{L}_{14}$ &  $\mathcal{L}_{5}$ & $\mathcal{L}_{15}$ & $\mathcal{L}_{8}$ & \pmb{ $\mathcal{L}_{2} $} & $\mathcal{L}_{16}$ &  \pmb{$\mathcal{L}_{18} $} & $\mathcal{L}_{13}$ & $\mathcal{L}_{11}$ & $\mathcal{L}_{7}$ & $\mathcal{L}_{3}$ & $\mathcal{L}_{6}$ \\\hline
 $(1,2)$ & $\mathcal{L}_{4}$ &  \pmb{$\mathcal{L}_{11} $} & $\mathcal{L}_{17}$ & \pmb{ $\mathcal{L}_{3} $} & $\mathcal{L}_{5}$ & $\mathcal{L}_{9}$ &  $\mathcal{L}_{8}$ & $\mathcal{L}_{10}$ & $\mathcal{L}_{16}$ &  \pmb{$\mathcal{L}_{14} $} & $\mathcal{L}_{13}$ &  \pmb{$\mathcal{L}_{17} $} & $\mathcal{L}_{7}$ & $\mathcal{L}_{15}$ & $\mathcal{L}_{6}$ & $\mathcal{L}_{1}$ \\\hline
 $(1,3)$ &  \pmb{$\mathcal{L}_{16} $} &  \pmb{$\mathcal{L}_{18} $} & $\mathcal{L}_{2}$ & $\mathcal{L}_{1}$ & $\mathcal{L}_{15}$ & $\mathcal{L}_{12}$ & $\mathcal{L}_{6}$ &  $\mathcal{L}_{7}$ &  \pmb{$\mathcal{L}_{4} $} &  \pmb{$\mathcal{L}_{9} $} & $\mathcal{L}_{10}$ & $\mathcal{L}_{11}$ & $\mathcal{L}_{14}$ & $\mathcal{L}_{17}$ & $\mathcal{L}_{5}$ & $\mathcal{L}_{8}$ \\\hline
 $(2,0)$ & $\mathcal{L}_{6}$ & $\mathcal{L}_{7}$ & $\mathcal{L}_{13}$ & $\mathcal{L}_{12}$ & $\mathcal{L}_{3}$ & $\mathcal{L}_{4}$ &  \pmb{$\mathcal{L}_{1} $} &  \pmb{$\mathcal{L}_{16} $} &  $\mathcal{L}_{5}$ & $\mathcal{L}_{8}$ & $\mathcal{L}_{17}$ & $\mathcal{L}_{15}$ & $\mathcal{L}_{9}$ & $\mathcal{L}_{11}$ &  \pmb{$\mathcal{L}_{14} $} &  \pmb{$\mathcal{L}_{10} $} \\\hline
 $(2,1)$ & $\mathcal{L}_{14}$ & $\mathcal{L}_{5}$ & $\mathcal{L}_{15}$ & $\mathcal{L}_{8}$ &  \pmb{$\mathcal{L}_{9} $} & $\mathcal{L}_{2}$ &  \pmb{$\mathcal{L}_{10} $} & $\mathcal{L}_{18}$ & $\mathcal{L}_{11}$ &  $\mathcal{L}_{7}$ & $\mathcal{L}_{3}$ & $\mathcal{L}_{6}$ &  \pmb{$\mathcal{L}_{19} $} & $\mathcal{L}_{16}$ &  \pmb{$\mathcal{L}_{1} $} & $\mathcal{L}_{12}$ \\\hline
 $(2,2)$ & $\mathcal{L}_{5}$ & $\mathcal{L}_{9}$ & $\mathcal{L}_{8}$ & $\mathcal{L}_{10}$ & $\mathcal{L}_{2}$ &  \pmb{$\mathcal{L}_{13} $} & $\mathcal{L}_{18}$ &  \pmb{$\mathcal{L}_{11} $} & $\mathcal{L}_{7}$ & $\mathcal{L}_{15}$ &  $\mathcal{L}_{6}$ & $\mathcal{L}_{1}$ & $\mathcal{L}_{16}$ &  \pmb{$\mathcal{L}_{4} $} & $\mathcal{L}_{12}$ &  \pmb{$\mathcal{L}_{19} $} \\\hline
 $(2,3)$ & $\mathcal{L}_{15}$ & $\mathcal{L}_{12}$ & $\mathcal{L}_{6}$ & $\mathcal{L}_{7}$ &  \pmb{$\mathcal{L}_{16} $} &  \pmb{$\mathcal{L}_{18} $} & $\mathcal{L}_{3}$ & $\mathcal{L}_{4}$ & $\mathcal{L}_{14}$ & $\mathcal{L}_{17}$ & $\mathcal{L}_{5}$ &  $\mathcal{L}_{8}$ &  \pmb{$\mathcal{L}_{2} $} &  \pmb{$\mathcal{L}_{13} $} & $\mathcal{L}_{9}$ & $\mathcal{L}_{11}$ \\\hline
 $(3,0)$ & $\mathcal{L}_{11}$ & $\mathcal{L}_{14}$ &  \pmb{$\mathcal{L}_{18} $} &  \pmb{$\mathcal{L}_{6} $} & $\mathcal{L}_{12}$ & $\mathcal{L}_{15}$ &  \pmb{$\mathcal{L}_{16} $} &  \pmb{$\mathcal{L}_{13} $} & $\mathcal{L}_{1}$ & $\mathcal{L}_{2}$ & $\mathcal{L}_{7}$ & $\mathcal{L}_{5}$ &  $\mathcal{L}_{4}$ & $\mathcal{L}_{3}$ & $\mathcal{L}_{10}$ & $\mathcal{L}_{9}$ \\\hline
 $(3,1)$ &  \pmb{$\mathcal{L}_{12} $} & $\mathcal{L}_{13}$ &  \pmb{$\mathcal{L}_{5} $} & $\mathcal{L}_{16}$ &  \pmb{$\mathcal{L}_{17} $} & $\mathcal{L}_{11}$ &  \pmb{$\mathcal{L}_{19} $} & $\mathcal{L}_{14}$ & $\mathcal{L}_{8}$ & $\mathcal{L}_{3}$ & $\mathcal{L}_{9}$ & $\mathcal{L}_{4}$ & $\mathcal{L}_{10}$ &  $\mathcal{L}_{1}$ & $\mathcal{L}_{7}$ & $\mathcal{L}_{2}$ \\\hline
 $(3,2)$ & $\mathcal{L}_{13}$ &  \pmb{$\mathcal{L}_{8} $} & $\mathcal{L}_{16}$ &  \pmb{$\mathcal{L}_{15} $} & $\mathcal{L}_{11}$ &  \pmb{$\mathcal{L}_{17} $} & $\mathcal{L}_{14}$ &  \pmb{$\mathcal{L}_{19} $} & $\mathcal{L}_{3}$ & $\mathcal{L}_{6}$ & $\mathcal{L}_{4}$ & $\mathcal{L}_{9}$ & $\mathcal{L}_{1}$ & $\mathcal{L}_{10}$ &  $\mathcal{L}_{2}$ & $\mathcal{L}_{4}$ \\\hline
 $(3,3)$ &  \pmb{$\mathcal{L}_{7} $} &  \pmb{$\mathcal{L}_{16} $} & $\mathcal{L}_{11}$ & $\mathcal{L}_{14}$ &  \pmb{$\mathcal{L}_{8} $} &  \pmb{$\mathcal{L}_{19} $} & $\mathcal{L}_{12}$ & $\mathcal{L}_{15}$ & $\mathcal{L}_{9}$ & $\mathcal{L}_{10}$ & $\mathcal{L}_{1}$ & $\mathcal{L}_{2}$ & $\mathcal{L}_{13}$ & $\mathcal{L}_{6}$ & $\mathcal{L}_{4}$ &  $\mathcal{L}_{3}$ \\\hline
 
\end{tabular}
}
}
\label{table3}
\caption{Latin Square representing the clustering at the relay for the case $\gamma e^{j \theta}=-1+j$, with the 4-PSK symbols that A(B) sent in the first and second channel use along the rows(columns)}
\end{figure*}

\end{example}

\section{QUANTIZATION OF THE COMPLEX FADE STATE PLANE} 
In practice, $\gamma e^{j\theta}$ can take any value in the complex plane (it takes a value equal to one of the singular fade states with zero probability). As explained in Section II, one of the Latin Squares obtained, which remove the singular fade states needs to be chosen, depending on the value of $\gamma e^{j\theta}$. For a $\gamma e^{j\theta}$ which is not a singular fade state, among all the Latin Squares which remove the singular fade states, the Latin Square $\mathcal{C}^{h}, h \in \mathcal{H}$ which has the maximum value of the minimum cluster distance at $\gamma e^{j\theta}$ is chosen. In other words, for a given $\gamma e^{j\theta} \notin \mathcal{H},$ the clustering is chosen to be the one which removes the singular fade state $h \in \mathcal{H}$ which maximizes the metric $d^{2}_{min}\left(\mathcal{C}^{h}, \gamma e^{j \theta}\right)$ given in \eqref{cl3}. In this way, the $\gamma e^{j\theta}-${plane} is quantized into $\vert \mathcal{H} \vert$ point set, depending on which one of the obtained Latin Squares is chosen.\\

 For $(x_A ,x_B) \neq (x'_A,x'_B) \in \mathcal{S}^2,$ let $\mathcal{D}(\gamma,\theta,x_{A},x_{B},x'_A,x'_B)$  be defined as,

{\vspace{-.3 cm}
\begin{align}
\label{decision_metric}
\mathcal{D}(\gamma,\theta,x_A,x_B,x'_A,x'_B)&=\vert (x_A-x'_A)+\gamma e^{j\theta} (x_B-x'_B)\vert.
\end{align}
}

In the following lemma, it is shown that for a given $\gamma e^{j \theta},$ choosing the clustering $\mathcal{C}^h  \in \mathcal{C}^\mathcal{H},$ where $h \in \mathcal{H},$ that maximizes $d^{2}_{min}\left(\mathcal{C}^{h}, \gamma e^{j \theta}\right)$ given in \eqref{cl3}, is the same as choosing the clustering $\mathcal{C}^{\left[-\frac{\mathsf{x_A}-\mathsf{x'_A}}{\mathsf{x_B}-\mathsf{x'_B}}\right]},$ where $(\mathsf{x_A} ,\mathsf{x_B}) \neq (\mathsf{x'_A},\mathsf{x'_B}) \in \mathcal{S}^2,$ that minimizes the simpler metric given in \eqref{decision_metric}.\\
\begin{lemma}
\label{lemma_criterion}
If the complex fade state $\gamma e^{j \theta}$ and  the clustering $\mathcal{C}^{\left[-\frac{\mathsf{x_A}-\mathsf{x'_A}}{\mathsf{x_B}-\mathsf{x'_B}}\right]} \in \mathcal{C}^{\mathcal{H}}$ are such that, {\footnotesize $$\arg\min_{(x_A,x_B) \neq (x'_A,x'_B) \in \mathcal{S}^2} \mathcal{D}(\gamma,\theta,x_A,x_B,x'_A,x'_B)=(\mathsf{x_A},\mathsf{x_B},\mathsf{x'_A},\mathsf{x'_B}),$$}then $\left[-\frac{\mathsf{x_A}-\mathsf{x'_A}}{\mathsf{x_B}-\mathsf{x'_B}}\right] \in \mathcal{H},$ maximizes the metric $d^{2}_{min}\left(\mathcal{C}^{h}, \gamma e^{j \theta}\right)$ given in \eqref{cl3}, among all $h \in \mathcal{H}.$
\begin{proof}
The squared minimum distance of the effective constellation at the relay $d_{min}(\gamma e^{j\theta})$ is given by \eqref{dist}.\\

\begin{figure*}
\scriptsize
\begin{align}
\hline
\label{eqn_f1}
f_1(\gamma e^{j\theta})&=\min_{(x_{A_1},x_{B_1}) \neq (x'_{A_1},x'_{B_1}),(x_{A_2},x_{B_2}) = (x'_{A_2},x'_{B_2})} \left \lbrace \vert (x_{A_1}-x'_{A_1})+\gamma e^{j\theta} (x_{B_1}-x'_{B_1})\vert^2+\vert (x_{A_2}-x'_{A_2})+\gamma e^{j\theta} (x_{B_2}-x'_{B_2})\vert^2\right\rbrace\\
\hline
\label{eqn_f2}
f_2(\gamma e^{j\theta})&=\min_{(x_{A_1},x_{B_1}) = (x'_{A_1},x'_{B_1}),(x_{A_2},x_{B_2}) \neq (x'_{A_2},x'_{B_2})} \left \lbrace \vert (x_{A_1}-x'_{A_1})+\gamma e^{j\theta} (x_{B_1}-x'_{B_1})\vert^2+\vert (x_{A_2}-x'_{A_2})+\gamma e^{j\theta} (x_{B_2}-x'_{B_2})\vert^2\right\rbrace\\
\hline
\label{eqn_f3}
f_3(\gamma e^{j\theta})&=\min_{(x_{A_1},x_{B_1}) \neq (x'_{A_1},x'_{B_1}),(x_{A_2},x_{B_2}) \neq (x'_{A_2},x'_{B_2})} \left \lbrace \vert (x_{A_1}-x'_{A_1})+\gamma e^{j\theta} (x_{B_1}-x'_{B_1})\vert^2+\vert (x_{A_2}-x'_{A_2})+\gamma e^{j\theta} (x_{B_2}-x'_{B_2})\vert^2\right\rbrace
\end{align} 
\hrule
\end{figure*}
%
%

Let $f_1$, $f_2$ and $f_3$ be functions of $\gamma e^{j\theta}$ defined as in \eqref{eqn_f1}-\eqref{eqn_f3} given in the next page. We have, $$d_{min}^2(\gamma e^{j\theta})=\min\lbrace f_1(\gamma e^{j\theta}),f_2(\gamma e^{j\theta}),f_3(\gamma e^{j\theta})\rbrace.$$ From \eqref{eqn_f1} and \eqref{eqn_f2}, it follows that 

{\footnotesize
\begin{equation}
\nonumber
f_1(\gamma e^{j\theta})=f_2(\gamma e^{j \theta})=\hspace{-0.5 cm}\min_{\: (x_A,x_B) \neq (x'_A,x'_B) \in  \mathcal{S}^2} \hspace{-0.5 cm}\vert (x_A-x'_B)+\gamma e^{j\theta} (x_B-x'_B)\vert ^2.
\end{equation}
}
From \eqref{eqn_f3}, it can be seen that,

\begin{equation}
\nonumber
f_3(\gamma e^{j\theta})\geq\min_{(x_A,x_B) \neq (x'_A,x'_B) \in  \mathcal{S}^2}\vert (x_A-x'_B)+\gamma e^{j\theta} (x_B-x'_B)\vert ^2.
\end{equation}
Hence, we have,
{\footnotesize
\begin{align}
\nonumber
d_{min}^2(\gamma e^{j\theta}) = \min_{(x_A,x_B) \neq (x'_A,x'_B) \in  \mathcal{S}^2}\vert (x_A-x'_B)+\gamma e^{j\theta} (x_B-x'_B)\vert ^2.
\end{align}
}
Since, { $$\arg\min_{x_A,x_B,x'_A,x'_B} \mathcal{D}(\gamma,\theta,x_A,x_B,x'_A,x'_B)=(\mathsf{x_A},\mathsf{x_B},\mathsf{x'_A},\mathsf{x'_B}),$$} we have,
$d_{min}^2(\gamma e^{j\theta}) = \vert (\mathsf{x}_A-\mathsf{x}'_B)+\gamma e^{j\theta} (\mathsf{x}_B-\mathsf{x}'_B)\vert ^2.$

For the clustering $\mathcal{C}^{\left[-\frac{\mathsf{x_A}-\mathsf{x'_A}}{\mathsf{x_B}-\mathsf{x'_B}}\right]}$ which removes the singular fade state $-\frac{\mathsf{x_A}-\mathsf{x'_A}}{\mathsf{x_B}-\mathsf{x'_B}}$, the minimum cluster distance is greater than $d_{min}(\gamma e^{j \theta})$, while for all other clusterings in the set $\mathcal{C}^\mathcal{H}$, it is equal to $d_{min}(\gamma e^{j\theta}).$ This completes the proof. 

\end{proof}
\end{lemma}

The decision criterion in Lemma \ref{lemma_criterion} based on which R chooses one of the Latin Squares obtained, is the same as the decision criterion for the two-way 2-stage relaying in \cite{MNR}. Hence, the quantization of the complex fade state plane for the ACF relaying is same as that of the two-way 2-stage relaying obtained in \cite{MNR}.\\

\section{SIMULATION RESULTS}

\begin{figure*}[htbp]
\centering
\includegraphics[totalheight=5.5in,width=7in]{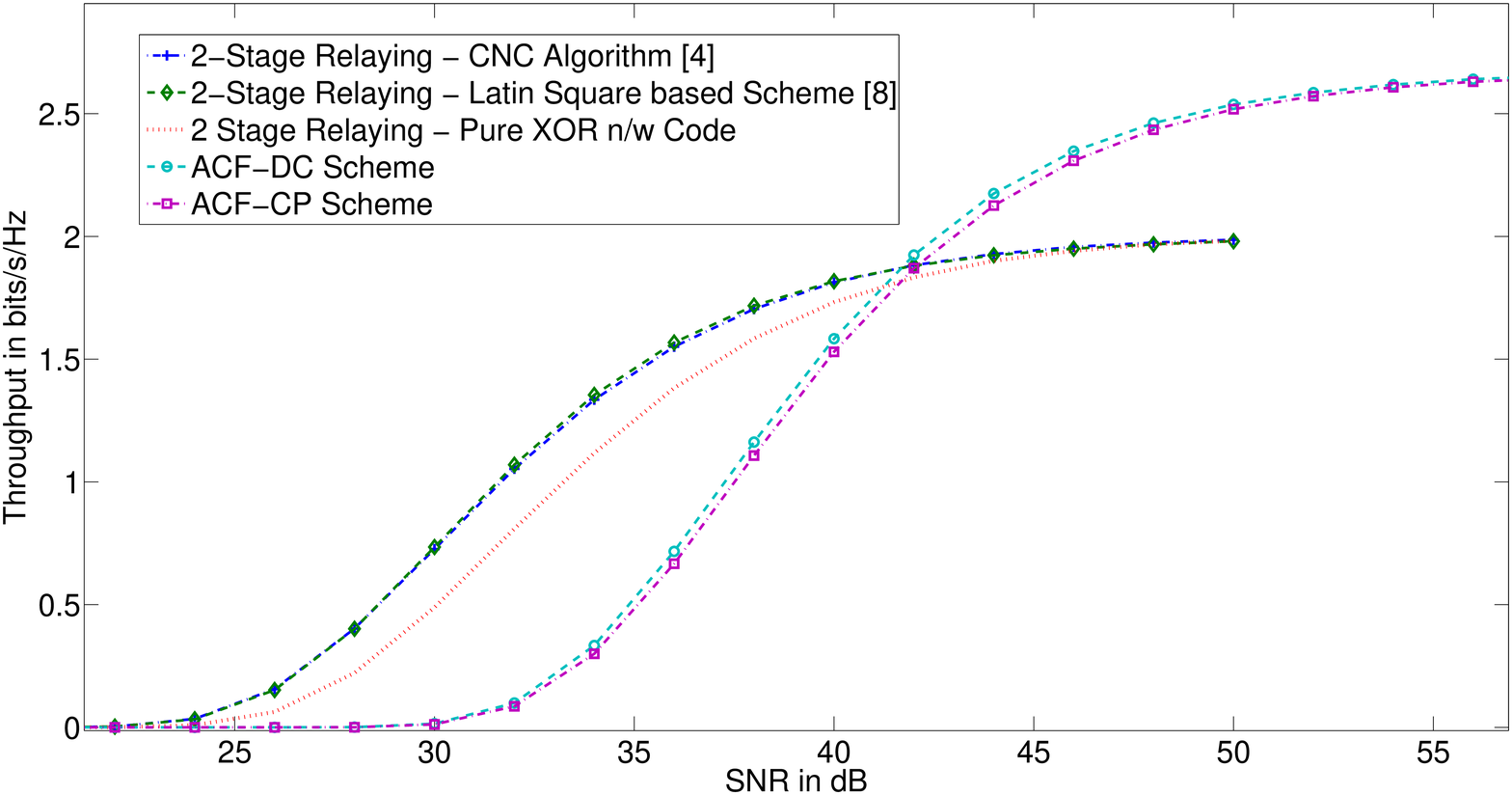}
\caption{SNR vs throughput curves for different schemes for 4-PSK signal set}	
\label{fig:tput_curves_4psk}	
\end{figure*}

\begin{figure*}[htbp]
\centering
\includegraphics[totalheight=5.5in,width=7in]{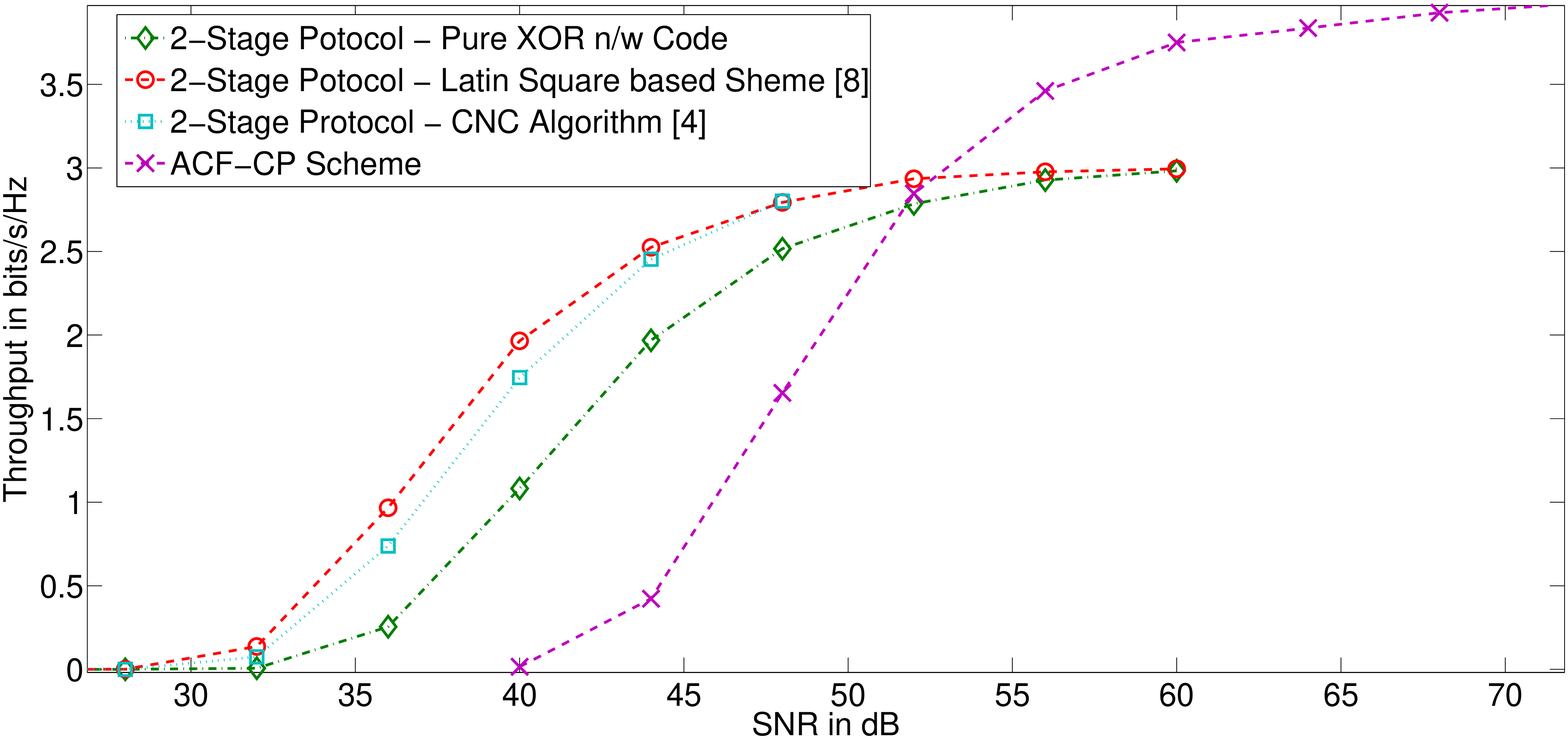}
\caption{SNR vs throughput curves for different schemes for 8-PSK signal set}	
\label{fig:tput_curves_8psk}	
\end{figure*}

The simulation results presented are for the case when $H_A$, $H_B$, $H'_A$ and $H'_B$ are distributed according to Rayleigh distribution, with the variances of all the fading links equal to 0 dB. It is assumed that the AWGN noises at the three nodes are of variance 0 dB. By SNR, we mean the average energies of the signal set used at the three nodes A, B and R, which are assumed to be equal. The frame length of a transmission is taken to be 256 bits.\\
 
 Consider the case when 4-PSK signal set is used at A and B. Fig. \ref{fig:tput_curves_4psk} shows the SNR vs end-to-end sum throughput curves for the following schemes: Closest-Neighbour Clustering (CNC) Algorithm based scheme for the two-way 2-stage relaying proposed in \cite{KoPoTa}, the Scheme based on Latin Squares for two-way 2-stage relaying proposed in \cite{NMR}, the scheme in which XOR network code is used irrespective of the channel condition, the Cartesian Product based scheme for ACF (ACF-CP) relaying and the Direct Clustering based scheme for ACF (ACF-DC) relaying. It can be seen from Fig. \ref{fig:tput_curves_4psk} that the schemes based on the ACF relaying perform better than the schemes based on 2-stage relaying at high SNR. From Fig. \ref{fig:tput_curves_4psk}, it follows that when the SNR is greater than 42 dB, the ACF-DC scheme outperforms all other schemes. The maximum throughput achieved by the ACF relaying schemes is 8/3 bits/s/Hz, whereas it is 2 bits/s/Hz for the 2-stage two-way relaying schemes. Also, as seen from Fig. \ref{fig:tput_curves_4psk}, the ACF-DC scheme performs better than the ACF-CP scheme. The reason for this is that the maximum cardinality of the signal set used during the BC phase is 25 for the ACF-CP scheme whereas it is 20 for the ACF-DC scheme.\\

Consider the case when 8-PSK signal set is used at A and B. It was shown in \cite{NMR} that for the two-way 2-stage relaying scheme with 8-PSK signal set, all the clusterings which the remove the singular fade states have exactly 8 clusters. Hence, the ACF-CP scheme, in which the clusterings are obtained by taking the Cartesian Product of the clusterings corresponding to the two-way 2-stage relaying scheme, have exactly 64 clusters (note that 64 is the minimum number of clusters required for conveying 6 information bits). Since the Cartesian Product itself results in the minimum number of clusters, the ACF-DC scheme is not considered for this case. Fig. \ref{fig:tput_curves_8psk} shows the SNR vs end-to-end sum throughput curves for the different scheme. Similar to the 4-PSK case, at high SNR, the ACF-CP scheme provides a larger throughput than the 2-stage relaying schemes. The maximum throughput achieved by the ACF-CP scheme is 4 bits/s/Hz, whereas it is 3 bits/s/Hz for the 2-stage relaying schemes.\\


\section{Conclusion}
We proposed a scheme based on the ACF protocol for two-way relaying that utilizes totally three channel uses of the wireless two-way relaying channel unlike the 2-stage protocol that uses four channel uses, assuming that the users A and B transmit points from the same 4-PSK constellation. The network codes used at the relay during the Broadcast Phase were obtained using two methods: by taking the Cartesian Product of the clusterings proposed in \cite{NMR} for two user 2-stage case and by completing the Latin Square filled partially with the singularity removal constraints for a given fade state. Using the second method called Direct Clustering, the maximum size of the resulting constellation used by the relay node R in the BC phase was reduced to 20, as compared to the Cartesian Product based approach which results in the constellation size being 25 for these cases. Having obtained all the Latin Squares, the complex plane was quantized depending on which one of the obtained Latin Squares maximizes the minimum cluster distance. This quantization was shown to be the same as that achieved in \cite{MNR} for the two-way 2-stage relaying scenario. Simulation results showed that the ACF protocol based schemes outperform the schemes proposed in \cite{KoPoTa} and \cite{NMR}, at high SNR.\\\\

\begin{center}
\textsc{Acknowledgments}
\end{center}
This work was supported partly by the DRDO-IISc program on Advanced Research in Mathematical Engineering through a research grant as well as the INAE Chair Professorship grant to B. S. Rajan.

\appendices

\section{}
\noindent \textit{Clustering that removes the singular fade state $0.5+0.5j$:}\\
The Cartesian Product of the clustering $\mathcal{C}^{\left[0.5+0.5j\right]}$ with itself denoted by $\mathcal{D}^{\left[0.5+0.5j\right]}=\left\{ \mathcal{C}^{\left\{l_{i},l_{j}\right\}}~ |~ i,j=1,2,3,4,5\right\}$ contains exactly $25$ clusters. The clusters in $\mathcal{D}^{\left[0.5+0.5j\right]}$ have been listed below. 

{\footnotesize
\begin{align}
\nonumber
\vspace{-0.5cm}
\mathcal{C}^{\left\{l_{1},l_{1}\right\}}=&\left\{((0,1),(0,1)), ((0,1),(1,2)), ((0,1),(2,3)), \right.\\
\nonumber 
& \left.  ((1,2),(0,1)), ((1,2),(1,2)), ((1,2),(2,3)), \right.\\
\nonumber 
& \left.  ((2,3),(0,1)), ((2,3),(1,2)), ((2,3),(2,3))\right\}\\
\nonumber
\mathcal{C}^{\left\{l_{1},l_{2}\right\}}=&\left\{((0,1),(0,2)), ((0,1),(1,3)), ((0,1),(3,0)), \right.\\
\nonumber 
& \left.  ((1,2),(0,2)), ((1,2),(1,3)), ((1,2),(3,0)), \right.\\
\nonumber 
& \left.  ((2,3),(0,2)), ((2,3),(1,3)), ((2,3),(3,0))\right\}\\
\nonumber
\mathcal{C}^{\left\{l_{1},l_{3}\right\}}=&\left\{((0,1),(0,3)), ((0,1),(2,0)), ((0,1),(3,1)), \right.\\
\nonumber 
& \left.  ((1,2),(0,3)), ((1,2),(2,0)), ((1,2),(3,1)), \right.\\
\nonumber 
& \left.  ((2,3),(0,3)), ((2,3),(2,0)), ((2,3),(3,1))\right\}\\
\nonumber
\mathcal{C}^{\left\{l_{1},l_{4}\right\}}=&\left\{((0,1),(1,0)), ((0,1),(2,1)), ((0,1),(3,2)), \right.\\
\nonumber 
& \left.  ((1,2),(1,0)), ((1,2),(2,1)), ((1,2),(3,2)), \right.\\
\nonumber 
& \left.  ((2,3),(1,0)), ((2,3),(2,1)), ((2,3),(3,2))\right\}\\
\nonumber
\mathcal{C}^{\left\{l_{1},l_{5}\right\}}=&\left\{((0,1),(0,0)), ((0,1),(1,1)), ((0,1),(2,2)), ((0,1),(3,3)),\right.\\
\nonumber 
& \left.  ((1,2),(0,0)), ((1,2),(1,1)), ((1,2),(2,2)), ((1,2),(3,3)), \right.\\
\nonumber 
& \left.  ((2,3),(0,0)), ((2,3),(1,1)), ((2,3),(2,2)), ((2,3),(3,3))\right\}\\
\nonumber
\mathcal{C}^{\left\{l_{2},l_{1}\right\}}=&\left\{((0,2),(0,1)), ((0,2),(1,2)), ((0,2),(2,3)), \right.\\
\nonumber 
& \left.  ((1,3),(0,1)), ((1,3),(1,2)), ((1,3),(2,3)), \right.\\
\nonumber 
& \left.  ((3,0),(0,1)), ((3,0),(1,2)), ((3,0),(2,3))\right\}\\
\nonumber
\mathcal{C}^{\left\{l_{2},l_{2}\right\}}=&\left\{((0,2),(0,2)), ((0,2),(1,3)), ((0,2),(3,0)), \right.\\
\nonumber 
& \left.  ((1,3),(0,2)), ((1,3),(1,3)), ((1,3),(3,0)), \right.\\
\nonumber 
& \left.  ((3,0),(0,2)), ((3,0),(1,3)), ((3,0),(3,0))\right\}
\end{align}
\begin{align}
\nonumber
\mathcal{C}^{\left\{l_{2},l_{3}\right\}}=&\left\{((0,2),(0,3)), ((0,2),(2,0)), ((0,2),(3,1)), \right.\\
\nonumber 
& \left.  ((1,3),(0,3)), ((1,3),(2,0)), ((1,3),(3,1)), \right.\\
\nonumber 
& \left.  ((3,0),(0,3)), ((3,0),(2,0)), ((3,0),(3,1))\right\}\\
\nonumber
\mathcal{C}^{\left\{l_{2},l_{4}\right\}}=&\left\{((0,2),(1,0)), ((0,2),(2,1)), ((0,2),(3,2)), \right.\\
\nonumber 
& \left.  ((1,3),(1,0)), ((1,3),(2,1)), ((1,3),(3,2)), \right.\\
\nonumber 
& \left.  ((3,0),(1,0)), ((3,0),(2,1)), ((3,0),(3,2))\right\}\\
\nonumber
\mathcal{C}^{\left\{l_{2},l_{5}\right\}}=&\left\{((0,2),(0,0)), ((0,2),(1,1)), ((0,2),(2,2)), ((0,2),(3,3)), \right.\\
\nonumber 
& \left.  ((1,3),(0,0)), ((1,3),(1,1)), ((1,3),(2,2)), ((1,3),(3,3)), \right.\\
\nonumber 
& \left.  ((3,0),(0,0)), ((3,0),(1,1)), ((3,0),(2,2)), ((3,0),(3,3))\right\}\\
\nonumber
\mathcal{C}^{\left\{l_{3},l_{1}\right\}}=&\left\{((0,3),(0,1)), ((0,3),(1,2)), ((0,3),(2,3)), \right.\\
\nonumber 
& \left.  ((2,0),(0,1)), ((2,0),(1,2)), ((2,0),(2,3)), \right.\\
\nonumber 
& \left.  ((3,1),(0,1)), ((3,1),(1,2)), ((3,1),(2,3))\right\}\\
\nonumber
\mathcal{C}^{\left\{l_{3},l_{2}\right\}}=&\left\{((0,3),(0,2)), ((0,3),(1,3)), ((0,3),(3,0)), \right.\\
\nonumber 
& \left.  ((2,0),(0,2)), ((2,0),(1,3)), ((2,0),(3,0)), \right.\\
\nonumber 
& \left.  ((3,1),(0,2)), ((3,1),(1,3)), ((3,1),(3,0))\right\}\\
\nonumber
\mathcal{C}^{\left\{l_{3},l_{3}\right\}}=&\left\{((0,3),(0,3)), ((0,3),(2,0)), ((0,3),(3,1)), \right.\\
\nonumber 
& \left.  ((2,0),(0,3)), ((2,0),(2,0)), ((2,0),(3,1)), \right.\\
\nonumber 
& \left.  ((3,1),(0,3)), ((3,1),(2,0)), ((3,1),(3,1))\right\}\\
\nonumber
\mathcal{C}^{\left\{l_{3},l_{4}\right\}}=&\left\{((0,3),(1,0)), ((0,3),(2,1)), ((0,3),(3,2)), \right.\\
\nonumber 
& \left.  ((2,0),(1,0)), ((2,0),(2,1)), ((2,0),(3,2)), \right.\\
\nonumber 
& \left.  ((3,1),(1,0)), ((3,1),(2,1)), ((3,1),(3,2))\right\}\\
\nonumber
\mathcal{C}^{\left\{l_{3},l_{5}\right\}}=&\left\{((0,3),(0,0)), ((0,3),(1,1)), ((0,3),(2,2)), ((0,3),(3,3)), \right.\\
\nonumber 
& \left.  ((2,0),(0,0)), ((2,0),(1,1)), ((2,0),(2,2)), ((2,0),(3,3)), \right.\\
\nonumber 
& \left.  ((3,1),(0,0)), ((3,1),(1,1)), ((3,1),(2,2)), ((3,1),(3,3))\right\}\\
\nonumber
\mathcal{C}^{\left\{l_{4},l_{1}\right\}}=&\left\{((1,0),(0,1)), ((1,0),(1,2)), ((1,0),(2,3)), \right.\\
\nonumber 
& \left.  ((2,1),(0,1)), ((2,1),(1,2)), ((2,1),(2,3)), \right.\\
\nonumber 
& \left.  ((3,2),(0,1)), ((3,2),(1,2)), ((3,2),(2,3))\right\}\\
\nonumber
\mathcal{C}^{\left\{l_{4},l_{2}\right\}}=&\left\{((1,0),(0,2)), ((1,0),(1,3)), ((1,0),(3,0)), \right.\\
\nonumber 
& \left.  ((2,1),(0,2)), ((2,1),(1,3)), ((2,1),(3,0)), \right.\\
\nonumber 
& \left.  ((3,2),(0,2)), ((3,2),(1,3)), ((3,2),(3,0))\right\}\\
\nonumber
\mathcal{C}^{\left\{l_{4},l_{3}\right\}}=&\left\{((1,0),(0,3)), ((1,0),(2,0)), ((1,0),(3,1)), \right.\\
\nonumber 
& \left.  ((2,1),(0,3)), ((2,1),(2,0)), ((2,1),(3,1)), \right.\\
\nonumber 
& \left.  ((3,2),(0,3)), ((3,2),(2,0)), ((3,2),(3,1))\right\}\\
\nonumber
\mathcal{C}^{\left\{l_{4},l_{4}\right\}}=&\left\{((1,0),(1,0)), ((1,0),(2,1)), ((1,0),(3,2)), \right.\\
\nonumber 
& \left.  ((2,1),(1,0)), ((2,1),(2,1)), ((2,1),(3,2)), \right.\\
\nonumber 
& \left.  ((3,2),(1,0)), ((3,2),(2,1)), ((3,2),(3,2))\right\}\\
\nonumber
\mathcal{C}^{\left\{l_{4},l_{5}\right\}}=&\left\{((1,0),(0,0)), ((1,0),(1,1)), ((1,0),(2,2)), ((1,0),(3,3)), \right.\\
\nonumber 
& \left.  ((2,1),(0,0)), ((2,1),(1,1)), ((2,1),(2,2)), ((2,1),(3,3)), \right.\\
\nonumber 
& \left.  ((3,2),(0,0)), ((3,2),(1,1)), ((3,2),(2,2)), ((3,2),(3,3))\right\}\\
\nonumber
\mathcal{C}^{\left\{l_{5},l_{1}\right\}}=&\left\{((0,0),(0,1)), ((0,0),(1,2)), ((0,0),(2,3)), ((1,1),(0,1)), \right.\\
\nonumber 
& \left.  ((1,1),(1,2)), ((1,1),(2,3)), ((2,2),(0,1)), ((2,2),(1,2)), \right.\\
\nonumber 
& \left.  ((2,2),(2,3)), ((3,3),(0,1)), ((3,3),(1,2)), ((3,3),(2,3))\right\}\\
\nonumber
\mathcal{C}^{\left\{l_{5},l_{2}\right\}}=&\left\{((0,0),(0,2)), ((0,0),(1,3)), ((0,0),(3,0)), ((1,1),(0,2)), \right.\\
\nonumber 
& \left.  ((1,1),(1,3)), ((1,1),(3,0)), ((2,2),(0,2)), ((2,2),(1,3)), \right.\\
\nonumber 
& \left.  ((2,2),(3,0)), ((3,3),(0,2)), ((3,3),(1,3)), ((3,3),(3,0))\right\}\\
\nonumber
\mathcal{C}^{\left\{l_{5},l_{3}\right\}}=&\left\{((0,0),(0,3)), ((0,0),(2,0)), ((0,0),(3,1)), ((1,1),(0,3)), \right.\\
\nonumber 
& \left.  ((1,1),(2,0)), ((1,1),(3,1)), ((2,2),(0,3)), ((2,2),(2,0)), \right.\\
\nonumber 
& \left.  ((2,2),(3,1)), ((3,3),(0,3)), ((3,3),(2,0)), ((3,3),(3,1))\right\}\\
\nonumber
\mathcal{C}^{\left\{l_{5},l_{4}\right\}}=&\left\{((0,0),(1,0)), ((0,0),2,1)), ((0,0),(3,2)), ((1,1),(1,0)), \right.\\
\nonumber 
& \left.  ((1,1),(2,1)), ((1,1),(3,2)), ((2,2),(1,0)), ((2,2),(2,1)), \right.\\
\nonumber 
& \left.  ((2,2),(3,2)), ((3,3),(1,0)), ((3,3),(2,1)), ((3,3),(3,2))\right\}\\
\nonumber
\mathcal{C}^{\left\{l_{5},l_{5}\right\}}=&\left\{((0,0),(0,0)), ((0,0),(1,1)), ((0,0),(2,2)), ((0,0),(3,3)), \right.\\
\nonumber 
& \left.  ((1,1),(0,0)), ((1,1),(1,1)), ((1,1),(2,2)), ((1,1),(3,3)), \right.\\
\nonumber 
& \left.  ((2,2),(0,0)), ((2,2),(1,1)), ((2,2),(2,2)), ((2,2),(3,3)), \right.\\
\nonumber 
& \left.  ((3,3),(0,0)), ((3,3),(1,1)), ((3,3),(2,2)), ((3,3),(3,3)) \right\}
\end{align}
}

The entries in the above constraints are of the form $ ((x_{A_{1}},x_{B_{1}}),(x_{A_{2}},x_{B_{2}}))$. In order to represent these clusters by a Latin Square of side $16$, with $ (x_{A_{1}},x_{A_{2}})$ along the rows, and $(x_{B_{1}},x_{B_{2}})$ along the columns each entry must be made of the form $((x_{A_{1}},x_{A_{2}}),(x_{B_{1}},x_{B_{2}}))$. Thus the constraints on the Latin Square as dictated by the clusters above are as follows:

{\footnotesize
\begin{align}
\nonumber
\vspace{-0.5cm}
\mathcal{L}_{1}:=&\left\{((0,0),(1,1)), ((0,1),(1,2)), ((0,2),(1,3)), \right.\\
\nonumber 
& \left.  ((1,0),(2,1)), ((1,1),(2,2)), ((1,2),(2,3)), \right.\\
\nonumber 
& \left.  ((2,0),(3,1)), ((2,1),(3,2)), ((2,2),(3,3))\right\}\\
\nonumber
\mathcal{L}_{2}:=&\left\{((0,0),(1,2)), ((0,1),(1,3)), ((0,3),(1,0)), \right.\\
\nonumber 
& \left.  ((1,0),(2,2)), ((1,1),(2,3)), ((1,3),(2,0)), \right.\\
\nonumber 
& \left.  ((2,0),(3,2)), ((2,1),(3,3)), ((2,3),(3,0))\right\}\\
\nonumber
\mathcal{L}_{3}:=&\left\{((0,0),(1,3)), ((0,2),(1,0)), ((0,3),(1,1)), \right.\\
\nonumber 
& \left.  ((1,0),(1,3)), ((1,2),(2,0)), ((1,3),(2,1)), \right.\\
\nonumber 
& \left.  ((2,0),(3,3)), ((2,2),(3,0)), ((2,3),(3,1))\right\}\\
\nonumber
\mathcal{L}_{4}:=&\left\{((0,1),(1,0)), ((0,2),(1,1)), ((0,3),(1,2)), \right.\\
\nonumber 
& \left.  ((1,1),(2,0)), ((1,2),(2,1)), ((1,3),(2,2)), \right.\\
\nonumber 
& \left.  ((2,1),(3,0)), ((2,2),(3,1)), ((2,3),(3,2))\right\}\\
\nonumber
\mathcal{L}_{5}:=&\left\{((0,0),(1,0)), ((0,1),(1,1)), ((0,2),(1,2)), ((0,3),(1,3)),\right.\\
\nonumber 
& \left.  ((1,0),(2,0)), ((1,1),(2,1)), ((1,2),(2,2)), ((1,3),(2,3)), \right.\\
\nonumber 
& \left.  ((2,0),(3,0)), ((2,1),(3,1)), ((2,2),(3,2)), ((2,3),(3,3))\right\}\\
\nonumber
\mathcal{L}_{6}:=&\left\{((0,0),(2,1)), ((0,1),(2,2)), ((0,2),(2,3)), \right.\\
\nonumber 
& \left.  ((1,0),(3,1)), ((1,1),(3,2)), ((1,2),(3,3)), \right.\\
\nonumber 
& \left.  ((3,0),(0,1)), ((3,1),(0,2)), ((3,2),(0,3))\right\}\\
\nonumber
\mathcal{L}_{7}:=&\left\{((0,0),(2,2)), ((0,1),(2,3)), ((0,3),(2,0)), \right.\\
\nonumber 
& \left.  ((1,0),(3,2)), ((1,1),(3,3)), ((1,3),(3,0)), \right.\\
\nonumber 
& \left.  ((3,0),(0,2)), ((3,1),(0,3)), ((3,3),(0,0))\right\}\\
\nonumber
\mathcal{L}_{8}:=&\left\{((0,0),(2,3)), ((0,2),(2,0)), ((0,3),(2,1)), \right.\\
\nonumber 
& \left.  ((1,0),(3,3)), ((1,2),(3,0)), ((1,3),(3,1)), \right.\\
\nonumber 
& \left.  ((3,0),(0,3)), ((3,2),(0,0)), ((3,3),(0,1))\right\}\\
\nonumber
\mathcal{L}_{9}:=&\left\{((0,1),(2,0)), ((0,2),(2,1)), ((0,3),(2,2)), \right.\\
\nonumber 
& \left.  ((1,1),(3,0)), ((1,2),(3,1)), ((1,3),(3,2)), \right.\\
\nonumber 
& \left.  ((3,1),(0,0)), ((3,2),(0,1)), ((3,3),(0,2))\right\}\\
\nonumber
\mathcal{L}_{10}:=&\left\{((0,0),(2,0)), ((0,1),(2,1)), ((0,2),(2,2)), ((0,3),(2,3)), \right.\\
\nonumber 
& \left.  ((1,0),(3,0)), ((1,1),(3,1)), ((1,2),(3,2)), ((1,3),(3,3)), \right.\\
\nonumber 
& \left.  ((3,0),(0,0)), ((3,1),(0,1)), ((3,2),(0,2)), ((3,3),(0,3))\right\}\\
\nonumber
\mathcal{L}_{11}:=&\left\{((0,0),(3,1)), ((0,1),(3,2)), ((0,2),(3,3)), \right.\\
\nonumber 
& \left.  ((2,0),(0,1)), ((2,1),(0,2)), ((2,2),(0,3)), \right.\\
\nonumber 
& \left.  ((3,0),(1,1)), ((3,1),(1,2)), ((3,2),(1,3))\right\}\\
\nonumber
\mathcal{L}_{12}:=&\left\{((0,0),(3,2)), ((0,1),(3,3)), ((0,3),(3,0)), \right.\\
\nonumber 
& \left.  ((2,0),(0,2)), ((2,1),(0,3)), ((2,3),(0,0)), \right.\\
\nonumber 
& \left.  ((3,0),(1,2)), ((3,1),(1,3)), ((3,3),(1,0))\right\}\\
\nonumber
\mathcal{L}_{13}:=&\left\{((0,0),(3,3)), ((0,2),(3,0)), ((0,3),(3,1)), \right.\\
\nonumber 
& \left.  ((2,0),(0,3)), ((2,2),(0,0)), ((2,3),(0,1)), \right.\\
\nonumber 
& \left.  ((3,0),(1,3)), ((3,2),(1,0)), ((3,3),(1,1))\right\}\\
\nonumber
\mathcal{L}_{14}:=&\left\{((0,1),(3,0)), ((0,2),(3,1)), ((0,3),(3,2)), \right.\\
\nonumber 
& \left.  ((2,1),(0,0)), ((2,2),(0,1)), ((2,3),(0,2)), \right.\\
\nonumber 
& \left.  ((3,1),(1,0)), ((3,2),(1,1)), ((3,3),(1,2))\right\}\\
\nonumber
\mathcal{L}_{15}:=&\left\{((0,0),(3,0)), ((0,1),(3,1)), ((0,2),(3,2)), ((0,3),(3,3)), \right.\\
\nonumber 
& \left.  ((2,0),(0,0)), ((2,1),(0,1)), ((2,2),(0,2)), ((2,3),(0,3)), \right.\\
\nonumber 
& \left.  ((3,0),(1,0)), ((3,1),(1,1)), ((3,2),(1,2)), ((3,3),(1,3))\right\}\\
\nonumber
\mathcal{L}_{16}:=&\left\{((1,0),(0,1)), ((1,1),(0,2)), ((1,2),(0,3)), \right.\\
\nonumber 
& \left.  ((2,0),(1,1)), ((2,1),(1,2)), ((2,2),(1,3)), \right.\\
\nonumber 
& \left.  ((3,0),(2,1)), ((3,1),(2,2)), ((3,2),(2,3))\right\}
\end{align}
\begin{align}
\nonumber
\mathcal{L}_{17}:=&\left\{((1,0),(0,2)), ((1,1),(0,3)), ((1,3),(0,0)), \right.\\
\nonumber 
& \left.  ((2,0),(1,2)), ((2,1),(1,3)), ((2,3),(1,0)), \right.\\
\nonumber 
& \left.  ((3,0),(2,2)), ((3,1),(2,3)), ((3,3),(2,0))\right\}\\
\nonumber
\mathcal{L}_{18}:=&\left\{((1,0),(0,3)), ((1,2),(0,0)), ((1,3),(0,1)), \right.\\
\nonumber 
& \left.  ((2,0),(1,3)), ((2,2),(1,0)), ((2,3),(1,1)), \right.\\
\nonumber 
& \left.  ((3,0),(1,3)), ((3,2),(2,0)), ((3,3),(2,1))\right\}\\
\nonumber
\mathcal{L}_{19}:=&\left\{((1,1),(0,0)), ((1,2),(0,1)), ((1,3),(0,2)), \right.\\
\nonumber 
& \left.  ((2,1),(1,0)), ((2,2),(1,1)), ((2,3),(1,2)), \right.\\
\nonumber 
& \left.  ((3,1),(2,0)), ((3,2),(2,1)), ((3,3),(2,2))\right\}\\
\nonumber
\mathcal{L}_{20}:=&\left\{((1,0),(0,0)), ((1,1),(0,1)), ((1,2),(0,2)), ((1,3),(0,3)), \right.\\
\nonumber 
& \left.  ((2,0),(1,0)), ((2,1),(1,1)), ((2,2),(1,2)), ((2,3),(1,3)), \right.\\
\nonumber 
& \left.  ((3,0),(2,0)), ((3,1),(2,1)), ((3,2),(2,2)), ((3,3),(2,3))\right\}\\
\nonumber
\mathcal{L}_{21}:=&\left\{((0,0),(0,1)), ((0,1),(0,2)), ((0,2),(0,3)), ((1,0),(1,1)), \right.\\
\nonumber 
& \left.  ((1,1),(1,2)), ((1,2),(1,3)), ((2,0),(2,1)), ((2,1),(2,2)), \right.\\
\nonumber 
& \left.  ((2,2),(2,3)), ((3,0),(3,1)), ((3,1),(3,2)), ((3,2),(3,3))\right\}\\
\nonumber
\mathcal{L}_{22}:=&\left\{((0,0),(0,2)), ((0,1),(0,3)), ((0,3),(0,0)), ((1,0),(1,2)), \right.\\
\nonumber 
& \left.  ((1,1),(1,3)), ((1,3),(1,0)), ((2,0),(2,2)), ((2,1),(2,3)), \right.\\
\nonumber 
& \left.  ((2,3),(2,0)), ((3,0),(3,2)), ((3,1),(3,3)), ((3,3),(3,0))\right\}\\
\nonumber
\mathcal{L}_{23}:=&\left\{((0,0),(0,3)), ((0,2),(0,0)), ((0,3),(0,1)), ((1,0),(1,3)), \right.\\
\nonumber 
& \left.  ((1,2),(1,0)), ((1,3),(1,1)), ((2,0),(2,3)), ((2,2),(2,0)), \right.\\
\nonumber 
& \left.  ((2,3),(2,1)), ((3,0),(3,3)), ((3,2),(3,0)), ((3,3),(3,1))\right\}\\
\nonumber
\mathcal{L}_{24}:=&\left\{((0,1),(0,0)), ((0,2),(0,1)), ((0,3),(0,2)), ((1,1),(1,0)), \right.\\
\nonumber 
& \left.  ((1,2),(1,1)), ((1,3),(1,2)), ((2,1),(2,0)), ((2,2),(2,1)), \right.\\
\nonumber 
& \left.  ((2,3),(2,2)), ((3,1),(3,0)), ((3,2),(3,1)), ((3,3),(3,2))\right\}\\
\nonumber
\mathcal{L}_{25}:=&\left\{((0,0),(0,0)), ((0,1),(0,1)), ((0,2),(0,2)), ((0,3),(0,3)), \right.\\
\nonumber 
& \left.  ((1,0),(1,0)), ((1,1),(1,1)), ((1,2),(1,2)), ((1,3),(1,3)), \right.\\
\nonumber 
& \left.  ((2,0),(2,0)), ((2,1),(2,1)), ((2,2),(2,2)), ((2,3),(2,3)), \right.\\
\nonumber 
& \left.  ((3,0),(3,0)), ((3,1),(3,1)), ((3,2),(3,2)), ((3,3),(3,3)) \right\}\\
\nonumber
\end{align}
}

\section{}
\noindent \textit{Clustering that removes the singular fade state $1+j$:}\\
The Cartesian Product of the clustering $\mathcal{C}^{\left[1+j\right]}$ with itself denoted by $\mathcal{D}^{\left[1+j\right]}$ contains exactly $25$ clusters. The clusters in $\mathcal{D}^{\left[1+j\right]}$ are as follows. 

{\footnotesize
\begin{align}
\nonumber
\vspace{-0.5cm}
\mathcal{C}^{\left\{l_{1},l_{1}\right\}}=&\left\{((0,1),(0,1)), ((0,1),(2,3)), ((0,1),(3,0)), \right.\\
\nonumber 
& \left.  ((2,3),(0,1)), ((2,3),(2,3)), ((2,3),(3,0))\right.\\
\nonumber
& \left.  ((3,0),(0,1)), ((3,0),(2,3)), ((3,0),(3,0)), \right\}\\
\nonumber 
\mathcal{C}^{\left\{l_{1},l_{2}\right\}}=&\left\{((0,1),(0,3)), ((0,1),(1,0)), ((0,1),(3,2)), \right.\\
\nonumber 
& \left.  ((2,3),(0,3)), ((2,3),(1,0)), ((2,3),(3,2))\right.\\
\nonumber
& \left.  ((3,0),(0,3)), ((3,0),(1,0)), ((3,0),(3,2)), \right\}\\
\nonumber 
\mathcal{C}^{\left\{l_{1},l_{3}\right\}}=&\left\{((0,1),(1,2)), ((0,1),(2,0)), ((0,1),(3,1)), \right.\\
\nonumber 
& \left.  ((2,3),(1,2)), ((2,3),(2,0)), ((2,3),(3,1))\right.\\
\nonumber
& \left.  ((3,0),(1,2)), ((3,0),(2,0)), ((3,0),(3,1)), \right\}\\
\nonumber 
\mathcal{C}^{\left\{l_{1},l_{4}\right\}}=&\left\{((0,1),(0,2)), ((0,1),(1,3)), ((0,1),(2,1)), \right.\\
\nonumber 
& \left.  ((2,3),(0,2)), ((2,3),(1,3)), ((2,3),(2,1))\right.\\
\nonumber
& \left.  ((3,0),(0,2)), ((3,0),(1,3)), ((3,0),(2,1)), \right\}\\
\nonumber 
\mathcal{C}^{\left\{l_{1},l_{5}\right\}}=&\left\{((0,1),(0,0)), ((0,1),(1,1)), ((0,1),(2,2)), ((0,1),(3,3)),\right.\\
\nonumber 
& \left.  ((2,3),(0,0)), ((2,3),(1,1)), ((2,3),(2,2)), ((2,3),(3,3)),\right.\\
\nonumber
& \left.  ((3,0),(0,0)), ((3,0),(1,1)), ((3,0),(2,2)), ((3,0),(3,3))\right\}\\
\nonumber
\mathcal{C}^{\left\{l_{2},l_{1}\right\}}=&\left\{((0,3),(0,1)), ((0,3),(2,3)), ((0,3),(3,0)), \right.\\
\nonumber 
& \left.  ((1,0),(0,1)), ((1,0),(2,3)), ((1,0),(3,0))\right.\\
\nonumber
& \left.  ((3,2),(0,1)), ((3,2),(2,3)), ((3,2),(3,0)), \right\}
\end{align}
\begin{align}
\nonumber 
\mathcal{C}^{\left\{l_{2},l_{2}\right\}}=&\left\{((0,3),(0,3)), ((0,3),(1,0)), ((0,3),(3,2)), \right.\\
\nonumber 
& \left.  ((1,0),(0,3)), ((1,0),(1,0)), ((1,0),(3,2))\right.\\
\nonumber
& \left.  ((3,2),(0,3)), ((3,2),(1,0)), ((3,2),(3,2)), \right\}\\
\nonumber 
\mathcal{C}^{\left\{l_{2},l_{3}\right\}}=&\left\{((0,3),(1,2)), ((0,3),(2,0)), ((0,3),(3,1)), \right.\\
\nonumber 
& \left.  ((1,0),(1,2)), ((1,0),(2,0)), ((1,0),(3,1))\right.\\
\nonumber
& \left.  ((3,2),(1,2)), ((3,2),(2,0)), ((3,2),(3,1)), \right\}\\
\nonumber 
\mathcal{C}^{\left\{l_{2},l_{4}\right\}}=&\left\{((0,3),(0,2)), ((0,3),(1,3)), ((0,3),(2,1)), \right.\\
\nonumber 
& \left.  ((1,0),(0,2)), ((1,0),(1,3)), ((1,0),(2,1))\right.\\
\nonumber
& \left.  ((3,2),(0,2)), ((3,2),(1,3)), ((3,2),(2,1)), \right\}\\
\nonumber 
\mathcal{C}^{\left\{l_{2},l_{5}\right\}}=&\left\{((0,3),(0,0)), ((0,3),(1,1)), ((0,3),(2,2)), ((0,3),(3,3)),\right.\\
\nonumber 
& \left.  ((1,0),(0,0)), ((1,0),(1,1)), ((1,0),(2,2)), ((1,0),(3,3)),\right.\\
\nonumber
& \left.  ((3,2),(0,0)), ((3,2),(1,1)), ((3,2),(2,2)), ((3,2),(3,3))\right\}\\
\nonumber
\mathcal{C}^{\left\{l_{3},l_{1}\right\}}=&\left\{((1,2),(0,1)), ((1,2),(2,3)), ((1,2),(3,0)), \right.\\
\nonumber 
& \left.  ((2,0),(0,1)), ((2,0),(2,3)), ((2,0),(3,0))\right.\\
\nonumber
& \left.  ((3,1),(0,1)), ((3,1),(2,3)), ((3,1),(3,0)), \right\}\\
\nonumber 
\mathcal{C}^{\left\{l_{3},l_{2}\right\}}=&\left\{((1,2),(0,3)), ((1,2),(1,0)), ((1,2),(3,2)), \right.\\
\nonumber 
& \left.  ((2,0),(0,3)), ((2,0),(1,0)), ((2,0),(3,2))\right.\\
\nonumber
& \left.  ((3,1),(0,3)), ((3,1),(1,0)), ((3,1),(3,2)), \right\}\\
\nonumber 
\mathcal{C}^{\left\{l_{3},l_{3}\right\}}=&\left\{((1,2),(1,2)), ((1,2),(2,0)), ((1,2),(3,1)), \right.\\
\nonumber 
& \left.  ((2,0),(1,2)), ((2,0),(2,0)), ((2,0),(3,1))\right.\\
\nonumber
& \left.  ((3,1),(1,2)), ((3,1),(2,0)), ((3,1),(3,1)), \right\}\\
\nonumber 
\mathcal{C}^{\left\{l_{3},l_{4}\right\}}=&\left\{((1,2),(0,2)), ((1,2),(1,3)), ((1,2),(2,1)), \right.\\
\nonumber 
& \left.  ((2,0),(0,2)), ((2,0),(1,3)), ((2,0),(2,1))\right.\\
\nonumber
& \left.  ((3,1),(0,2)), ((3,1),(1,3)), ((3,1),(2,1)), \right\}\\
\nonumber 
\mathcal{C}^{\left\{l_{3},l_{5}\right\}}=&\left\{((1,2),(0,0)), ((1,2),(1,1)), ((1,2),(2,2)), ((1,2),(3,3)),\right.\\
\nonumber 
& \left.  ((2,0),(0,0)), ((2,0),(1,1)), ((2,0),(2,2)), ((2,0),(3,3)),\right.\\
\nonumber
& \left.  ((3,1),(0,0)), ((3,1),(1,1)), ((3,1),(2,2)), ((3,1),(3,3))\right\}\\
\nonumber
\mathcal{C}^{\left\{l_{4},l_{1}\right\}}=&\left\{((0,2),(0,1)), ((0,2),(2,3)), ((0,2),(3,0)), \right.\\
\nonumber 
& \left.  ((1,3),(0,1)), ((1,3),(2,3)), ((1,3),(3,0))\right.\\
\nonumber
& \left.  ((2,1),(0,1)), ((2,1),(2,3)), ((2,1),(3,0)), \right\}\\
\nonumber 
\mathcal{C}^{\left\{l_{4},l_{2}\right\}}=&\left\{((0,2),(0,3)), ((0,2),(1,0)), ((0,2),(3,2)), \right.\\
\nonumber 
& \left.  ((1,3),(0,3)), ((1,3),(1,0)), ((1,3),(3,2))\right.\\
\nonumber
& \left.  ((2,1),(0,3)), ((2,1),(1,0)), ((2,1),(3,2)), \right\}\\
\nonumber 
\mathcal{C}^{\left\{l_{4},l_{3}\right\}}=&\left\{((0,2),(1,2)), ((0,2),(2,0)), ((0,2),(3,1)), \right.\\
\nonumber 
& \left.  ((1,3),(1,2)), ((1,3),(2,0)), ((1,3),(3,1))\right.\\
\nonumber
& \left.  ((2,1),(1,2)), ((2,1),(2,0)), ((2,1),(3,1)), \right\}\\
\nonumber 
\mathcal{C}^{\left\{l_{4},l_{4}\right\}}=&\left\{((0,2),(0,2)), ((0,2),(1,3)), ((0,2),(2,1)), \right.\\
\nonumber 
& \left.  ((1,3),(0,2)), ((1,3),(1,3)), ((1,3),(2,1))\right.\\
\nonumber
& \left.  ((2,1),(0,2)), ((2,1),(1,3)), ((2,1),(2,1)), \right\}\\
\nonumber
\mathcal{C}^{\left\{l_{4},l_{5}\right\}}=&\left\{((0,2),(0,0)), ((0,2),(1,1)), ((0,2),(2,2)), ((0,2),(3,3)),\right.\\
\nonumber 
& \left.  ((1,3),(0,0)), ((1,3),(1,1)), ((1,3),(2,2)), ((1,3),(3,3)),\right.\\
\nonumber
& \left.  ((2,1),(0,0)), ((2,1),(1,1)), ((2,1),(2,2)), ((2,1),(3,3))\right\}\\
\nonumber
\mathcal{C}^{\left\{l_{5},l_{1}\right\}}=&\left\{((0,0),(0,1)), ((0,0),(2,3)), ((0,0),(3,0)), ((1,1),(0,1)), \right.\\
\nonumber 
& \left.  ((1,1),(2,3)), ((1,1),(3,0)), ((2,2),(0,1)), ((2,2),(2,3)), \right.\\
\nonumber 
& \left.  ((2,2),(3,0)), ((3,3),(0,1)), ((3,3),(2,3)), ((3,3),(3,0))\right\}\\
\nonumber
\mathcal{C}^{\left\{l_{5},l_{2}\right\}}=&\left\{((0,0),(0,3)), ((0,0),(1,0)), ((0,0),(3,2)), ((1,1),(0,3)), \right.\\
\nonumber 
& \left.  ((1,1),(1,0)), ((1,1),(3,2)), ((2,2),(0,3)), ((2,2),(1,0)), \right.\\
\nonumber 
& \left.  ((2,2),(3,2)), ((3,3),(0,3)), ((3,3),(1,0)), ((3,3),(3,2))\right\}\\
\nonumber
\mathcal{C}^{\left\{l_{5},l_{3}\right\}}=&\left\{((0,0),(1,2)), ((0,0),(2,0)), ((0,0),(3,1)), ((1,1),(1,2)), \right.\\
\nonumber 
& \left.  ((1,1),(2,0)), ((1,1),(3,1)), ((2,2),(1,2)), ((2,2),(2,0)), \right.\\
\nonumber 
& \left.  ((2,2),(3,1)), ((3,3),(1,2)), ((3,3),(2,0)), ((3,3),(3,1))\right\}\\
\nonumber
\mathcal{C}^{\left\{l_{5},l_{4}\right\}}=&\left\{((0,0),(0,2)), ((0,0),1,3)), ((0,0),(2,1)), ((1,1),(0,2)), \right.\\
\nonumber 
& \left.  ((1,1),(1,3)), ((1,1),(2,1)), ((2,2),(0,2)), ((2,2),(1,3)), \right.\\
\nonumber 
& \left.  ((2,2),(2,1)), ((3,3),(0,2)), ((3,3),(1,3)), ((3,3),(2,1))\right\}\\
\nonumber
\end{align}
\begin{align}
\nonumber
\mathcal{C}^{\left\{l_{5},l_{5}\right\}}=&\left\{((0,0),(0,0)), ((0,0),(1,1)), ((0,0),(2,2)), ((0,0),(3,3)), \right.\\
\nonumber 
& \left.  ((1,1),(0,0)), ((1,1),(1,1)), ((1,1),(2,2)), ((1,1),(3,3)), \right.\\
\nonumber 
& \left.  ((2,2),(0,0)), ((2,2),(1,1)), ((2,2),(2,2)), ((2,2),(3,3)), \right.\\%
\nonumber 
& \left.  ((3,3),(0,0)), ((3,3),(1,1)), ((3,3),(2,2)), ((3,3),(3,3)) \right\}\\
\nonumber
\end{align}
}

The constraints for the Latin Square representing the above clustering are as follows:

{\footnotesize
\begin{align}
\nonumber
\vspace{-0.5cm}
\mathcal{L}_{1}:=&\left\{((0,0),(1,1)), ((0,2),(1,3)), ((0,3),(1,0)), \right.\\
\nonumber 
& \left.  ((2,0),(3,1)), ((2,2),(3,3)), ((2,3),(3,0))\right.\\
\nonumber
& \left.  ((3,0),(0,1)), ((3,2),(0,3)), ((3,3),(0,0)), \right\}\\
\nonumber 
\mathcal{L}_{2}:=&\left\{((0,0),(1,3)), ((0,1),(1,0)), ((0,3),(1,2)), \right.\\
\nonumber 
& \left.  ((2,0),(3,3)), ((2,1),(3,0)), ((2,3),(3,2))\right.\\
\nonumber
& \left.  ((3,0),(0,3)), ((3,1),(0,0)), ((3,3),(0,2)), \right\}\\
\nonumber 
\mathcal{L}_{3}:=&\left\{((0,1),(1,2)), ((0,2),(1,0)), ((0,3),(1,1)), \right.\\
\nonumber 
& \left.  ((2,1),(3,2)), ((2,2),(3,0)), ((2,3),(3,1))\right.\\
\nonumber
& \left.  ((3,1),(0,2)), ((3,2),(0,0)), ((3,3),(0,1)), \right\}\\
\nonumber 
\mathcal{L}_{4}:=&\left\{((0,0),(1,2)), ((0,1),(1,3)), ((0,2),(1,1)), \right.\\
\nonumber 
& \left.  ((2,0),(3,2)), ((2,1),(3,3)), ((2,2),(3,1))\right.\\
\nonumber
& \left.  ((3,0),(0,2)), ((3,1),(0,3)), ((3,2),(0,1)), \right\}\\
\nonumber 
\mathcal{L}_{5}:=&\left\{((0,0),(1,0)), ((0,1),(1,1)), ((0,2),(1,2)), ((0,3),(1,3)),\right.\\
\nonumber 
& \left.  ((2,0),(3,0)), ((2,1),(3,1)), ((2,2),(3,2)), ((2,3),(3,3)),\right.\\
\nonumber
& \left.  ((3,0),(0,0)), ((3,1),(0,1)), ((3,2),(0,2)), ((3,3),(0,3))\right\}\\
\nonumber
\mathcal{L}_{6}:=&\left\{((0,0),(3,1)), ((0,2),(3,3)), ((0,3),(3,0)), \right.\\
\nonumber 
& \left.  ((1,0),(0,1)), ((1,2),(0,3)), ((1,3),(0,0))\right.\\
\nonumber
& \left.  ((3,0),(2,1)), ((3,2),(2,3)), ((3,3),(2,0)), \right\}\\
\nonumber
\mathcal{L}_{7}:=&\left\{((0,0),(3,3)), ((0,1),(3,0)), ((0,3),(3,2)), \right.\\
\nonumber 
& \left.  ((1,0),(0,3)), ((1,1),(0,0)), ((1,3),(0,2))\right.\\
\nonumber
& \left.  ((3,0),(2,3)), ((3,1),(2,0)), ((3,3),(2,2)), \right\}\\
\nonumber 
\mathcal{L}_{8}:=&\left\{((0,1),(3,2)), ((0,2),(3,0)), ((0,3),(3,1)), \right.\\
\nonumber 
& \left.  ((1,1),(0,2)), ((1,2),(0,0)), ((1,3),(0,1))\right.\\
\nonumber
& \left.  ((3,1),(2,2)), ((3,2),(2,0)), ((3,3),(2,1)), \right\}\\
\nonumber 
\mathcal{L}_{9}:=&\left\{((0,0),(3,2)), ((0,1),(3,3)), ((0,2),(3,1)), \right.\\
\nonumber 
& \left.  ((1,0),(0,2)), ((1,1),(0,3)), ((1,2),(0,1))\right.\\
\nonumber
& \left.  ((3,0),(2,2)), ((3,1),(2,3)), ((3,2),(2,1)), \right\}\\
\nonumber 
\mathcal{L}_{10}:=&\left\{((0,0),(3,0)), ((0,1),(3,1)), ((0,2),(3,2)), ((0,3),(3,3)),\right.\\
\nonumber 
& \left.  ((1,0),(0,0)), ((1,1),(0,1)), ((12),(0,2)), ((1,3),(0,3)),\right.\\
\nonumber
& \left.  ((3,0),(2,0)), ((3,1),(2,1)), ((3,2),(2,2)), ((3,3),(2,3))\right\}\\
\nonumber
\mathcal{L}_{11}:=&\left\{((1,0),(2,1)), ((1,2),(2,3)), ((1,3),(2,0)), \right.\\
\nonumber 
& \left.  ((2,0),(0,1)), ((2,2),(0,3)), ((2,3),(0,0))\right.\\
\nonumber
& \left.  ((3,0),(1,1)), ((3,2),(1,3)), ((3,3),(1,0)), \right\}\\
\nonumber 
\mathcal{L}_{12}:=&\left\{((1,0),(2,3)), ((1,1),(2,0)), ((1,3),(2,2)), \right.\\
\nonumber 
& \left.  ((2,0),(0,3)), ((2,1),(0,0)), ((2,3),(0,2))\right.\\
\nonumber
& \left.  ((3,0),(1,3)), ((3,1),(1,0)), ((3,3),(1,2)), \right\}\\
\nonumber 
\mathcal{L}_{13}:=&\left\{((1,1),(2,2)), ((1,2),(2,0)), ((1,3),(2,1)), \right.\\
\nonumber 
& \left.  ((2,1),(0,2)), ((2,2),(0,0)), ((2,3),(0,1))\right.\\
\nonumber
& \left.  ((3,1),(1,2)), ((3,2),(1,0)), ((3,3),(1,1)), \right\}\\
\nonumber 
\mathcal{L}_{14}:=&\left\{((1,0),(2,2)), ((1,1),(2,3)), ((1,2),(2,1)), \right.\\
\nonumber 
& \left.  ((2,0),(0,2)), ((2,1),(0,3)), ((2,2),(0,1))\right.\\
\nonumber
& \left.  ((3,0),(1,2)), ((3,1),(1,3)), ((3,2),(1,1)), \right\}\\
\nonumber 
\mathcal{L}_{15}:=&\left\{((1,0),(2,0)), ((1,1),(2,1)), ((1,2),(2,2)), ((1,3),(2,3)),\right.\\
\nonumber 
& \left.  ((2,0),(0,0)), ((2,1),(0,1)), ((2,2),(0,2)), ((2,3),(0,3)),\right.\\
\nonumber
& \left.  ((3,0),(1,0)), ((3,1),(1,1)), ((3,2),(1,2)), ((3,3),(1,3))\right\}
\end{align}
\begin{align}
\nonumber
\mathcal{L}_{16}:=&\left\{((0,0),(2,1)), ((0,2),(2,3)), ((0,3),(2,0)), \right.\\
\nonumber 
& \left.  ((1,0),(3,1)), ((1,2),(3,3)), ((1,3),(3,0))\right.\\
\nonumber
& \left.  ((2,0),(1,1)), ((2,2),(1,3)), ((2,3),(1,0)), \right\}\\
\nonumber 
\mathcal{L}_{17}:=&\left\{((0,0),(2,3)), ((0,1),(2,0)), ((0,3),(2,2)), \right.\\
\nonumber 
& \left.  ((1,0),(3,3)), ((1,1),(3,0)), ((1,3),(3,2))\right.\\
\nonumber
& \left.  ((2,0),(1,3)), ((2,1),(1,0)), ((2,3),(1,2)), \right\}\\
\nonumber
\mathcal{L}_{18}:=&\left\{((0,1),(2,2)), ((0,2),(2,0)), ((0,3),(2,1)), \right.\\
\nonumber 
& \left.  ((1,1),(3,2)), ((1,2),(3,0)), ((1,3),(3,1))\right.\\
\nonumber
& \left.  ((2,1),(1,2)), ((2,2),(1,0)), ((2,3),(1,1)), \right\}\\
\nonumber 
\mathcal{L}_{19}:=&\left\{((0,0),(2,2)), ((0,1),(2,3)), ((0,2),(2,1)), \right.\\
\nonumber 
& \left.  ((1,0),(3,2)), ((1,1),(3,3)), ((1,2),(3,1))\right.\\
\nonumber
& \left.  ((2,0),(1,2)), ((2,1),(1,3)), ((2,2),(1,1)), \right\}\\
\nonumber 
\mathcal{L}_{20}:=&\left\{((0,0),(2,0)), ((0,1),(2,1)), ((0,2),(2,2)), ((0,3),(2,3)),\right.\\
\nonumber 
& \left.  ((1,0),(3,0)), ((1,1),(3,1)), ((1,2),(3,2)), ((1,3),(3,3)),\right.\\
\nonumber
& \left.  ((2,0),(1,0)), ((2,1),(1,1)), ((2,2),(1,2)), ((2,3),(1,3))\right\}\\
\nonumber
\mathcal{L}_{21}:=&\left\{((0,0),(0,1)), ((0,2),(0,3)), ((0,3),(0,0)), ((1,0),(1,1)), \right.\\
\nonumber 
& \left.  ((1,2),(1,3)), ((1,3),(1,0)), ((2,0),(2,1)), ((2,2),(2,3)), \right.\\
\nonumber 
& \left.  ((2,3),(2,0)), ((3,0),(3,1)), ((3,2),(3,3)), ((3,3),(3,0))\right\}\\
\nonumber
\mathcal{L}_{22}:=&\left\{((0,0),(0,3)), ((0,1),(0,0)), ((0,3),(0,2)), ((1,0),(1,3)), \right.\\
\nonumber 
& \left.  ((1,1),(1,0)), ((1,3),(1,2)), ((2,0),(2,3)), ((2,1),(2,0)), \right.\\
\nonumber 
& \left.  ((2,3),(2,2)), ((3,0),(3,3)), ((3,1),(3,0)), ((3,3),(3,2))\right\}\\
\nonumber
\mathcal{L}_{23}:=&\left\{((0,1),(0,2)), ((0,2),(0,0)), ((0,3),(0,1)), ((1,1),(1,2)), \right.\\
\nonumber 
& \left.  ((1,2),(1,0)), ((1,3),(1,1)), ((2,1),(2,2)), ((2,2),(2,0)), \right.\\
\nonumber 
& \left.  ((2,3),(2,1)), ((3,1),(3,2)), ((3,2),(3,0)), ((3,3),(3,1))\right\}\\
\nonumber
\mathcal{L}_{24}:=&\left\{((0,0),(0,2)), ((0,1),(0,3)), ((0,2),(0,1)), ((1,0),(1,2)), \right.\\
\nonumber 
& \left.  ((1,1),(1,3)), ((1,2),(1,1)), ((2,0),(2,2)), ((2,1),(2,3)), \right.\\
\nonumber 
& \left.  ((2,2),(2,1)), ((3,0),(3,2)), ((3,1),(3,3)), ((3,2),(3,1))\right\}\\
\nonumber
\mathcal{L}_{25}:=&\left\{((0,0),(0,0)), ((0,1),(0,1)), ((0,2),(0,2)), ((0,3),(0,3)), \right.\\
\nonumber 
& \left.  ((1,0),(1,0)), ((1,1),(1,1)), ((1,2),(1,2)), ((1,3),(1,3)), \right.\\
\nonumber 
& \left.  ((2,0),(2,0)), ((2,1),(2,1)), ((2,2),(2,2)), ((2,3),(2,3)), \right.\\
\nonumber 
& \left.  ((3,0),(3,0)), ((3,1),(3,1)), ((3,2),(3,2)), ((3,3),(3,3)) \right\}\\
\nonumber
\end{align}
}

The Cartesian Product of the clustering $\mathcal{C}^{\left[1+j\right]}$ with itself, denoted by $\mathcal{D}^{\left[1+j\right]}$ can be represented by the Latin Square given in Fig. 8.\\

\section{}
\noindent \textit{Singularity removal constraints for the singular fade state $-1+j$:}\\
The singularity removal constraints for the singular fade state $-1+j$ are given in Fig. 17 in the next page.
\begin{figure*}
{\scriptsize
\begin{tabular}{||c||l|l|l||}\hline
        & ~~~~Singularity Removal Constraints for $\gamma e^{j \theta}=-1+j$ &~~~~~~~~ Latin Square Constraints for $\gamma e^{j \theta}=-1+j$ & Cluster \\\hline \hline

(1) &  $\left\{((0,0),(3,2)), ((3,2),(0,0)), ((0,0),(0,0)), ((3,2),(3,2)) \right\} $ &  $\left\{((0,3),(0,2)), ((3,0),(2,0)), ((0,0),(0,0)), ((3,3),(2,2)) \right\} $ & $\mathcal{L}_{1}$  \\\hline
(2) &  $\left\{((0,0),(0,1)), ((3,2),(3,3)),((3,2),(0,1)),((0,0),(3,3))\right\} $   &  $\left\{((0,0),(0,1)), ((3,3),(2,3)),((3,0),(2,1)),((0,3),(0,3))\right\}$ & $\mathcal{L}_{2}$  \\\hline
(3) &  $\left\{((0,0),(1,1)), ((3,2),(2,0)),((3,2),(1,1)),((0,0),(2,0))\right\} $   &  $\left\{((0,1),(0,1)), ((3,2),(2,0)),((3,1),(2,1)),((0,2),(0,0))\right\} $ &  $\mathcal{L}_{3}$ \\\hline
(4) &  $\left\{((0,0),(1,3)), ((3,2),(2,2)),((3,2),(1,3)),((0,0),(2,2))\right\} $   &  $\left\{((0,1),(0,3)), ((3,2),(2,2)),((3,1),(2,3)),((0,2),(0,2))\right\} $ &  $\mathcal{L}_{4}$  \\\hline
(5) &  $\left\{((0,1),(0,0)), ((3,3),(3,2)),((3,3),(0,0)),((0,1),(3,2))\right\} $   &  $\left\{((0,0),(1,0)), ((3,3),(3,2)),((3,0(,(3,0)),((0,3),(1,2))\right\} $ &   $\mathcal{L}_{4}$ \\\hline
(6) &  $\left\{((0,1),(1,1)), ((3,3),(2,0)),((3,3),(1,1)),((0,1),(2,0)) \right\} $ & $\left\{((0,1),(1,1)), ((3,2),(3,0)),((3,1),(3,1)),((0,2),(1,0)) \right\} $  &  $\mathcal{L}_{1}$  \\\hline
(7) &  $\left\{((0,1),(0,1)), ((3,3),(3,3)),((3,3),(0,1)),((0,1),(3,3)) \right\} $ & $\left\{((0,0),(1,1)), ((3,3),(3,3)),((3,0),(3,1)),((0,3),(1,3)) \right\}$ &  $\mathcal{L}_{3}$ \\\hline
(8) &  $\left\{((0,1),(1,3)), ((3,3),(2,2)),((3,3),(1,3)),((0,1),(2,2)) \right\} $  & $\left\{((0,1),(1,3)), ((3,2),(3,2)),((3,1),(3,3)),((0,2),(1,2)) \right\}$  &  $\mathcal{L}_{2}$ \\\hline
(9) &  $\left\{((1,1),(1,1)), ((2,0),(2,0)), ((2,0),(1,1)), ((1,1),(2,0)) \right\} $ & $\left\{((1,1),(1,1)), ((2,2),(0,0)), ((2,1),(0,1)), ((1,2),(1,0)) \right\}$ & $\mathcal{L}_{5}$  \\\hline
(10) &  $\left\{((1,1),(0,0)), ((2,0),(3,2)), ((2,0),(0,0)), ((1,1),(3,2)) \right\} $  & $\left\{((1,0),(1,0)), ((2,3),(0,2)), ((2,0),(0,0)), ((1,3),(1,2)) \right\}$  &  $\mathcal{L}_{6}$ \\\hline
(11)&  $\left\{((1,1),(0,1)), ((2,0),(3,3)), ((2,0),(0,1)), ((1,1),(3,3)) \right\} $  & $\left\{((1,0),(1,1)), ((2,3),(0,3)), ((2,0),(0,1)), ((1,3),(1,1)) \right\}$  &  $\mathcal{L}_{7}$ \\\hline
(12) &  $\left\{((1,1),(1,3)), ((2,0),(2,2)), ((2,0),(1,3)), ((1,1),(2,2)) \right\} $  &  $\left\{((1,1),(1,3)), ((2,2),(0,2)), ((2,1),(0,3)), ((1,2),(1,2)) \right\}$ &  $\mathcal{L}_{8}$ \\\hline
(13) &  $\left\{((1,3),(0,0)), ((2,2),(3,2)), ((2,2),(0,0)), ((1,3),(3,2)) \right\} $  &  $\left\{((1,0),(3,0)), ((2,3),(2,2)), ((2,0),(2,0)), ((1,3),(3,2)) \right\}$ & $\mathcal{L}_{5}$  \\\hline
(14) &  $\left\{((1,3),(0,1)), ((2,2),(3,3)), ((2,2),(0,1)), ((1,3),(3,3)) \right\} $  &  $\left\{((1,0),(3,1)), ((2,3),(2,3)), ((2,0),(2,1)), ((1,3),(3,3)) \right\}$ &  $\mathcal{L}_{8}$ \\\hline
(15) &  $\left\{((1,3),(1,1)), ((2,2),(2,0)), ((2,2),(1,1)), ((1,3),(2,0)) \right\} $  &  $\left\{((1,1),(3,1)), ((2,2),(2,0)), ((2,1),(2,1)), ((1,2),(3,0)) \right\}$ & $\mathcal{L}_{7}$  \\\hline
(16) &  $\left\{((1,3),(1,3)), ((2,2),(2,2)), ((2,2),(1,3)), ((1,3),(2,2)) \right\} $  & $\left\{((1,1),(3,3)), ((2,2),(2,2)), ((2,1),(2,3)), ((1,2),(3,2)) \right\}$ & $\mathcal{L}_{6}$ \\\hline
(17) &  $\left\{((0,0),(0,2)), ((3,2),(0,2)) \right\} $  & $\left\{((0,0),(0,2)), ((3,0),(2,2))\right\}$ & $\mathcal{L}_{7}$ \\\hline
(18) &  $\left\{((0,0),(0,3)), ((3,2),(0,3)) \right\} $  & $\left\{((0,0),(0,3)), ((3,0),(2,3))\right\}$ & $\mathcal{L}_{5}$ \\\hline
(19) &  $\left\{((0,0),(1,0)), ((3,2),(0,2)) \right\} $  & $\left\{((0,1),(0,0)), ((3,0),(2,2))\right\}$ & $\mathcal{L}_{8}$ \\\hline
(20) &  $\left\{((0,0),(1,2)), ((3,2),(1,2)) \right\} $  & $\left\{((0,1),(0,2)), ((3,1),(2,2))\right\}$ & $\mathcal{L}_{9}$ \\\hline
(21) &  $\left\{((0,0),(2,1)), ((3,2),(2,1)) \right\} $  & $\left\{((0,2),(0,1)), ((3,2),(2,1))\right\}$ & $\mathcal{L}_{6}$ \\\hline
(22) &  $\left\{((0,0),(2,3)), ((3,2),(2,3)) \right\} $  & $\left\{((0,2),(0,3)), ((3,2),(2,3))\right\}$ & $\mathcal{L}_{9}$  \\\hline
(23) &  $\left\{((0,0),(3,0)), ((3,2),(3,0)) \right\} $  & $\left\{((0,3),(0,0)), ((3,3),(2,0))\right\}$ & $\mathcal{L}_{9}$ \\\hline
(24) &  $\left\{((0,0),(3,1)), ((3,2),(3,1)) \right\} $  & $\left\{((0,3),(0,1)), ((3,3),(2,1))\right\}$ & $\mathcal{L}_{10}$ \\\hline
(25) &  $\left\{((0,2),(0,0)), ((0,2),(3,2)) \right\} $  & $\left\{((0,0),(2,0)), ((0,3),(2,2))\right\}$ & $\mathcal{L}_{8}$ \\\hline
(26) &  $\left\{((0,3),(0,0)), ((0,3),(3,2)) \right\} $  & $\left\{((0,0),(3,0)), ((0,3),(3,2))\right\}$ & $\mathcal{L}_{11}$ \\\hline
(27) &  $\left\{((1,0),(0,0)), ((1,0),(3,2)) \right\} $  & $\left\{((1,0),(0,0)), ((1,3),(0,2))\right\}$ & $\mathcal{L}_{2}$\\\hline
(28) &  $\left\{((1,2),(0,0)), ((1,2),(3,2)) \right\} $  & $\left\{((1,0),(2,0)), ((1,3),(2,2))\right\}$ & $\mathcal{L}_{10}$\\\hline
(29) &  $\left\{((2,1),(0,0)), ((2,1),(3,2)) \right\} $  & $\left\{((2,0),(1,0)), ((2,3),(1,2))\right\}$ & $\mathcal{L}_{3}$ \\\hline
(30) &  $\left\{((2,3),(0,0)), ((2,3),(3,2)) \right\} $  & $\left\{((2,0),(3,0)), ((2,3),(3,2))\right\}$ & $\mathcal{L}_{9}$\\\hline
(31) &  $\left\{((3,0),(0,0)), ((3,0),(3,2)) \right\} $  & $\left\{((3,0),(0,0)), ((3,3),(0,2))\right\}$ & $\mathcal{L}_{11}$ \\\hline
(32) &  $\left\{((3,1),(0,0)), ((3,1),(3,2)) \right\} $  & $\left\{((3,0),(1,0)), ((3,3),(1,2))\right\}$ & $\mathcal{L}_{12}$\\\hline
(33) &  $\left\{((0,1),(0,2)), ((3,3),(0,2)) \right\} $  & $\left\{((0,0),(1,2)), ((3,0),(3,2))\right\}$ & $\mathcal{L}_{10}$ \\\hline
(34) &  $\left\{((0,1),(0,3)), ((3,3),(0,3)) \right\} $  & $\left\{((0,0),(1,3)), ((3,0),(3,3))\right\}$ & $\mathcal{L}_{9}$ \\\hline
(35) &  $\left\{((0,1),(1,0)), ((3,3),(0,2)) \right\} $  & $\left\{((0,1),(1,0)), ((3,0),(3,2))\right\}$ & $\mathcal{L}_{10}$ \\\hline
(36) &  $\left\{((0,1),(1,2)), ((3,3),(1,2)) \right\} $  & $\left\{((0,1),(1,2)), ((3,1),(3,2))\right\}$ & $\mathcal{L}_{7}$ \\\hline
(37) &  $\left\{((0,1),(2,1)), ((3,3),(2,1)) \right\} $  & $\left\{((0,2),(1,1)), ((3,2),(3,1))\right\}$ & $\mathcal{L}_{10}$ \\\hline
(38) &  $\left\{((0,1),(2,3)), ((3,3),(2,3)) \right\} $  & $\left\{((0,2),(1,3)), ((3,2),(3,3))\right\}$ & $\mathcal{L}_{5}$  \\\hline
(39) &  $\left\{((0,1),(3,0)), ((3,3),(3,0)) \right\} $  & $\left\{((0,3),(1,0)), ((3,3),(3,0))\right\}$ & $\mathcal{L}_{13}$ \\\hline
(40) &  $\left\{((0,1),(3,1)), ((3,3),(3,1)) \right\} $  & $\left\{((0,3),(1,1)), ((3,3),(3,1))\right\}$ & $\mathcal{L}_{6}$ \\\hline
(41) &  $\left\{((0,2),(0,1)), ((0,2),(3,3)) \right\} $  & $\left\{((0,0),(2,1)), ((0,3),(2,3))\right\}$ & $\mathcal{L}_{12}$ \\\hline
(42) &  $\left\{((0,3),(0,1)), ((0,3),(3,3)) \right\} $  & $\left\{((0,0),(3,1)), ((0,3),(3,3))\right\}$ & $\mathcal{L}_{14}$ \\\hline
(43) &  $\left\{((1,0),(0,1)), ((1,0),(3,3)) \right\} $  & $\left\{((1,0),(0,1)), ((1,3),(0,3))\right\}$ & $\mathcal{L}_{1}$\\\hline
(44) &  $\left\{((1,2),(0,1)), ((1,2),(3,3)) \right\} $  & $\left\{((1,0),(2,1)), ((1,3),(2,3))\right\}$ & $\mathcal{L}_{11}$\\\hline
(45) &  $\left\{((2,1),(0,1)), ((2,1),(3,3)) \right\} $  & $\left\{((2,0),(1,1)), ((2,3),(1,3))\right\}$ & $\mathcal{L}_{4}$ \\\hline
(46) &  $\left\{((2,3),(0,1)), ((2,3),(3,3)) \right\} $  & $\left\{((2,0),(3,1)), ((2,3),(3,3))\right\}$ & $\mathcal{L}_{11}$\\\hline
(47) &  $\left\{((3,0),(0,1)), ((3,0),(3,3)) \right\} $  & $\left\{((3,0),(0,1)), ((3,3),(0,3))\right\}$ & $\mathcal{L}_{14}$ \\\hline
(48) &  $\left\{((3,1),(0,1)), ((3,1),(3,3)) \right\} $  & $\left\{((3,0),(1,1)), ((3,3),(1,3))\right\}$ & $\mathcal{L}_{15}$\\\hline
(49) &  $\left\{((1,1),(0,2)), ((2,0),(0,2)) \right\} $  & $\left\{((1,0),(1,2)), ((2,0),(0,2))\right\}$ & $\mathcal{L}_{13}$ \\\hline
(50) &  $\left\{((1,1),(0,3)), ((2,0),(0,3)) \right\} $  & $\left\{((1,0),(1,3)), ((2,0),(0,3))\right\}$ & $\mathcal{L}_{12}$ \\\hline
(51) &  $\left\{((1,1),(1,0)), ((2,0),(0,2)) \right\} $  & $\left\{((1,1),(1,0)), ((2,0),(0,2))\right\}$ & $\mathcal{L}_{14}$ \\\hline
(52) &  $\left\{((1,1),(1,2)), ((2,0),(1,2)) \right\} $  & $\left\{((1,1),(1,2)), ((2,1),(0,2))\right\}$ & $\mathcal{L}_{15}$ \\\hline
(53) &  $\left\{((1,1),(2,1)), ((2,0),(2,1)) \right\} $  & $\left\{((1,2),(1,1)), ((2,2),(0,1))\right\}$ & $\mathcal{L}_{9}$ \\\hline
(54) &  $\left\{((1,1),(2,3)), ((2,0),(2,3)) \right\} $  & $\left\{((1,2),(1,3)), ((2,2),(0,3))\right\}$ & $\mathcal{L}_{10}$  \\\hline
(55) &  $\left\{((1,1),(3,0)), ((2,0),(3,0)) \right\} $  & $\left\{((1,3),(1,0)), ((2,3),(0,0))\right\}$ & $\mathcal{L}_{15}$ \\\hline
(56) &  $\left\{((1,1),(3,1)), ((2,0),(3,1)) \right\} $  & $\left\{((1,3),(1,1)), ((2,3),(0,1))\right\}$ & $\mathcal{L}_{12}$ \\\hline
(57) &  $\left\{((0,2),(1,1)), ((0,2),(2,0)) \right\} $  & $\left\{((0,1),(2,1)), ((0,2),(2,0))\right\}$ & $\mathcal{L}_{13}$ \\\hline
(58) &  $\left\{((0,3),(1,1)), ((0,3),(2,0)) \right\} $  & $\left\{((0,1),(3,1)), ((0,2),(3,0))\right\}$ & $\mathcal{L}_{12}$ \\\hline
(59) &  $\left\{((1,0),(1,1)), ((1,0),(2,0)) \right\} $  & $\left\{((1,1),(0,1)), ((1,2),(0,0))\right\}$ & $\mathcal{L}_{4}$\\\hline
(60) &  $\left\{((1,2),(1,1)), ((1,2),(2,0)) \right\} $  & $\left\{((1,1),(2,1)), ((1,2),(2,0))\right\}$ & $\mathcal{L}_{16}$\\\hline
(61) &  $\left\{((2,1),(1,1)), ((2,1),(2,0)) \right\} $  & $\left\{((2,1),(1,1)), ((2,2),(1,0))\right\}$ & $\mathcal{L}_{2}$ \\\hline
(62) &  $\left\{((2,3),(1,1)), ((2,3),(2,0)) \right\} $  & $\left\{((2,1),(3,1)), ((2,2),(3,0))\right\}$ & $\mathcal{L}_{16}$\\\hline
(63) &  $\left\{((3,0),(1,1)), ((3,0),(2,0)) \right\} $  & $\left\{((3,1),(0,1)), ((3,2),(0,0))\right\}$ & $\mathcal{L}_{13}$ \\\hline
(64) &  $\left\{((3,1),(1,1)), ((3,1),(2,0)) \right\} $  & $\left\{((3,1),(1,1)), ((3,2),(1,0))\right\}$ & $\mathcal{L}_{11}$\\\hline
(65) &  $\left\{((1,3),(0,2)), ((2,2),(0,2)) \right\} $  & $\left\{((1,0),(3,2)), ((2,0),(2,2))\right\}$ & $\mathcal{L}_{17}$ \\\hline
(66) &  $\left\{((1,3),(0,3)), ((2,2),(0,3)) \right\} $  & $\left\{((1,0),(3,3)), ((2,0),(2,3))\right\}$ & $\mathcal{L}_{15}$ \\\hline
(67) &  $\left\{((1,3),(1,0)), ((2,2),(0,2)) \right\} $  & $\left\{((1,1),(3,0)), ((2,0),(2,2))\right\}$ & $\mathcal{L}_{11}$ \\\hline
(68) &  $\left\{((1,3),(1,2)), ((2,2),(1,2)) \right\} $  & $\left\{((1,1),(3,2)), ((2,1),(2,2))\right\}$ & $\mathcal{L}_{3}$ \\\hline
(69) &  $\left\{((1,3),(2,1)), ((2,2),(2,1)) \right\} $  & $\left\{((1,2),(3,1)), ((2,2),(2,1))\right\}$ & $\mathcal{L}_{15}$ \\\hline
(70) &  $\left\{((1,3),(2,3)), ((2,2),(2,3)) \right\} $  & $\left\{((1,2),(3,3)), ((2,2),(2,3))\right\}$ & $\mathcal{L}_{1}$  \\\hline
(71) &  $\left\{((1,3),(3,0)), ((2,2),(3,0)) \right\} $  & $\left\{((1,3),(3,0)), ((2,3),(2,0))\right\}$ & $\mathcal{L}_{14}$ \\\hline
(72) &  $\left\{((1,3),(3,1)), ((2,2),(3,1)) \right\} $  & $\left\{((1,3),(3,1)), ((2,3),(2,1))\right\}$ & $\mathcal{L}_{17}$ \\\hline
(73) &  $\left\{((0,2),(1,3)), ((0,2),(2,2)) \right\} $  & $\left\{((0,1),(2,3)), ((0,2),(2,2))\right\}$ & $\mathcal{L}_{14}$ \\\hline
(74) &  $\left\{((0,3),(1,3)), ((0,3),(2,2)) \right\} $  & $\left\{((0,1),(3,3)), ((0,2),(3,2))\right\}$ & $\mathcal{L}_{16}$ \\\hline
(75) &  $\left\{((1,0),(1,3)), ((1,0),(2,2)) \right\} $  & $\left\{((1,1),(0,3)), ((1,2),(0,2))\right\}$ & $\mathcal{L}_{17}$\\\hline
(76) &  $\left\{((1,2),(1,3)), ((1,2),(2,2)) \right\} $  & $\left\{((1,1),(2,3)), ((1,2),(2,2))\right\}$ & $\mathcal{L}_{13}$\\\hline
(77) &  $\left\{((2,1),(1,3)), ((2,1),(2,2)) \right\} $  & $\left\{((2,1),(1,3)), ((2,2),(1,2))\right\}$ & $\mathcal{L}_{18}$ \\\hline
(78) &  $\left\{((2,3),(1,3)), ((2,3),(2,2)) \right\} $  & $\left\{((2,1),(3,3)), ((2,2),(3,2))\right\}$ & $\mathcal{L}_{12}$\\\hline
(79) &  $\left\{((3,0),(1,3)), ((3,0),(2,2)) \right\} $  & $\left\{((3,1),(0,3)), ((3,2),(0,2))\right\}$ & $\mathcal{L}_{16}$ \\\hline
(80) &  $\left\{((3,1),(1,3)), ((3,1),(2,2)) \right\} $  & $\left\{((3,1),(1,3)), ((3,2),(1,2))\right\}$ & $\mathcal{L}_{14}$\\\hline
\end{tabular}
}
\label{case3}
\vspace{-.2 cm}
\caption{Singularity Removal Constraints Constraints for $\gamma e^{j \theta}=-1+j$}
\end{figure*}

\end{document}